\documentclass[11pt]{article}
\usepackage[utf8]{inputenc}
\usepackage[T1]{fontenc}
\usepackage[colorlinks=true]{hyperref}
\hypersetup{urlcolor=blue, citecolor=red, linkcolor=blue}

\usepackage{colortbl}

\usepackage{graphicx}
\usepackage[active]{srcltx}

\usepackage{amsmath,amssymb}
\usepackage{bm}
\usepackage{enumerate}
\usepackage{amsthm}
\usepackage[all,2cell]{xy}
\usepackage{mathrsfs}
\usepackage{mathtools}
\mathtoolsset{showonlyrefs}

\usepackage{tikz-cd}
\usetikzlibrary{cd}
\usepackage{wrapfig} 

\usepackage{xcolor}
\DeclareUnicodeCharacter{0301}{*************************************}

\hypersetup{urlcolor=blue, citecolor=red, linkcolor=blue}

\usepackage[a4paper, left=22mm, right=22mm, top=25mm, bottom=25mm]{geometry}

\usepackage{marginnote}

\usepackage{setspace}

\usepackage{imakeidx}
\makeindex[intoc]

\usepackage[nottoc]{tocbibind}
\usepackage[colorinlistoftodos]{todonotes} 

\theoremstyle{plain}
\newtheorem{theorem}{Theorem}[section]

\newtheorem{proposition}[theorem]{Proposition}

\newtheorem{lemma}[theorem]{Lemma}

\theoremstyle{definition}
\newtheorem{definition}[theorem]{Definition}

\newtheorem{example}[theorem]{Example}

\newcommand\restr[2]{{
  \left.\kern-\nulldelimiterspace 
  #1 
  \right|_{#2} 
}}

\newcommand{\R}{\mathbb{R}}

\newcommand{\Cinfty}{\mathscr{C}^\infty}
\newcommand{\T}{\mathrm{T}}

\newcommand{\cT}{\mathrm{T}^\ast}

\newcommand{\bomega}{{\boldsymbol{\omega}}}
\newcommand{\bh}{{\boldsymbol{h}}}
\newcommand{\Rp}{\mathfrak{Re}}
\renewcommand{\d}{\mathrm{d}}
\newcommand{\Ip}{\mathfrak{Im}}
\newcommand{\bfJ}{\mathbf{J}}

\newcommand*{\inn}[1]{\iota_{#1}}

\newcommand{\Lg}{\mathfrak{g}}
\newcommand{\x}{\times}

\newcommand{\Lie}{\mathscr{L}}

\newcommand{\X}{\mathfrak{X}}

\newcommand{\parder}[2]{\frac{\partial #1}{\partial #2}}

\DeclareMathOperator{\Ima}{Im}
\DeclareMathOperator{\Hess}{Hess}

\DeclareMathOperator{\Ad}{Ad}
\DeclareMathOperator{\Coad}{Coad}

\DeclareMathAlphabet{\mathpzc}{OT1}{pzc}{m}{it}
\def\d{\mathrm{d}}

\DeclareMathOperator{\pr}{pr}

\DeclareMathOperator{\SO}{SO}

\newcommand\xqed[1]{%
    \leavevmode\unskip\penalty9999 \hbox{}\nobreak\hfill
	\quad\hbox{#1}}
\newcommand\demo{\xqed{$\triangle$}}
\numberwithin{equation}{section}

\begin{document}

\parskip=3pt


\vspace{5em}

{\huge\sffamily\raggedright
\begin{spacing}{1.1}
    An energy-momentum method for ordinary differential equations with an underlying $k$-polysymplectic manifold
\end{spacing}
}
\vspace{2em}

{\large\raggedright
    \today
}

\vspace{3em}

{\Large\raggedright\sffamily
    Leonardo Colombo
}\vspace{1mm}\newline
{\raggedright
    Centro de Automática y Robótica (CSIC-UPM),\\ 
    Carretera de Campo Real, km 0, 200, 28500 Arganda del Rey, Spain.\\
    e-mail: \href{mailto:leonardo.colombo@car.upm-csic.es}{leonardo.colombo@car.upm-csic.es} --- orcid: \href{https://orcid.org/0000-0001-6493-6113}{0000-0001-6493-6113}
}

\bigskip

{\Large\raggedright\sffamily
    Javier de Lucas\footnote{Corresponding author}
}\vspace{1mm}\newline
{\raggedright
    Centre de Recherches Math\'ematiques, Universit\'e de Montr\'eal,\\
    Pavillon André--Aisenstadt, 2920, chemin de la Tour,  Montréal (Québec) Canada  H3T 1J4.\\
    \medskip
    
    Department of Mathematical Methods in Physics, University of Warsaw,\\ul. Pasteura 5, 02-093 Warszawa, Poland.\\
    e-mail: \href{mailto:javier.de.lucas@fuw.edu.pl}{javier.de.lucas@fuw.edu.pl} --- orcid: \href{https://orcid.org/0000-0001-8643-144X}{0000-0001-8643-144X}
}

\bigskip

{\Large\raggedright\sffamily
    Xavier Rivas
}\vspace{1mm}\newline
{\raggedright
    Escuela Superior de Ingeniería y Tecnología, Universidad Internacional de La Rioja,\\
    Av. de la Paz, 137, 26006 Logroño, La Rioja, Spain.\\
    e-mail: \href{mailto:xavier.rivas@unir.net}{xavier.rivas@unir.net} --- orcid: \href{https://orcid.org/0000-0002-4175-5157}{0000-0002-4175-5157}
}

\bigskip

{\Large\raggedright\sffamily
    Bartosz M. Zawora
}\vspace{1mm}\newline
{\raggedright
    Department of Mathematical Methods in Physics, University of Warsaw,\\
    ul. Pasteura 5, 02-093 Warszawa, Poland.\\
    e-mail: \href{mailto:b.zawora@uw.edu.pl}{b.zawora@uw.edu.pl} --- orcid: \href{https://orcid.org/0000-0003-4160-1411}{0000-0003-4160-1411}
}

\vspace{3em}

{\large\bf\raggedright
    Abstract
}\vspace{1mm}\newline
{\raggedright
    This work presents a comprehensive review of the $k$-polysymplectic Marsden--Weinstein reduction theory, rectifying prior errors and inaccuracies in the literature while introducing novel findings. It also emphasises the genuine practical significance of seemingly minor technical details. On this basis, we introduce a novel $k$-polysymplectic energy-momentum method, new related stability analysis techniques, and apply them to Hamiltonian systems of ordinary differential equations relative to a $k$-polysymplectic manifold. We provide detailed examples of both physical and mathematical significance, including the study of complex Schwarz equations related to the Schwarz derivative, a series of isotropic oscillators, integrable Hamiltonian systems,  quantum oscillators with dissipation, affine systems of differential equations, and polynomial dynamical systems.
}
\bigskip

\medskip


\bigskip


\newpage

{\setcounter{tocdepth}{2}
\def\baselinestretch{1}
\small
\def\addvspace#1{\vskip 1pt}
\parskip 0pt plus 0.1mm
\tableofcontents
}

\pagestyle{myheadings}

\markright{{\rm
    L. Colombo {\it et al.}
}
\ ------ \ \ 
{\sl
    A $k$-polysymplectic energy-momentum method
}}

\section{Introduction}

The classical energy-momentum method is a technique for analysing a Hamiltonian system on a symplectic manifold, particularly in the region near solutions whose evolution is induced by the Lie symmetries of the Hamiltonian system (see \cite{Blo_15} for a historical introduction and \cite{MS_88} for one of its foundational works). More specifically, it explores whether, over time, solutions converge towards or diverge from the solutions associated with the Lie symmetries of the Hamiltonian system. The classical energy-momentum method is grounded in the symplectic Marsden--Weinstein reduction theory and utilises stability analysis techniques. 

The main ideas behind the energy-momentum method can be traced back to Routh, Poincar\'e, Lyapunov, Arnold, Lewis, and Smale, among others (see \cite[Section 3.14]{Blo_15}). Then,
the classical energy-momentum method, devised and developed mainly by J. C. Simo and J. E. Marsden \cite{MS_88}, was successfully applied to many problems by numerous researchers \cite{AH_87, MPS_90, MR_99, MSLP_89, OPR_05, SLM_91, ZBM_98}. Over the years, the energy-momentum method was extended to deal with more general differential equations, e.g. stochastic Hamiltonian systems \cite{BZ_14}, discrete systems \cite{MR_99,ST_92}, etcetera  \cite{WK_92}. In this work, we develop a new energy-momentum method for Hamiltonian systems related to $k$-polysymplectic manifolds \cite{Awa_92,LMS_88}. 

 {\it $k$-Symplectic geometry} is a generalisation of symplectic geometry introduced by A. Awane \cite{Awa_92,AG_00}. Posteriorly, M. de León et al. \cite{LMOS_97,LMS_88,LMS_88a} and L. K. Norris \cite{MN_00,Nor_93} utilised $k$-symplectic geometry  to describe first-order field theories \cite{BBLSV_15,EMR_96,RSV_07}. $k$-Symplectic geometry is the same as the {\it polysymplectic geometry} described by G. C. Günther \cite{Gun_87}, but differs from the polysymplectic geometry introduced by G. Sardanashvily et al. \cite{GMS_97,Sar_95} and I. V. Kanatchikov \cite{Kan_98}. $k$-Symplectic manifolds have been widely used to study physical systems governed by systems of partial differential equations. In particular, it gives a geometric description of the Euler--Lagrange and the Hamilton--de Donder--Weyl field equations 
and the systems described by them. For instance, $k$-symplectic geometry enables us to describe their symmetries, conservation laws, reductions, etcetera \cite{Awa_92, Gun_87, MRSV_15, RSV_07}. As there are many $k$-symplectic-like definitions with related but mainly different and even contradictory meanings, it is relevant to fix properly the terminology. Hereafter, we will deal with {\it $k$-polysymplectic manifolds}, i.e. manifolds endowed with a closed nondegenerate differential two-form taking values in a $k$-dimensional vector space.

 Remarkably, $k$-polysymplectic geometry has proved to be useful in the analysis of systems of ordinary differential equations and their so-called superposition rules \cite{LV_15}. It is also worth stressing that the study of systems of ordinary differential equations via $k$-polysymplectic geometry differs substantially from the standard framework, which is focused on systems of partial differential equations, and leads to new lines of research.

More specifically, this work focuses on studying systems of first-order differential equations describing the integral curves of a vector field. Moreover, we assume that the vector field is Hamiltonian relative to a $k$-polysymplectic manifold, which here amounts to the fact that it is Hamiltonian relative to a series of presymplectic forms whose kernels have zero intersection. We aim to develop an energy-momentum method for such systems of ordinary differential equations with an underlying $k$-polysymplectic geometry. To achieve this goal, we will begin by reviewing and improving previous works on $k$-polysymplectic Marsden--Weinstein reductions \cite{Bla_19,LRVZ_23,GM_23,MRSV_15,MRS_04}, which is one of the basis of our $k$-polysymplectic energy-momentum method. Hopefully, our review will solve previous problems and inaccuracies in the $k$-polysymplectic reduction literature, and will allow us to understand the meaning of some of the findings of this work.  

The first $k$-polysymplectic reduction was developed by G\"unther \cite{Awa_92,LSV_15,Gun_87}. Unfortunately, his work was flawed due to the improper analysis of the double orthogonal relative to a $k$-polysymplectic form. More specifically, \cite[Lemma 7.5 and Theorem 7.7]{Gun_87} contain main G\"unther's mistakes, while \cite[Section 2.2]{MRSV_15} provides an interesting counterexample showing G\"unther's error\footnote{There is a typo in \cite[pg. 4]{MRSV_15} as its authors refer to Theorem 7.6 in \cite{Gun_87}, which is not a theorem, but a definition. Flawed G\"unther's reduction theorem is described in Theorem 7.7.}. Another similarly flawed attempt to develop a $k$-polysymplectic reduction was accomplished in \cite{MRS_04}. These mistakes were fixed in \cite{MRSV_15}, where sufficient conditions to accomplish a $k$-polysymplectic reduction were established. Despite that, \cite[Lemma 3.4]{MRSV_15} implicitly suggests that the Sard's Theorem justifies that it is enough to assume that the $k$-polysymplectic momentum map is a submersion. Although this assumption works very well in the classical symplectic Marsden--Weinstein reduction theory and Sard's Theorem can be used to justify it \cite{BR_04},  
our work proves that this condition is far from ideal in the $k$-polysymplectic geometry realm and why the Sard's Theorem cannot be used in this new context. Moreover, practical examples showing that it is convenient to assume that the momentum map in $k$-polysymplectic geometry is not a submersion are provided. Then, we stress that it is appropriate to use a formalism with $k$-polysymplectic momentum maps that admit only weak regular points, as accomplished in \cite{LRVZ_23}. We highlight that this provides a practical generalisation of the $k$-polysymplectic Marsden--Weinstein reduction and it completes the analysis performed in \cite{Bla_19,GM_23,Gun_87,MRSV_15}.

Necessary and sufficient conditions for a $k$-polysymplectic Marsden--Weinstein reduction were described implicitly in \cite[pg. 12]{MRSV_15} and spelled out in detail by Blacker in \cite{Bla_19}. Unfortunately, one of Blacker's main theorems, namely \cite[Theorem 3.22]{Bla_19}, has a small misleading typo in the statement of the conditions (as pointed out in \cite{GM_23}), in its proof, and it presents other minor technical issues concerning the existence of certain submanifold structures. These latter facts are shown and explained in this work for the first time. It is also worth noting that Blacker analyses the occurrence of orbifolds in $k$-polysymplectic Marsden--Weinstein reductions for regular values of momentum maps related to pathological Lie group actions.  

The need for the use of $\Ad^{*k}$-equivariant momentum maps in the $k$-polysymplectic Marsden--Weinstein reductions was removed in \cite{LRVZ_23} by extending to the $k$-polysymplectic realm the classical theory of affine Lie group actions on symplectic manifolds \cite{OR_04}. Next, Garc\'ia-Tora\~no and Mestdag reviewed in \cite{GM_23} the sufficient conditions for the $k$-polysymplectic Marsden--Weinstein reduction devised in \cite{MRSV_15}. They claimed that just one of the sufficient conditions for the $k$-polysymplectic reduction given in \cite[Theorem 3.17, condition (3.6)]{MRSV_15} is enough to ensure the existence of a $k$-polysymplectic Marsden--Weinstein reduction. In this work, we show a mistake in the proof of one of the main results in \cite{GM_23}, used to justify the previous claim. Indeed, we here point out that \cite[Lemma 3.1]{GM_23} is false via a counterexample, and prove the general independence of the conditions in \cite[Theorem 3.17]{MRSV_15}. Moreover, our work also explains other properties relative to such sufficient conditions.  

In order to illustrate a relevant example of $k$-polysymplectic Marsden--Weinstein reduction, we review the construction of a $k$-polysymplectic manifold induced by $k$ symplectic manifolds and a related $k$-polysymplectic Marsden--Weinstein reduction. It is worth noting that, in this case, and in our applications in 
Section \ref{Sec::examples}, the sufficient conditions for the $k$-polysymplectic Marsden--Weinstein reduction given in \cite{MRSV_15} are generally simpler to apply than  Blacker's necessary and sufficient conditions, as the conditions in \cite{MRSV_15} do not depend on double $k$-polysymplectic orthogonal spaces and can be verified using structures easily available in our examples.

Next, an energy-momentum method for Hamiltonian $k$-polysymplectic systems is developed. This entails the definition and characterisation of a relative equilibrium notion for $k$-polysymplectic Hamiltonian systems. In short, a relative equilibrium point for a $k$-polysymplectic Hamiltonian 
 system is a point at which the dynamics is determined by a Hamiltonian Lie symmetry of the $k$-polysymplectic Hamiltonian system.
Our $k$-polysymplectic energy-momentum method also requires the development of an appropriate modification of known symplectic stability techniques to a $k$-polysymplectic realm. In particular, the stability of relative equilibrium points for $k$-polysymplectic Hamiltonian systems is characterised by analysing the character of $k$ different functions having, mainly, degenerate critical points, namely their Hessians are degenerate at critical points. As in the symplectic case, a formal stability giving sufficient but not necessary conditions for the stability are given. The interest in our formal stability condition is justified by our applications. Although the formal stability condition is easy to verify and can be used in many cases, it is worth noting that proving its properties is quite more difficult than in the symplectic case. Moreover, we here just sketch that it is possible to develop many other alternative sufficient conditions to ensure stability.

Then, some applications of our $k$-polysymplectic energy-momentum method are developed. In particular, the theory of Lie systems is used to transform certain automorphic Lie systems \cite{CL_11,CGM_00,CGM_07,Win_83} into $k$-polysymplectic Hamiltonian systems. A Lie system is a non-autonomous system of first-order differential equations whose general solution can be written as an autonomous function, a {\it superposition rule}, of a generic family of particular solutions and some constants. Lie systems are very important due to their applications and mathematical properties \cite{CL_11,LS_20}.   Automorphic Lie systems are Lie systems in Lie groups of special relevance, in particular, in control theory \cite{CR_03}. 
A $k$-polysymplectic manifold is used to study complex Schwarz equations, which are here studied through the theory of Lie systems and $k$-polysymplectic geometry for the first time (see \cite{LS_20} for the analysis of the real, simpler, case). It is worth noting that the complex Schwarz equation provides the description, when written as a first-order system of differential equations, of certain properties of the Schwarz derivative, which has applications in string theory, modular forms, hypergeometric functions \cite{GR_07,Hil_97,Leh_79}, and other related equations \cite{BC_20}. Automorphic Lie systems related to quantum oscillators with dissipative terms are also studied via $k$-polysymplectic techniques. 
We develop methods to study certain dynamical systems via Hamiltonian $k$-polysymplectic systems. This is applied to a family of $k$ particles in a three-dimensional space, that are under the effect of different isotropic potentials and have no interaction between them. In this case, the techniques of our $k$-polysymplectic energy-momentum method are illustrated. Furthermore, a particular type of affine Lie system is used to show certain aspects of our $k$-polysymplectic energy-momentum method. Potentially, the ideas used in this latter example could be used to study affine control systems of a similar type \cite{CR_03}. Other examples related to differential equations with polynomial coefficients are presented and analysed.

The structure of the paper goes as follows. Section \ref{Se::Fun} presents the basic notions and terminology to be used in our work. More particularly, Section \ref{Sec::Stab} provides a review of the fundamentals of Lyapunov stability.
In Section \ref{Sec::moment-maps}, we delve into the theory of $k$-polysymplectic manifolds, introducing the concept of an $\bomega$-Hamiltonian vector field and function on such a manifold in Section \ref{OHam}, and studying $k$-polysymplectic momentum maps in Section \ref{Sec::MomentuMap}. 
Section \ref{Sec::reduction} is dedicated to enhancing the existing Marsden--Weinstein reduction procedures for $k$-polysymplectic manifolds and presenting a $k$-polysymplectic Marsden--Weinstein reduction of the dynamics governed by an $\bomega$-Hamiltonian vector field. Note that this implies that some previous results, like \cite[Theorem 4.4]{MRSV_15}, are here slightly modified to analyse more efficiently systems of ordinary differential equations. Relevantly, this section surveys and corrects many inaccuracies and mistakes in the previous literature.
Section \ref{Sec::energy-momentum} introduces an energy-momentum method for systems of ordinary differential equations (ODEs) with an underlying $k$-polysymplectic structure. We define and characterise the concept of a relative equilibrium point for such systems. A theory of stability for the analysis of relative equilibrium points for $k$-polysymplectic Hamiltonian systems is presented. 
In Section \ref{Sec::examples}, we thoroughly examine several relevant examples, including the complex Schwarz equation, the product of multiple symplectic manifolds along with a related family of isotropic oscillators, an affine first-order system of differential equations related to Lie systems and, potentially, to control systems, and  quantum harmonic oscillators with dissipative terms.
Finally, Section \ref{Sec::Conclusions} summarises the conclusions of our work and offers insights into potential avenues for further development.

\section{Fundamentals}\label{Se::Fun}

Let us set some general assumptions and notation to be used throughout this work. It is hereafter assumed that all structures are smooth. Manifolds are real, Hausdorff, connected, paracompact, and finite-dimensional. Differential forms are assumed to have constant rank unless otherwise stated. Summation over crossed repeated indices is understood, although it can be explicitly detailed at times to improve the clarity of our presentation. 
All our considerations are local to stress our main ideas and to avoid technical problems concerning the global manifold structure of quotient spaces and similar issues. Hereafter, $\mathfrak{X}(P)$ and $\Omega^k(P)$ stand for the $\Cinfty(P)$-modules of vector fields and differential $k$-forms on a manifold $P$.

\subsection{Lyapunov stability}\label{Sec::Stab}

Let us establish some fundamental notions and theorems on the stability of dynamical systems used in our $k$-polysymplectic formulation of the energy-momentum method \cite{LZ_21, Zaw_21}.

Since all manifolds considered in this work are paracompact and Hausdorff, they admit a Riemannian metric $\mathbf{g}$ \cite{Lee_09}. The topology induced by $\mathbf{g}$ is the one of the manifold \cite{Lee_09,Lee_12, Zaw_21}. The metric $\mathbf{g}$ induces a distance in $P$ so that the distance between two points $x_1,x_2\in P$ is given by 
\begin{equation}
    d_{\mathbf{g}}(x_1,x_2):=\inf\left\{\ell_{\mathbf{g}}(\gamma) \ \mid\ \gamma:[0,1]\rightarrow P\,,\ \ \gamma(0)=x_1\,,\ \ \gamma(1)=x_2\right\}\,,
\end{equation}
where $\ell_{\mathbf{g}}(\gamma)$ is the length of the smooth curve $\gamma:[0,1]\rightarrow  P$ relative to the metric $\mathbf{g}$. Moreover, consider 
\begin{equation}\label{Eq::NonAutDyn}
     \frac{\d x}{\d t} = X(x)\,,\qquad \forall x\in P\,,
\end{equation}
where $X$ is a vector field on $P$.

A point $x_e\in P$ is an {\it equilibrium point} of \eqref{Eq::NonAutDyn}, or indistinctly $X$, if $X(x_e)=0$. Furthermore, $x_e$ is {\it stable} if, for every ball $B_{x_e,\varepsilon}:=\{x\in P\mid d_{\mathbf{g}}(x,x_e)<\epsilon\}$, there exists a radius $\delta(\varepsilon,x_e)$ such that every solution $x(t)$ of \eqref{Eq::NonAutDyn} with initial condition $x(t_0)=x_0\in B_{x_e,\delta(\varepsilon,x_e)}$ for some $t_0\in \mathbb{R}$ is contained in $B_{x_e,\varepsilon}$ for $t>t_0$. An equilibrium point $x_e\in P$ is {\it unstable} if it is not stable. 

The fact that the topology of a manifold is the same as the topology induced for any metric on it allows one to show that every $d_{\mathbf{g}}$, independently of the associated $\mathbf{g}$, induces the same stable and unstable points for \eqref{Eq::NonAutDyn}.

Lyapunov theory studies the stability of equilibrium points of first-order differential equations. Let $\dot{\mathcal{M}}:P\rightarrow \mathbb{R}$ be defined as follows
\[
\dot{\mathcal{M}}(x):=(X\mathcal{M})(x)\,,
\qquad \forall x\in P\,.
\]

Let us recall the basic Lyapunov theorem for autonomous systems \eqref{Eq::NonAutDyn}.

\begin{theorem}
\label{Th::LyapunovStabilityTheory}
Let $x_e$ be an equilibrium point of \eqref{Eq::NonAutDyn} and let $\mathcal{M}:P\rightarrow \mathbb{R}$ be a continuous function such that $\mathcal{M}(x_e)=0$, $\mathcal{M}(x)>0$, and $\dot{\mathcal{M}}(x)\leq 0$ for every $x\in B_{x_e,r}$ and some $r\in \mathbb{R}^+$. Then, $x_e$ is stable.
\end{theorem}
In the literature, the function $\mathcal{M}$ is called a {\it Lyapunov function} \cite{Vid_02}.

\subsection{On \texorpdfstring{$k$}{}-polysymplectic manifolds}\label{Sec::moment-maps}

This section recalls the basic notions in $k$-polysymplectic geometry to be used later on. This is relevant as a single term may refer to different not equivalent geometric concepts in the literature. 

Hereafter, we work with differential $\ell$-forms on $P$ that take values in $\R^k$. The space of such forms is denoted by $\Omega^\ell(P,\R^k)$, while its elements will be written in bold. Moreover, $\R^k$ has a fixed basis $\{e_1,\ldots,e_k\}$ giving rise to a dual basis $\{e^1,\ldots, e^k\}$ in $\R^{k*}$. Hence, an element $\bomega\in\Omega^\ell(P,\R^k)$ can always be written as $\bomega = \omega^\alpha\otimes e_\alpha$ for some uniquely defined differential $\ell$-forms $\omega^1,\ldots,\omega^k$ on $P$. A differential $\ell$-form on $P$ taking values in $\mathbb{R}^k$, let us say $\bomega$, is nondegenerate if  
$$ 
\ker \boldsymbol{\omega}=\ker (\omega^\alpha\otimes e_\alpha):=\bigcap_{\alpha = 1}^k \ker\omega^\alpha = 0\,. 
$$

Let us introduce the following definition that will be useful to simplify the notation of our further work. Let $\bm\vartheta = \vartheta^\alpha\otimes e_\alpha\in\Omega^\ell(P,\R^k)$ be an $\R^k$-valued differential $\ell$-form on $P$. Then, the contraction of $\bm \vartheta$ with a vector field $X\in\mathfrak{X}(P)$ is defined as
\[
\iota_X \bm\vartheta := (\iota_X\vartheta^\alpha)\otimes e_\alpha=:\langle \bm\vartheta,X\rangle\in\Omega^{\ell-1}(P,\R^k)\,.
\]
In short, the exterior differential, the Lie derivative with respect to vector fields, and many other operations on differential forms can naturally be extended to $\ell$-differential forms taking values in vector spaces by considering the natural action of the above-mentioned operations on the components of valued differential forms and extending them to $\Omega^{\ell}(P,\mathbb{R}^k)$ by linearity.

A $k$-vector field on a manifold $P$ is, essentially, a family of $k$ vector fields on $P$. We write $\mathfrak{X}(P,\mathbb{R}^k)$ for the space of $k$-vector fields on $P$ and its elements will be written in bold. Moreover, a $k$-vector field, let us say ${\boldsymbol{X}}$, can always be written in a unique manner as ${\boldsymbol{X
}}=X_\alpha\otimes e_\alpha$ for a family $X_1,\ldots, X_k$ of vector fields on $P$. The contraction of a $k$-vector field ${\boldsymbol{X}}=X_\alpha\otimes e_\alpha$ with a $k$-differential form $\boldsymbol{\omega}=\omega^\alpha\otimes e_\alpha$ is the function on $P$ defined as follows
$$
\iota_{\boldsymbol{X}}\boldsymbol{\omega}:=\iota_{X_\alpha}\omega^\alpha=:\langle \boldsymbol{\omega},{\boldsymbol{X}}\rangle\,.
$$

Now, let us turn to one of the main fundamental notions to be studied in this paper.

\begin{definition}
A {\it $k$-polysymplectic form} on $P$ is a closed nondegenerate $\mathbb{R}^k$-valued differential two-form $\boldsymbol{\omega}$ on $P$. The pair $(P,\boldsymbol{\omega})$ is called a {\it $k$-polysymplectic manifold}.
\end{definition}
Consider a $k$-polysymplectic manifold $(P,\bm\omega)$, and let $W_p\subset\T_pP$ at some $p\in P$. The {\it k-polysymplectic orthogonal complement} of $W_p$ with respect to $(P,\bm\omega)$ is 
\[
W_p^{\perp,k} := \{v_p\in \T_p P\mid \bm\omega(w_p,v_p)=0\,,\,\, \forall w_p\in W_p\}.
\]
$k$-Polysymplectic manifolds are called, for simplicity, polysymplectic manifolds in the literature \cite{MRSV_15}. Nevertheless, the latter term may be misleading as refers here to a different concept shown below. Hence, to avoid confusion, we will use the full term \textit{$k$-polysymplectic manifold}. Let us define polysymplectic manifolds, $k$-polysymplectic manifolds, and related notions. 

\begin{definition}\label{Def::CoNotions}
    Let $P$ be an $n(k+1)$-dimensional manifold. Then,
    \begin{itemize}
        \item A {\it polysymplectic form} on $P$ is a nondegenerate differential two-form, $\bomega$, taking values in $\mathbb{R}^k$. 
        We call $(P,\boldsymbol{\omega})$ a {\it polysymplectic manifold}.
        \item A {\it $k$-symplectic structure} on $P$ is a pair $(\boldsymbol{\omega},\mathcal{D})$, where $(P,\boldsymbol{\omega})$ is a polysymplectic manifold and $\mathcal{D}\subset \T P$ is an integrable distribution on $P$ of rank $nk$ such that
        \[
        \restr{\boldsymbol{\omega}}{\mathcal{D}\times \mathcal{D}} = 0\,.
        \]
        In this case, $(P,\boldsymbol{\omega},\mathcal{D})$ is a {\it $k$-symplectic manifold}. We call $\mathcal{D}$ a {\it polarisation} of $(P,\boldsymbol{\omega})$.
    \end{itemize}
    If the two-form $\boldsymbol{\omega}$ is exact, namely $\boldsymbol{\omega} = \d\boldsymbol{\theta}$ for some $\boldsymbol{\theta} \in\Omega^1(P,\mathbb{R}^k)$, in any of the notions in Definition \ref{Def::CoNotions}, then such concepts are said to be {\it exact}.
\end{definition}

Note that the difference between polysymplectic and $k$-polysymplectic manifolds relies on the fact that in the polysymplectic case, the dimension of the manifold is proportional to $k+1$ if the polysymplectic form takes values in $\mathbb{R}^k$.

\subsection{On \texorpdfstring{$\bomega$}{}-Hamiltonian functions and vector fields}\label{OHam}

Let us survey the basic theory on $k$-polysymplectic vector fields and functions. Recall that we will not be concerned with the local or global character of the structures to be defined next.

\begin{definition}
Given a $k$-polysymplectic manifold $(P,\bomega=\omega^\alpha\otimes e_\alpha)$, a vector field $Y\in\X(P)$ is {\it $\boldsymbol{\omega}$-Hamiltonian} if it is Hamiltonian with respect to all the presymplectic forms $\omega^1,\dotsc,\omega^k$, namely $\inn{Y}\omega^\alpha$ is closed for $\alpha = 1,\dotsc,k$. Let us denote by $\X_\bomega(P)$ the space of $\bomega$-Hamiltonian vector fields in a $k$-polysymplectic manifold $(P,\bomega)$.

\end{definition}

Note that if $\iota_Y\omega^\alpha$ is closed, then it generally  admits a potential function only locally. Anyhow, this work is mainly concerned with local aspects and the fact that the potential function may not be globally defined will not have any repercussions in what follows. 

It is convenient for the study of $\bomega$-Hamiltonian vector fields to introduce some generalisation of the Hamiltonian function notion for presymplectic forms to deal simultaneously with all associated $h^1,\ldots,h^k$ (see \cite{Awa_92, LV_15} for details). 

\begin{definition}
    Given a $k$-polysymplectic manifold $(P,\bomega=\omega^\alpha\otimes e_\alpha)$, we say that $\bh = h^\alpha\otimes e_\alpha$ is an \textit{$\bomega$-Hamiltonian function} if there exists a vector field $X_\bh$ on $P$ such that  $\inn{X_\bh}\bomega=
    \d \bh$, namely $\inn{X_\bh}\omega^\alpha=\d h^\alpha$ for $\alpha = 1,\dotsc,k$. In this case, we call $\bh$ an \textit{$\boldsymbol{\omega}$-Hamiltonian function} for $X_\bh$. We write $\Cinfty_{\boldsymbol{\omega}}(P)$ for the space of $\boldsymbol{\omega}$-Hamiltonian functions of $(P,\bomega)$.
\end{definition}

An ${\bomega}$-Hamiltonian vector field (resp. function) will be simply called $k$-Hamiltonian at times, if $\bomega$ is understood from context or its specific expression is not relevant. 
In \cite{Mer_97}, the author defined the $k$-Hamiltonian system associated with the $\R^k$-valued Hamiltonian function $\bh$ as the vector field $X_\bh$ of the above definition. Moreover, A. Awane \cite{Awa_92} called $\bh$ a Hamiltonian map of $X$ when $X$ is additionally an infinitesimal automorphism of a certain distribution on which it is assumed that the presymplectic forms of the $k$-symplectic distribution vanish.

\begin{example}\label{ex1}
    Consider the two-polysymplectic manifold $(\R^3,\bomega)$, where $\{u,v,w\}$ are linear coordinates on $\R^3$ and $\bomega = \omega^1\otimes e_1 + \omega^2\otimes e_2$, where     (see \cite{LV_15} for details)
    $$
        \omega^1 =- \frac{4w}{v^2}\d u\wedge\d w + \frac{1}{v}\d v\wedge\d w + \frac{4w^2}{v^3}\d u\wedge\d v\,,\qquad \omega^2 = -\frac{4}{v^2}\d u\wedge\d w + \frac{8w}{v^3}\d u\wedge\d v\,,
    $$
    is a two-polysymplectic form. The vector fields
    $$
        X_1 = 4u^2\parder{}{u} + 4uv\parder{}{v} + v^2\parder{}{w}\,,\qquad X_2 = \parder{}{u}\,,
    $$
    are $\bomega$-Hamiltonian with  $\bomega$-Hamiltonian functions
    \begin{gather}
        \boldsymbol{f} = \Big(4uw-8\displaystyle\frac{u^2w^2}{v^2}-\displaystyle\frac{v^2}{2}\Big)\otimes e_1 + \Big(4u-16\displaystyle\frac{u^2w}{v^2}\Big) \otimes e_2\,,\quad
        \boldsymbol{g} = -2\displaystyle\frac{w^2}{v^2}\otimes e_1-4\frac{w}{v^2}\otimes e_2\,,
    \end{gather}
    respectively, relative to the two-polysymplectic form $\bomega$.\demo
\end{example}

   Every $\bomega$-Hamiltonian vector field is associated with at least one $\bomega$-Hamiltonian function. Conversely, every $\bomega$-Hamiltonian function induces a unique $\bomega$-Hamiltonian vector field.

\begin{proposition} The space $\Cinfty_\bomega(P)$ relative to $k$-polysymplectic manifold $(P,\bomega)$  becomes a Lie algebra when endowed with the natural operations
\begin{equation*}
\boldsymbol{h} + \boldsymbol{g}:=(h^\alpha + g^\alpha)\otimes e_\alpha\,,\qquad \lambda \cdot \bh:= \lambda h^\alpha\otimes e_\alpha\,,
\end{equation*}
where $\bh =  h^\alpha\otimes e_\alpha$, $\boldsymbol{g} = g^\alpha\otimes e_\alpha\in \Cinfty_\bomega(P)$, $\lambda\in\R$, and the Lie bracket $\{\cdot,\cdot\}_\bomega:\Cinfty_\bomega(P)\times \Cinfty_\bomega(P)\rightarrow \Cinfty_\bomega(P)$ of the form
\begin{equation}\label{LieB}
    \{\bh,\boldsymbol{g}\}_\bomega=\{h^1,g^1\}_{\omega^1}\otimes e_1 + \dotsb + \{h^k,g^k\}_{\omega^k}\otimes e_k\,,
\end{equation}
where $\{\cdot,\cdot\}_{\omega^\alpha}$ is the Poisson bracket naturally induced by the presymplectic form $\omega^\alpha$, with $\alpha=1,\ldots,k$.
\end{proposition}

The product of $\bomega$-Hamiltonian functions  
\[ 
\bh\star\boldsymbol{g} = (h^1g^1)\otimes e_1 + \dotsb + (h^kg^k)\otimes e_k\,,
\]
is not in general an $\bomega$-Hamiltonian function \cite[pg. 2239]{LV_15}. Hence,
  $(\Cinfty_\bomega(P),\star,\{\cdot,\cdot\}_\bomega)$ is not in general a Poisson algebra \cite[pg. 2239]{LV_15}. Moreover, the map $\{\bh,\cdot\}_\bomega:\boldsymbol{g}\in \Cinfty_\bomega(P)\mapsto \{\boldsymbol{g},\bh\}_\bomega\in \Cinfty_\bomega(P)$, with $\bh\in \Cinfty_\bomega(P)$, is not, in general, a derivation with respect to $\star$ neither. Hence, $k$-polysymplectic geometry is quite different from Poisson and presymplectic geometry. Nevertheless, $\{\bh,\boldsymbol{g}\}_{\bomega}=0$ for every locally constant function $\boldsymbol{g}\in\Cinfty_\bomega(P)$ and any $\boldsymbol{h}\in \Cinfty_\bomega(P)$. This Lie algebra admits other properties, as shown next.

\begin{proposition}
\label{Prop::Hamkfun}Consider a $k$-polysymplectic manifold $(P,\bomega)$. Every $\bomega$-Hamiltonian vector field $X_{\bh}$ acts as a derivation on the Lie algebra $(\Cinfty_\bomega(P),\{\cdot,\cdot\}_\bomega)$ in the form
$$
    X_{\bh}\boldsymbol{f} = \{\boldsymbol{f},\bh\}_\bomega\,,\qquad \forall \boldsymbol{f}\in \Cinfty_\bomega(P)\,,
$$
where $\bh$ is an $\bomega$-Hamiltonian function for $X_{\bh}$.
\end{proposition}




\subsection{\texorpdfstring{$k$}{}-Polysymplectic momentum maps}\label{Sec::MomentuMap}

Let us survey the theory of $k$-polysymplectic momentum maps. Note that the presented results are not restricted to $\Ad^{*k}$-equivariant momentum maps (see \cite{LRVZ_23} for further details).

\begin{definition}
        A Lie group action $\Phi\colon G\times P\to P$ on a $k$-polysymplectic manifold $(P,\boldsymbol\omega)$ is  a \textit{$k$-polysymplectic Lie group action} if $\Phi_g^*\boldsymbol\omega=\boldsymbol\omega$ for each $g\in G$. In other words,
\begin{equation}
    \Lie_{\xi_P}\bomega = 0\,,\qquad \forall \xi\in\mathfrak{g}\,,
\end{equation}
where $\xi_P$ is the fundamental vector field of $\Phi$ related to $\xi\in\mathfrak{g}$, namely $\xi_P(p) = \restr{\frac{\d}{\d t}}{t=0}\Phi(\exp(t\xi),p)$ for any $p\in P$.
\end{definition} 

\begin{definition}
\label{Def::PolysymMomentumMap}
A {\it $k$-polysymplectic momentum map} for a Lie group action $\Phi: G\times P\rightarrow P$ with respect to a $k$-polysymplectic manifold $(P,\boldsymbol\omega)$ is a mapping $\mathbf{J}^\Phi:P\rightarrow (\mathfrak{g}^*)^k$ such that
\begin{equation}
\label{Eq::PolycoMomentumMap}
\inn{\xi_P}\boldsymbol\omega=(\inn{\xi_P}\omega^\alpha)\otimes e_\alpha=\d\left\langle \mathbf{J}^\Phi,\xi\right\rangle \,,\qquad \forall \xi\in \mathfrak{g}\,.
\end{equation}
\end{definition}
Equation \eqref{Eq::PolycoMomentumMap} implies that $\mathbf{J}^\Phi:P\rightarrow(\Lg^*)^k$ satisfies 
\begin{equation}
\label{Eq:EqCond}
\inn{\boldsymbol{\xi}_P}\boldsymbol{\omega}=\d\left\langle \mathbf{J}^\Phi,\boldsymbol{\xi}\right\rangle \,,\qquad\forall\boldsymbol{\xi}\in\Lg^k\,.
\end{equation}
and conversely. For simplicity, we will write $\langle {\bf J}^\Phi,\bm\xi\rangle=:{\bf J}^\Phi_{\bm\xi}$.

Before continuing studying $k$-polysymplectic momentum maps, recall that every Lie group $G$ gives rise to Lie group action $I:(g,h)\in G\times G\mapsto I_{g}(h)=ghg^{-1}\in G$, such that $I_g:h\in G\mapsto I(g,h)\in G$ for every $g\in G$. Then, the {\it adjoint action} of $G$ on its Lie algebra, $\mathfrak{g}$, reads $\Ad:(g,v)\in G\times\mathfrak{g}\mapsto \Ad_g(v)=\T_eI_g(v)\in  \mathfrak{g}$. In turn, the {\it co-adjoint action} becomes $\Ad^*:(g,\vartheta)\in G\times \mathfrak{g}^*\mapsto \Ad^*_{g^{-1}}\vartheta= \vartheta\circ \Ad_{g^{-1}}\in \mathfrak{g}^*$.

The following definition has been widely used in the literature \cite{MRSV_15}, although we will see that the $\Ad^{*k}$-equivariance condition is no longer necessary (see \cite{LRVZ_23} for details). Moreover, we have changed the standard notation $\Coad^k$ to $\Ad^{*k}$ to shorten it.

\begin{definition}
    A $k$-polysymplectic momentum map $\mathbf{J}^\Phi:P\rightarrow (\Lg^*)^k$ is {\it $\Ad^{*k}$-equivariant} if
    \[
    \mathbf{J}^\Phi\circ \Phi_g =\Ad^{*k}_{g^{-1}}\circ\,\, \mathbf{J}^\Phi\,,\quad \forall g\in G\,,
    \]
    \begin{minipage}{12cm}
    where $\Ad^{*k}_{g^{-1}}=\Ad^*_{g^{-1}}\stackrel{}\otimes \overset{(k)}{\dotsb}\otimes \Ad^*_{g^{-1}}$ and
    \[
    \begin{array}{rccc}
    \Ad^{*k}&:G\x(\Lg^*)^k & \longrightarrow & (\Lg^*)^k\\
    & (g,\boldsymbol \mu) &\longmapsto & \Ad^{*k}_{g^{-1}}\boldsymbol\mu\,
    \end{array}.
    \]
In other words, the diagram aside is commutative for every $g\in G$.
\end{minipage}
\begin{minipage}{4cm}
    \begin{tikzcd}
    P
    \arrow[r,"\mathbf{J}^\Phi"]
    \arrow[d,"\Phi_g"]& (\mathfrak{g}^*)^k
    \arrow[d,"\Ad^{*k}_{g^{-1}}"]\\
    P
    \arrow[r,"\mathbf{J}^\Phi"]&
    (\Lg^*)^k.
    \end{tikzcd}
\end{minipage}
\end{definition}

To simplify the notation, let us introduce the following definition. 

\begin{definition}
 A {\it$G$-invariant $\bomega$-Hamiltonian system} is a tuple $(P,\boldsymbol\omega,\bm h,{\bf J}^\Phi)$, where $(P,\boldsymbol{\bomega})$ is a $k$-polysymplectic manifold, $\bm h$ is a $\bomega$-Hamiltonian function associated with $X_{\bm h}$, the map $\Phi:G\x P\rightarrow P$ is a $k$-polysymplectic Lie group action satisfying $\Phi_g^*\bm h=\bm h$ for every $g\in G$, and $\mathbf{J}^{\Phi}$ is a $k$-polysymplectic momentum map related to $\Phi$. An {\it $\Ad^{*k}$-equivariant $G$-invariant $\bomega$-polysymplectic Hamiltonian system} is a $G$-invariant $\bomega$-Hamiltonian system $(P,\boldsymbol\omega,\bm h,{\bf J}^\Phi)$ such that ${\bf J}^\Phi$  is $\Ad^{*k}$-equivariant.
\end{definition}

For simplicity, one sometimes calls {\it $\bomega$-Hamiltonian system} a triple $(P,\bomega,\bh)$ for a certain $\bomega$-Hamiltonian function $\bh$. 

Let us provide the formalism needed to avoid the $\Ad^{*k}$-equivariantness.

\begin{proposition}
\label{Prop::PsiConstant}
Let $(P,\boldsymbol\omega,\bm h,\mathbf{J}^\Phi)$ be a $G$-invariant $\bomega$-Hamiltonian system. If
\[
    \boldsymbol\psi _{g,\boldsymbol{\xi}}:P\ni x\longmapsto  {\bf J}^\Phi_{\boldsymbol{\xi}}(\Phi_g(x))-{\bf J}^\Phi_{\Ad_{g^{-1}}^k\boldsymbol{\xi}}(x)\in\mathbb{R}\,,\quad\forall g\in G\,,\quad\forall \boldsymbol{\xi}\in\mathfrak{g}^k\,,
\]
then $\boldsymbol\psi_{g,\boldsymbol{\xi}}$ is constant on $P$ for every $g\in G$ and $\boldsymbol{\xi}\in\mathfrak{g}^k$. Moreover, $\boldsymbol\sigma:G\ni g\mapsto {\boldsymbol \sigma}(g)\in (\mathfrak{g}^*)^k$, which is uniquely determined by the condition $\langle \boldsymbol\sigma(g),\boldsymbol{\xi}\rangle =\boldsymbol\psi_{g,\boldsymbol{\xi}}$ for every $
{\bm \xi}\in \mathfrak{g}^k$, satisfies
\[    \boldsymbol\sigma(g_1g_2)=\boldsymbol\sigma(g_1)+\Ad^{*k}_{g_1^{-1}}\boldsymbol\sigma(g_2)\,,\quad\forall g_1,g_2\in G\,.
\]
\end{proposition}

The map $\boldsymbol\sigma:G\rightarrow (\mathfrak{g}^*)^k$ of the form
\[
\boldsymbol\sigma(g)=\mathbf{J}^\Phi\circ \Phi_g-\Ad^{*k}_{g^{-1}}\mathbf{J}^\Phi,\qquad g\in G\,,
\]
is called the {\it co-adjoint cocycle} associated with the $k$-polysymplectic momentum map $\mathbf{J}^\Phi$ on $P$. Moreover, ${\bf J}^\Phi$ is an $\Ad^{*k}$-equivariant $k$-polysymplectic momentum map if and only if $\boldsymbol\sigma=0$. 

A map $\boldsymbol\sigma:G\rightarrow(\mathfrak{g}^*)^k$ is a {\it coboundary} if there exists $\boldsymbol\mu\in(\mathfrak{g}^*)^k$ such that
\[
\boldsymbol\sigma(g)=\boldsymbol\mu-\Ad_{g^{-1}}^{*k}\boldsymbol\mu\,,\qquad \forall g\in G\,.
\]

\begin{proposition}\label{Prop::GenEqJPolySym}
Let ${\bf J}^\Phi:P\rightarrow (\mathfrak{g}^*)^k$ be a $k$-polysymplectic momentum map related to a $k$-polysymplectic action $\Phi:G\times P\rightarrow P$ with co-adjoint cocycle $\boldsymbol\sigma$. Then, 

\begin{minipage}{12cm}
\begin{enumerate}[{\rm(1)}]
    \item there exists a Lie group action of $G$ on $(\mathfrak{g}^*)^k$ of the form
    \[
    \boldsymbol\Delta:G\times (\mathfrak{g}^*)^k\ni(g,\boldsymbol \mu)\mapsto \boldsymbol\sigma(g)+\Ad^{*k}_{g^{-1}}\boldsymbol\mu=:\boldsymbol\Delta_g(\boldsymbol\mu)\in(\mathfrak{g}^*)^k\,,
    \]   
    \item the $k$-polysymplectic momentum map ${\bf J}^\Phi$ is equivariant with respect to $\boldsymbol\Delta$, in other words, for every $g\in G$, one has the commutative diagram aside.
\end{enumerate}
\end{minipage}
\begin{minipage}{4cm}
\begin{center}
    \begin{tikzcd}
    P
    \arrow[r,"\mathbf{J}^\Phi"]
    \arrow[d,"\Phi_g"]& (\mathfrak{g}^*)^k
    \arrow[d,"\boldsymbol\Delta_g"]\\
    P
    \arrow[r,"\mathbf{J}^\Phi"]&
    (\mathfrak{g}^*)^k\,.
    \end{tikzcd}
    \end{center}
\end{minipage}
\end{proposition}

Proposition \ref{Prop::GenEqJPolySym} ensures that every $k$-polysymplectic momentum map $\mathbf{J}^\Phi$ gives rise to an equivariant $k$-polysymplectic momentum map relative to a new action $\boldsymbol\Delta: G\times(\mathfrak{g}^*)^k\rightarrow(\mathfrak{g}^*)^k$, called a {\it $k$-polysymplectic affine Lie group action}. Note that a $k$-polysymplectic affine Lie group action can also be expressed by writing $\boldsymbol\Delta(g,(\mu^1,
\ldots,\mu^k))=(\Delta^1_g\mu^1,\ldots,\Delta^k_g\mu^k)\in (\mathfrak{g}^*)^k$, where the mappings $\Delta^1,\ldots,\Delta^k$ take the form $\Delta^\alpha: G\times \mathfrak{g}^*\ni(g,\vartheta)\mapsto \Ad_{g^{-1}}^*\vartheta+\sigma^\alpha(g)=\Delta_g^\alpha(\vartheta)\in\mathfrak{g}^*$  and ${\boldsymbol \sigma}(g)=(\sigma^1(g),\ldots,\sigma^k(g))$, where  $\sigma^\alpha(g)=\mathbf{J}^\Phi_\alpha\circ\Phi_g-\Ad^*_{g^{-1}}\mathbf{J}^\Phi_\alpha$ for $\alpha=1,\ldots,k$ and ${\bf J}^\Phi_1,\ldots, {\bf J}^\Phi_k$ are the coordinates of ${\bf J}^\Phi$.

\section{\texorpdfstring{$k$}{}-Polysymplectic Marsden--Weinstein reduction}\label{Sec::reduction}

Let us now review previous results in the literature for the $k$-polysymplectic Marsden--Weinstein reduction to correct previous mistakes and inaccuracies. Furthermore, we introduce the reduction of the dynamical system governed by an $\bomega$-Hamiltonian vector field. This concept is novel, as prior research has focused on dynamical systems given by Hamiltonian $k$-vector fields \cite{Bla_19, MRSV_15}. In particular, this section  first reviews the previous $k$-polysymplectic Marsden--Weinstein reduction theory and explains some, only apparently, minor inaccuracies. After that, we focus on solving a mistake in one of the main results in \cite{GM_23}, concerning the conditions to obtain a $k$-polysymplectic reduction. Finally, in Subsection \ref{Subsec::Conditions}, we analyse the relations between the conditions for the $k$-polysymplectic reduction given in \cite{MRSV_15}.

\subsection{A review on the \texorpdfstring{$k$}{}-polysymplectic Marsden--Weinstein reduction}
Let us recall several definitions that are useful for what follows. Some technical assumptions will be first set to improve the applicability of $k$-polysymplectic Marsden--Weinstein reductions. A {\it weak regular value} of a mapping $\phi:M\rightarrow N$ is a point $x_0\in N$ such that $\phi^{-1}(x_0)$ is a submanifold of $M$ and $\ker \T_p\phi = \T_p[\phi^{-1}(x_0)]$ for every $p\in \phi^{-1}(x_0)$. In particular, regular values of $\phi$ are weak regular values too. Moreover, a Lie group action $\Phi: G\times M\rightarrow M$ is {\it quotientable} \cite{Alb_89} when the space of orbits of the action of $G$ on $M$, let us say $M/G$, is a manifold and the projection $\pi: M\rightarrow M/G$ is a submersion. In particular, this occurs when $\Phi$ is free and proper.

Let us comment on the regular values of $k$-polysymplectic momentum maps. The codomain of a $k$-polysymplectic momentum map ${\bf J}^{\Phi}:P\rightarrow \mathfrak{g}^{*k}$ may have a large dimension, even larger than the dimension of $P$, for instance, due to the presence of $k$ copies of $\mathfrak{g}^*$. This implies that it may be impossible for ${\bf J}^\Phi$ to be a submersion when $k$ is large enough. Being a submersion is the typical condition used in many types of Marsden--Weinstein reductions \cite{GM_23, MRSV_15}. But this property is harder to satisfy in $k$-polysymplectic geometry. Note that it is sometimes assumed in the literature that the Sard's Theorem ensures that ${\bf J}^\Phi$ is frequently a submersion because the set of singular points in $P$ of ${\bf J}^\Phi$, i.e. the set of points where ${\bf J}^\Phi$ is not a submersion, has an image with zero measure (see \cite[Lemma 3.4]{MRSV_10} or \cite{BR_04}). Nevertheless, the whole image of ${\bf J}^\Phi$ may also be a zero measure subset and, in this case, it may happen that ${\bf J}^\Phi$ is not a submersion at points in a dense subset of $P$. Indeed, ${\bf J}^\Phi$ is not a submersion at any point in $P$ when $ k\dim\mathfrak{g}^{*}>\dim P$. In such a case, ${\bf J}^\Phi$ has no regular points in $\mathfrak{g}^{*k}$. That is one of the reasons why the analysis of weak regular values for $k$-polysymplectic momentum maps in \cite{LRVZ_23} is relevant. It also explains why in the symplectic case, when $k=1$, the assumption of ${\bf J}^\Phi$ being a submersion is not so problematic. Note also that one has to assume some regularity conditions on the coordinates of ${\bf J}_1^{\Phi},\ldots,{\bf J}_k^{\Phi}$ to ensure that their level sets are submanifolds, but such mappings do not use to have regular values in $k$-polysymplectic problems for $k>1$.

It is also worth stressing that Blacker in \cite[Theorem 3.22]{Bla_19} does not provide any explicit assumption in the structure of $\mathbf{J}^{\Phi-1}(\bm\mu)$, although it is implicitly assumed that $\mathbf{J}^{\Phi-1}(\bm\mu)$ is a manifold. 
In general, Blacker's work \cite{Bla_19} does not analyse in detail the technical conditions on the manifold structure of ${\bf J}^{\Phi-1}(\bm \mu)$. Notwithstanding, the structure of spaces of the form ${\bf J}^{\Phi-1}(\bm \mu)/G_{\bm \mu}$ is investigated.

Lemma \ref{Lemm::NonAdPerpPS} below will be used to characterise in the next section the so-called $k$-polysymplectic relative equilibrium points of $G$-invariant $\bomega$-Hamiltonian systems. More importantly, Lemma \ref{Lemm::NonAdPerpPS} is introduced to prove $k$-polysymplectic Marsden--Weinstein reduction theorems. The proof of Lemma \ref{Lemm::NonAdPerpPS} appears in \cite{LRVZ_23}.  Interestingly, G\"unther's wrong version of Lemma \ref{Lemm::NonAdPerpPS}  made his $k$-polysymplectic reduction to be flawed. In fact, G\"unther states in \cite[Lemma 7.5]{Gun_87} a wrong expression for condition (1) in Lemma 
\ref{Lemm::NonAdPerpPS}. In his work, G\"unther implicitly claims that, as in the symplectic case, one has
$$
\ker \iota_{\bm \mu}^*\bm\omega=\T_{p}\left({\bf J}^{\Phi-1}({\bm\mu})\right)^{\perp,k}\cap \T_{p}\left({\bf J}^{\Phi-1}({\bm\mu})\right)=\T_{p}(G p)\cap \T_{p}\left({\bf J}^{\Phi-1}({\bm\mu})\right)=\T_{p}(G^{\bm\Delta}_{\bm\mu} p),
$$
but the equality between the second and the third expressions is only an inclusion $\supset$ (see \cite[pg. 12]{MRSV_15}). It is worth stressing that G\"unther justifies his Lemma 7.5 by merely saying that its proof is like in the symplectic case \cite[pg. 48]{Gun_87}. 
Moreover, \cite{MRS_04} includes a related mistake. Finally, we refer to \cite[Sections 1 and  2.2]{MRSV_15} for a comment on these errors. 

\begin{lemma}
\label{Lemm::NonAdPerpPS}
Let $(P,\bomega, \bm h,\mathbf{J}^\Phi)$ be a $G$-invariant $\bomega$-Hamiltonian system and let $\bm\mu\in(\Lg^*)^k$ be a weak regular value of  ${\bf J}^\Phi:P\rightarrow (\Lg^*)^k$. Then, for every $p\in {\bf J}^{\Phi-1}(\bm\mu)$, one has 
\begin{enumerate}[{\rm(1)}]
\item $\T_{p}(G^{\bm\Delta}_{\bm\mu} p) =\T_{p}(G p)\cap \T_{p}\left({\bf J}^{\Phi-1}({\bm\mu})\right)$, 
\item $\T_{p}({\bf J}^{\Phi-1}(\bm\mu)) = \T_{p}(Gp)^{\perp,k}$.
\end{enumerate}
\end{lemma}

Let us review the conditions of the $k$-polysymplectic Marsden--Weinstein reduction theorem, which will be crucial in the $k$-polysymplectic energy-momentum method to correct a mistake in one of the main results in \cite{GM_23}, in fact, the one, \cite[Proposition 1]{GM_23}, giving the name to the paper. 

Recall that the first correct $k$-polysymplectic Marsden--Weinstein reduction theory can be found in \cite{MRSV_15}. The necessary and sufficient conditions to perform a reduction were given by C. Blacker in \cite{Bla_19}, although there is a relevant typo in his theorem, as commented in \cite{GM_23}. The $k$-polysymplectic Marsden--Weinstein reduction theorem was proved in \cite{MRSV_15} assuming that the $k$-polysymplectic momentum map $\mathbf{J}^\Phi:P\rightarrow (\mathfrak{g}^*)^k$ is $\Ad^{*k}$-equivariant. A version of the $k$-polysymplectic Marsden--Weinstein reduction theorem without this condition was accomplished in \cite{LRVZ_23}. In its correct and most modern form, the reduction theorem reads as in Theorem \ref{Th::PolisymplecticReductionJ} below (see \cite[Theorem 5.10]{LRVZ_23} for details). Note that when we say that $\bm \mu$ is a weakly regular value of ${\bf J}^\Phi$, we also assume that all the components of $\bm\mu$ are weakly regular too. It is worth stressing that even if $\bm \mu$ is a  regular value of ${\bf J}^\Phi$, then each component $\mu^
\alpha$ of $\bm\mu$ does not need to be regular for ${\bf J}^\Phi_\alpha$ since ${\bf J}_\alpha^{\Phi-1}(\mu^\alpha)\supset {\bf J}^{\Phi-1}(\bm\mu)$.

\begin{theorem}[$k$-polysymplectic Marsden--Weinstein reduction theorem]\label{Th::PolisymplecticReductionJ}
    Consider a $G$-invariant $\bomega$-Hamiltonian system  $(P,\bomega,\bm h,\mathbf{J}^\Phi)$. Assume that $\boldsymbol\mu=(
    \mu^1,\ldots,
    \mu^k)\in (\Lg^*)^k$ is a weak regular value of $\mathbf{J}^\Phi$ and $G_{\boldsymbol\mu}^{\boldsymbol\Delta}$ acts in a quotientable manner on $\mathbf{J}^{\Phi-1}(\boldsymbol\mu)$. Let $G^{
    \Delta^\alpha}_{\mu^\alpha}$ denote the isotropy group at $\mu^\alpha$ of the Lie group action $\Delta^\alpha:(g,\vartheta)\in G\x\mathfrak{g}^*\mapsto  \Delta^\alpha(g,\vartheta)\in \mathfrak{g}^*$ for $\alpha=1,\ldots,k$. Moreover, let the following (sufficient) conditions hold
    \begin{equation}\label{Eq::PolysymplecticReduction1eq}
        \ker (\T_p\mathbf{J}_\alpha^\Phi) = \T_p(\mathbf{J}^{\Phi-1}(\boldsymbol\mu))+\ker\omega^\alpha_p + \T_p(G^{\Delta^\alpha}_{\mu^\alpha} p)\,,\qquad \alpha=1,\ldots,k\,,
    \end{equation}
    \begin{equation}\label{Eq::PolysymplecticReduction2eq}
        \T_p(G_{\boldsymbol\mu}^{\boldsymbol\Delta} p) = \bigcap^k_{\alpha=1}\left(\ker\omega^\alpha_p+\T_p(G^{\Delta^\alpha}_{\mu^\alpha}p)\right)\cap \T_p(\mathbf{J}^{\Phi-1}(\boldsymbol\mu))\,,
    \end{equation}
    for every $p\in {\bf J}^{\Phi-1}({\boldsymbol \mu})$. Then,  $(\mathbf{J}^{\Phi-1}(\boldsymbol\mu)/G^{\boldsymbol\Delta}_{\boldsymbol\mu},\bomega_{\boldsymbol\mu})$ is a $k$-polysymplectic manifold, with ${\bomega}_{\boldsymbol \mu}$ being uniquely determined by
    \[ \pi_{\boldsymbol\mu}^*\bomega_{\boldsymbol\mu}=\jmath_{\boldsymbol\mu}^*\bomega
    \]
    where $\jmath_{\boldsymbol\mu}:\mathbf{J}^{\Phi-1}(\boldsymbol\mu)\hookrightarrow P$ is the canonical immersion and $\pi_{\boldsymbol\mu}:\mathbf{J}^{\Phi-1}(\boldsymbol\mu)\rightarrow \mathbf{J}^{\Phi-1}(\boldsymbol\mu)/G^{\boldsymbol\Delta}_{\boldsymbol\mu}$ is the canonical projection.
\end{theorem}

The following theorem shows the reduction of the dynamics given by an $\bomega$-Hamiltonian vector field $X_{\bh}$ on $P$ as a consequence of Theorem \ref{Th::PolisymplecticReductionJ}, which will be essential for our $k$-polysymplectic energy-momentum method. Note that in  previous works on $k$-polysymplectic Marsden--Weinstein reductions, the $k$-polysymplectic Marsden--Weinstein  reduction theorem has been applied to reduce the dynamics given by an ${\bm \omega}$-Hamiltonian $k$-vector field \cite[Theorem 4.4]{MRSV_15}. Nevertheless, since our $k$-polysymplectic Marsden--Weinstein reduction theorem concerns just ${\bm \omega}$-Hamiltonian vector fields, the conditions of that theorem can be simplified as follows.

\begin{theorem}
\label{Th::Xreduction}
    Let $(P,\bomega,\bm h,\mathbf{J}^\Phi)$ be a $G$-invariant $\bomega$-Hamiltonian system and let $\Phi_{g*}{\bm h}={\bm h}$ for each $g\in G$. Then, the one-parametric group of diffeomorphisms $F_t$ of the vector field $X_{\bm h}$ induces the one-parametric group of diffeomorphisms $\mathcal{F}_t$ of the vector field $X_{\boldsymbol{f}_{\boldsymbol{\mu}}}$ on $\mathbf{J}^{\Phi-1}(\boldsymbol{\mu})/G^{\boldsymbol{\Delta}}_{\boldsymbol{\mu}}$ such that $\iota_{X_{\boldsymbol{f}_{\boldsymbol{\mu}}}}\boldsymbol{\omega}_{\boldsymbol{\mu}}=\d\boldsymbol{f}_{\boldsymbol{\mu}}$ and $\jmath_{\boldsymbol{\mu}}^*\boldsymbol{h}=\pi_{\boldsymbol{\mu}}^*\boldsymbol{f}_{\boldsymbol{\mu}}$.
\end{theorem}
\begin{proof}
   First, note that $\Phi_g^*\boldsymbol{h}=\boldsymbol{h}$ and our assumptions, in particular $\Phi_g^*\boldsymbol{\omega}=\boldsymbol{\omega}$,  yield $\Phi_{g*}X_{\boldsymbol{h}}=X_{\boldsymbol{h}}$ for each $g\in G$. Therefore,
   \[
   \inn{X_{\bm h}}\d \langle \mathbf{J}^\Phi,\xi\rangle=-\iota_{\xi_P}\iota_{X_{\bm h}}\boldsymbol{\omega}=-\iota_{\xi_P}\d\boldsymbol{h}=0\,,\qquad \forall \xi\in \mathfrak{g}\,.
   \]
    Hence, $X_{\bm h}$ is tangent to $\mathbf{J}^{\Phi-1}(\boldsymbol{\mu})$.
    Next, for every $\xi\in\mathfrak{g}$, we have
    \[\inn{[\xi_P,X_{\bm h}]}\boldsymbol{\omega}=\Lie_{\xi_P}\inn{X_{\bm h}}\boldsymbol{\omega}-\inn{\xi_P}\Lie_{X_{\bm h}}\boldsymbol{\omega}=0\,,
    \]
    so by the virtue of $\ker\boldsymbol{\omega}=0$, we obtain that $[\xi_P,X_{\bm h}]=0$. Thus, the vector field $X_{\bm h}$ projects onto a vector field $Y$ on the reduced manifold $\mathbf{J}^{\Phi-1}(\boldsymbol{\mu})/G^{\boldsymbol{\Delta}}_{\boldsymbol{\mu}}$. In other words, the one-parametric group of diffeomorphisms  $F_t$ of $X_{\bm h}$ induces the one-parametric group of diffeomorphisms $\mathcal{F}_t$ of $Y$ so that $\pi_{\boldsymbol{\mu}}\circ F_t = \mathcal{F}_t\circ \pi_{\boldsymbol{\mu}}$ for each $t\in \mathbb{R}$. Then, by Theorem \ref{Th::PolisymplecticReductionJ}, one has
    \begin{equation}
    \label{Eq::Xfreduction}
    \jmath_{\boldsymbol{\mu}}^*\d\boldsymbol{h} = \jmath^*_{\boldsymbol{\mu}}(\inn{X_{\bm h}}\boldsymbol{\omega}) = \inn{X_{\bm h}}\jmath^*_{\boldsymbol{\mu}}\boldsymbol{\omega} = \inn{X_{\bm h}}\pi^*_{\boldsymbol{\mu}}\boldsymbol{\omega}_{\boldsymbol{\mu}}, = \pi^*_{\boldsymbol{\mu}}(\inn{Y}\boldsymbol{\omega}_{\boldsymbol{\mu}})\,,
    \end{equation}
    where we denoted by $X_{\bm h}$ both the vector field $X_{\bm h}$ on $P$ itself and its restriction to $\mathbf{J}^{\Phi-1}(\boldsymbol{\mu})$. The same slight abuse of notation will be hereafter done to simplify the notation. 
    
    Due to the invariance of ${\boldsymbol{h}}$ relative to $G_{\boldsymbol{
    \mu}}^{\boldsymbol{\Delta}}$, there exists a reduced $\R^k$-valued function $\boldsymbol{f}_{\bm\mu}$ on ${\bf J}^{\Phi-1}(\bm\mu)$ such that $\jmath_{\bm\mu}^*{\bh}=\pi_{\bm\mu}^*{\bm f}_{\bm \mu}$.     
    Finally,  expression \eqref{Eq::Xfreduction} gives
    \begin{equation*}
\pi_{\bm\mu}^*\d\boldsymbol{f}_{\boldsymbol{\mu}}=\jmath_{\boldsymbol{\mu}}^*\d\boldsymbol{h}=\pi_{\boldsymbol{\mu}}^*\iota_{Y}\boldsymbol{\omega}_{\boldsymbol{\mu}}
    \end{equation*}
    which shows that $Y=X_{\boldsymbol{f}_{\bm\mu}}$ is an $\bm{\omega}_{\boldsymbol{\mu}}$-Hamiltonian vector field and $\boldsymbol{f}_{\bm\mu}$ is an $\bm\omega_{\bm\mu}$-Hamiltonian function associated with $X_{\boldsymbol{f}_{\bm\mu}}$.
\end{proof}

Now, let us recall the sufficient and necessary conditions for a $k$-polysymplectic reduction given by Blacker in \eqref{Eq::BlackerCondition}. His main result is described in Theorem \ref{Th::BlackerReduction} with our notation and we have corrected the typo in \cite[Theorem 3.22]{Bla_19} on the $k$-polysymplectic Marsden--Weinstein reduction. It is worth noting that the typo also appears in the proof of \cite[Theorem 3.22]{Bla_19} and is evident after applying \cite[Theorem 2.14]{Bla_19} to ${\bm\omega}_x$. Theorem \ref{Th::BlackerReduction} also adds certain essential technical conditions that were not explicitly written in \cite[Theorem 3.22]{Bla_19}. As remarked by Blacker in \cite{Bla_09}, but apparently not noticed by Mestdag and Garc\'ia-Tora\~no in \cite{GM_23}, the sufficient and necessary condition \eqref{Eq::BlackerCondition} also appeared previously to Blacker in a different more implicit manner in \cite[pg. 12]{MRSV_15}. It is worth seeing also the related work \cite{Ni15} treating the reduction of poly-Poisson structures.

\begin{theorem}
\label{Th::BlackerReduction}Let $(P,\boldsymbol{\omega},\bh,\mathbf{J}^\Phi)$ be an $\Ad^{*k}$-equivariant $G$-invariant $\bomega$-Hamiltonian system and let $\bm\mu\in (\mathfrak{g}^*)^k$ be a fixed regular value of $\mathbf{J}^\Phi$. If the stabiliser subgroup $G_{\bm\mu}$ of $\bm\mu$ under the $\Ad^{*k}$ action is connected, and  $P_{\bm\mu}=\mathbf{J}^{\Phi-1}(\bm\mu)/G_{\bm\mu}$ is a smooth manifold, then there is a unique $\mathbb{R}^k$-valued two-form
$\bm\omega_{\bm\mu}\in \Omega^2(P_{\bm\mu},\mathbb{R}^k)$ such that
$\pi_{\bm\mu}^*\bomega_{\bm\mu}=\jmath^*_{\bm\mu}\bomega$ where $\jmath_{\bm\mu} : \mathbf{J}^{\Phi-1}
(\bm\mu)\hookrightarrow P$ is the inclusion and $\pi_{\bm\mu}:\mathbf{J}^{\Phi-1}
(\bm\mu) \rightarrow P_{\bm\mu}$ is the canonical projection. The form
$\bomega_{\bm\mu}$ is closed and nondegenerate if and only if 
\begin{equation}
\label{Eq::BlackerCondition}
\T_{p}(G_{\bm\mu}p)=(\T_{p}(Gp)^{\perp,k})^{\perp,k}\cap \T_{p}(Gp)^{\perp,k}\,,\qquad \forall p\in {\bf J}^{\Phi-1}({\bm \mu})\,.
\end{equation}
\end{theorem}

For the sake of completeness, let us now consider the first example of a $k$-polysymplectic Marsden--Weinstein reduction related to a non-regular value of a $k$-polysymplectic momentum map. More examples with potential practical applications will be shown in Section \ref{Sec::examples}. Let us analyse  the completely integrable, and separable in variables, system in $\mathbb{R}^{2k}$ of the form
\begin{equation}\label{eq:ComInSys}
\frac{\d I_\alpha}{\d t}=0\,,\qquad \frac{\d\theta_\alpha}{\d t} = F_\alpha(I_\alpha)\,,\qquad \alpha=1,\ldots,k>1\,,
\end{equation}
for some arbitrary functions $F_1,\ldots,F_k:
\mathbb{R}
\rightarrow \mathbb{R}$. 
This related to an $\bomega$-polysymplectic Hamiltonian system on $\mathbb{R}^{2k}$ relative to the $k$-polysymplectic form $\bomega = \omega^\alpha\otimes e_\alpha$, where $\omega^1,\ldots,\omega^k$ are the presymplectic forms
$$
\omega^\alpha = \d\theta^\alpha\wedge \d I^\alpha\,,\qquad \alpha=1,\ldots,k\,,
$$
where it is important to stress that the right-hand side is not summed over the indices $\alpha=1,\ldots,k$. 
One has the basis of fundamental vector fields  $\partial/\partial \theta^1,\ldots,\partial/\partial \theta^{k-1}$ associated with the Lie group action
$$
\Phi:(\lambda_1,\ldots,\lambda_{k-1};\theta_1,\ldots,\theta_k,I)\in \mathbb{R}^{k-1}\times\mathbb{R}^{2k}\mapsto (\lambda_1+\theta_1,\ldots,\lambda_{k-1}+\theta_{k-1},\theta_k,I)\in\mathbb{R}^{2k},
$$
with $I=(I_1,\ldots,I_k)\in \mathbb{R}^k$.  Note that the functions $F_1,\ldots,F_k$ have been chosen to be of the form $F_\alpha=F_\alpha(I_\alpha)$, with $\alpha=1,\ldots,k$, to ensure that \eqref{eq:ComInSys} is $\bm\omega$-Hamiltonian. The latter also explains why \eqref{eq:ComInSys} is called separable. One may now consider a $k$-polysymplectic momentum map
\[
{\bf J}^{\Phi}:(\theta,I)\in \mathbb{R}^{2k}\longmapsto (I_1,\ldots,0)\otimes e_1+\ldots+(0,\ldots,I_{k-1})\otimes e_{k-1}+(0,\ldots,0)\otimes e_{k} \in \left(\mathbb{R}^{(k-1)*}\right)^k\,,
\]
which has no regular points (the codomain of ${\bf J}^\Phi$ has dimension larger than its domain for $k>3$) and it is $\Ad^{*k}$-equivariant. 
Note that \eqref{eq:ComInSys} gives rise to an $\mathbb{R}^{k-1}$-invariant $\bomega$-Hamiltonian system. 

One may consider the reductions of $\bomega$ and \eqref{eq:ComInSys} for any value of $\bm \mu=(\mu^1,\ldots,0)\otimes e_1+\ldots+(0,\ldots,\mu^{k-1})\otimes e_{k-1}\in \left(\mathbb{R}^{(k-1)*}\right)^k$. Then,
\[
    {\bf J}^{\Phi-1}(\bm \mu) = \{(\theta_\alpha,I_\alpha)\in \mathbb{R}^{2k}\mid I_1=\mu^1,\ \dotsc,\ I_{k-1}=\mu^{k-1},\ \theta_1,\ldots, \theta_k,I_k\in \mathbb{R}\}\simeq \mathbb{R}^{k}\times \mathbb{R}\,.
\]
The isotropy subgroup $\mathbb{R}^{k-1}_{\bm\mu}\simeq \mathbb{R}^{k-1}$ acts on ${\bf J}^{\Phi-1}(\bm\mu)$ via $\Phi$ and the reduced manifold is diffeomorphic to  $\mathbb{R}^2$. The presymplectic forms $\omega^1,\ldots,\omega^{k-1}$ become zero after reducing, but  the reduction of $\omega^k$ is symplectic. Hence, $\bomega_{\bm\mu}$ becomes a $k$-polysymplectic form with only one symplectic form different from zero. Since the $\bomega$-Hamiltonian function of the initial system is a first integral of the $\theta_1,\ldots,\theta_{k-1}$, one can project the initial system onto 
\[
\frac{\d I_k}{\d t}=0\,,\qquad \frac{\d\theta_k}{\d t} = F_k(I_k)\,,
\]
which is Hamiltonian relative to $\d\theta_k\wedge \d I_k$, where $\theta_k, I_k$ are considered as variables in $\mathbb{R}^2$ in the natural manner.

\subsection{On the conditions for the \texorpdfstring{$k$}--polysymplectic Marsden--Weinstein reduction}
\label{Subsec::Conditions}

It was claimed in \cite[Proposition 1]{GM_23}  that condition \eqref{Eq::PolysymplecticReduction2eq} is enough to ensure that there exists a $k$-polysymplectic Marsden--Weinstein reduction. In this section, we first show that this is not true and the proposition in \cite[Proposition 1]{GM_23} is incorrect. This is done by pointing out a mistake in the proof of \cite[Proposition 1]{GM_23} and then giving a counterexample where \eqref{Eq::PolysymplecticReduction2eq} is satisfied, but there is no $k$-polysymplectic Marsden--Weinstein reduction and, indeed, \eqref{Eq::PolysymplecticReduction1eq} does not hold. Next, we illustrate that it may happen that \eqref{Eq::PolysymplecticReduction1eq} is satisfied, but \eqref{Eq::PolysymplecticReduction2eq} is not. Finally, we prove an example of a possible $k$-polysymplectic reduction where \eqref{Eq::PolysymplecticReduction1eq} and \eqref{Eq::PolysymplecticReduction2eq} are not simultaneously satisfied. To keep our exposition simple and highlight our main ideas, we restrict in this subsection all $k$-polysymplectic momentum maps to the ${\rm Ad}^{*k}$-invariant case, as done in \cite{GM_23,MRSV_15}.

First, the proof for \cite[Proposition 1]{GM_23} has a mistake, as there is an inclusion written in the opposite way. In particular, since $\T_p(\mathbf{J}^{\Phi-1}(\bm\mu))\subset \T_p(\mathbf{J}_\alpha^{\Phi-1}(\mu^\alpha))$ for $\alpha=1,\ldots,k$ and every $p\in \mathbf{J}^{\Phi-1}(\bm\mu)$ for a regular $\bm\mu\in \mathfrak{g}^{*k}$, one has
\begin{multline*}
    \left\{v\in \T_pP\mid\omega^1(v,\T_p\mathbf{J}^{\Phi-1}_1(\mu^1))=\dotsb=\omega^k(v,\T_p\mathbf{J}^{\Phi-1}_k(\mu^k))=0\right\} \\
    \subset \left\{v\in \T_pP\mid\omega^1(v,\T_p\mathbf{J}^{\Phi-1}(\bm\mu))=\dotsb=\omega^k(v,\T_p\mathbf{J}^{\Phi-1}(\bm\mu)) = 0\right\}
\end{multline*}
instead of
\begin{multline*}
    \left\{v\in \T_pP\mid \omega^1(v,\T_p\mathbf{J}^{\Phi-1}_1(\mu^1))=\dotsb=\omega^k(v,\T_p\mathbf{J}^{\Phi-1}_k(\mu^k))=0\right\}\\
    \supset \left\{v\in \T_pP\mid \omega^1(v,\T_p\mathbf{J}^{\Phi-1}(\bm\mu))=\dotsb=\omega^k(v,\T_p\mathbf{J}^{\Phi-1}(\bm\mu)) = 0\right\}
\end{multline*}
as claimed at the end of page 8 in the proof of \cite[Proposition 1]{GM_23}. In other words, if $v$ is perpendicular to $\T_p(\mathbf{J}^{\Phi-1}(\bm\mu))$ relative to each $\omega^\alpha$, one cannot infer that $v$ is perpendicular to each $\T_p(\mathbf{J}^{\Phi-1}_\alpha(\mu^\alpha))$ relative to $\omega^\alpha$ for $\alpha=1,\dotsc,k$, since the latter conditions are more restrictive.
Then, the proof of \cite[Proposition 1]{GM_23} only gives
$$
    \bigcap_{\alpha=1}^k(\ker \jmath^*_{\mu^\alpha}\omega^\alpha|_p)\cap \T_p{\mathbf{J}^{\Phi-1}(\bm\mu)}\subset (\T_{p}(Gp)^{\perp,k})^{\perp,k}\cap \T_{p}(Gp)^{\perp,k}\,,\qquad \forall p\in \mathbf{J}^{\Phi-1}(\bm\mu)\,,
$$
instead of the claimed
$$
\bigcap_{\alpha=1}^k(\ker \jmath^*_{\mu^\alpha}\omega^\alpha|_p)\cap \T_p{\mathbf{J}^{\Phi-1}(\bm\mu)}\supset (\T_{p}(Gp)^{\perp,k})^{\perp,k}\cap \T_{p}(Gp)^{\perp,k}\,,\qquad \forall p\in \mathbf{J}^{\Phi-1}(\bm\mu)\,,
$$
which makes the proof of Proposition 1 fail in proving \eqref{Eq::BlackerCondition}, namely the $k$-polysymplectic Marsden--Weinstein reduction necessary and sufficient condition, and, therefore, the statement of Proposition 1. Indeed, the above mistake is ultimately due to the fact that \cite[Proposition 1]{GM_23} is false and the comments that follow in \cite{GM_23} contain some inaccuracies.

Let us provide a counterexample to show that \cite[Proposition 1]{GM_23} does not hold. More specifically, we here describe an $\mathbb{R}$-invariant $\boldsymbol{\omega}$-Hamiltonian system relative to a two-symplectic form satisfying condition \eqref{Eq::PolysymplecticReduction2eq} but not giving rise to a $k$-polysymplectic Marsden--Weinstein reduction. Before that, it is convenient to recall some results from 
\cite{MRSV_15}.

It was proved in \cite{MRSV_15} that $\ker\omega_p^\alpha\subset \ker \T_p\mathbf{J}^\Phi_\alpha$ on $\mathbf{J}^{\Phi-1}(\bm\mu)$, which allows one to define the following commutative diagram (see \cite[pg. 12]{MRSV_15})
\begin{center}
    \begin{tikzcd}[column sep=2cm]
    \T_p(\bfJ^{\Phi\,-1}(\bm\mu)) \arrow[r, "\jmath"] \arrow[rr, bend left=20, "\pi^\alpha_p"] & \ker\T_p\bfJ_\alpha^\Phi \arrow[r, "\pi"] & \dfrac{\ker\T_p\bfJ_\alpha^\Phi}{\ker\omega^\alpha_p}
    \end{tikzcd}
\end{center}
for all $p\in\mathbf{J}^{\Phi-1}(\bm\mu)$, where $\jmath$ and $\pi$ are the canonical injection and projections, respectively. For simplicity, the equivalence class of an element $v$ in a quotient will be denoted by $[v]$. To avoid making the notation too complicated, the specific meaning of $[v]$ will be understood from context. According to Proposition 3.12 in \cite{MRSV_15}, the above diagram induces the maps
$$
    \widetilde{\pi}_p^\alpha:\frac{\T_p(\mathbf{J}^{\Phi-1}(\bm\mu))}{\T_p(G_{\bm\mu} p)}\longrightarrow \frac{\frac{\ker \T_p\mathbf{J}^\Phi_\alpha}{\ker \omega_p^\alpha}}{\{[ (\xi_P)_p]\mid\xi\in \mathfrak{g}_{\mu^\alpha}\}}\,,\qquad \alpha=1,\ldots,k\,,\qquad \forall p\in\mathbf{J}^{\Phi-1}(\bm\mu)\,,
$$
where $\mathfrak{g}_{\mu^\alpha}$ is the Lie algebra of $G_{\mu^\alpha}$ and $\{[ (\xi_P)_p]\mid\xi\in \mathfrak{g}_{\mu^\alpha}\}=\pr^{P}_\alpha(\{ (\xi_P)_p\mid\xi\in \mathfrak{g}_{\mu^\alpha}\})$ and ${\rm pr}^P_\alpha:\T_pP\rightarrow \T_pP/\ker \omega_p^\alpha$ is the canonical projection onto the quotient.

The conditions \eqref{Eq::PolysymplecticReduction1eq} at $p\in P$ are equivalent to each $\widetilde{\pi}_p^\alpha$ being surjective, respectively \cite[Lemma 3.15]{MRSV_15}, while \eqref{Eq::PolysymplecticReduction2eq} amounts to $0=\bigcap_{\alpha=1}^k\ker \widetilde{\pi}_p^\alpha$ (see \cite[Lemma 3.16]{MRSV_15}).

Consider $P=\mathbb{R}^4$ with linear coordinates $\{x,y,z,t\}$ and the presymplectic forms
$$
    \omega^1 = \d x\wedge \d y\,,\qquad \omega^2 = \d x\wedge \d t + \d y\wedge\d z\,,
$$
which give rise to a two-polysymplectic form $\bm\omega=\omega^1\otimes e_1+\omega^2\otimes e_2$, because $\omega^2$ is a symplectic form and $\ker\omega^1\cap \ker \omega^2=0$.
Consider the Lie group action 
$\Phi:(\lambda;x,y,z,t)\in \mathbb{R}\times \mathbb{R}^4\mapsto (x+\lambda,y,z,t)\in \mathbb{R}^4$. The Lie algebra of fundamental vector fields of $\Phi$ is $V=\langle \partial_x\rangle\simeq \mathbb{R}$. Moreover, $\Phi$ admits a two-polysymplectic momentum map relative to $(\mathbb{R}^4,\bomega)$ given by 
$$
    \mathbf{J}^\Phi:(x,y,z,t)\in \mathbb{R}^4\longmapsto \bm\mu=(y,t)\in (\mathbb{R}^*)^2,
$$
which is clearly $\Ad^{*2}$-equivariant. Additionally, $\mathbf{J}^\Phi$ is regular for every value of $(\mathbb{R}^*)^2$. Hence,
$\mathbf{J}^{\Phi-1}(y,t)=\{(x,y,z,t)\in \mathbb{R}^4:x,z\in \mathbb{R}\}\simeq \mathbb{R}^2$ is a submanifold for every $(y,t)\in (\mathbb{R}^*)^2$ and
$$
    \T_p(\mathbf{J}^{\Phi-1}(y,t))=\langle \partial_x,\partial_z\rangle\,,\qquad \forall p\in \mathbf{J}^{\Phi-1}(y,t)\,.
$$
Moreover, $G_{\bm\mu}=\mathbb{R}$ for each ${\bm\mu}=(y,t)\in (\mathbb{R}^*)^2$ and $G_{\bm\mu}$ acts freely and properly on $\mathbf{J}^{\Phi-1}(\bm\mu)$. 
Let us prove that condition \eqref{Eq::PolysymplecticReduction2eq} does not imply nor the reduction of $\boldsymbol{\omega}$, namely \eqref{Eq::BlackerCondition},  neither \eqref{Eq::PolysymplecticReduction1eq}.

In our example, one has $\bm\mu=(y,t)$ with $\mu^1=y$ and $\mu^2=t$, while
$$
    \ker \T_p\mathbf{J}^\Phi_1 = \langle \partial_x,\partial_z,\partial_t\rangle\,,\quad \ker \omega^1=\langle \partial_t,\partial_z\rangle\,,\quad \ker \T_p\mathbf{J}^\Phi_2 = \langle \partial_x,\partial_y,\partial_z\rangle\,,\quad \ker \omega^2 = 0\,,
$$
and
$$
    \{ [(\xi_P)_p]:\xi\in \mathfrak{g}_{\mu^1}\}=\langle [\partial_x]\rangle\,,\qquad \{ [(\xi_P)_p]:\xi\in \mathfrak{g}_{\mu^2}\}=\langle [\partial_x]\rangle\,
$$
on ${\bf J}^{\Phi-1}(\bm\mu)$. Then, we have the mappings
$$
    \widetilde{\pi}_p^1:\langle [\partial_z]\rangle = \T_p(\mathbf{J}^{\Phi-1}(\bm\mu))/\T_p(G_{\bm\mu} p)\longmapsto \langle 0\rangle= (\ker \T_p\mathbf{J}^\Phi_1/\ker \omega_p^1)/\langle [\partial_x]\rangle
$$
and
$$
    \widetilde{\pi}_p^2:\langle [\partial_z]\rangle = \T_p(\mathbf{J}^{\Phi-1}(\bm\mu))/\T_p(G_{\bm\mu} p)\longmapsto \langle [\partial_y],[\partial_z]\rangle=(\ker \T_p\mathbf{J}^\Phi_2/\ker \omega_p^2)/\langle [\partial_x]\rangle\,.
$$
As $\widetilde{\pi}_p^2([\partial_z])=[\partial_z]$, we have
$$
    \ker \widetilde{\pi}_p^1=\langle [\partial_z]\rangle\,,\qquad \ker \widetilde{\pi}_p^2=\langle 0\rangle\,.
$$
Hence, $\ker\widetilde{\pi}_p^1\cap \ker \widetilde{\pi}_p^2 = 0$ and condition \eqref{Eq::PolysymplecticReduction2eq} is satisfied. But $\Ima\widetilde{\pi}_p^2=\langle [\partial_z]\rangle$ and $\widetilde{\pi}_p^2$ is not surjective. Thus, condition \eqref{Eq::PolysymplecticReduction1eq} does not hold for $\alpha=2$ in our example. In fact, $\omega^1,\omega^2$ become isotropic when restricted to $\mathbf{J}^{\Phi-1}(\bm\mu)$ and give rise to two zero differential two-forms on $\mathbf{J}^{\Phi-1}(\bm\mu)/G_{\bm\mu}$, which is a one-dimensional manifold. Hence, no two-symplectic manifold is induced on $\mathbf{J}^{\Phi-1}(\bm\mu)/G_{\bm\mu}$ despite that condition \eqref{Eq::PolysymplecticReduction2eq} is satisfied.

One can directly prove that condition \eqref{Eq::PolysymplecticReduction2eq} is satisfied in the previous example, but condition \eqref{Eq::PolysymplecticReduction1eq} is not. This shows more easily that Proposition 1 in \cite{GM_23} is false and  that \eqref{Eq::PolysymplecticReduction2eq} does not imply \eqref{Eq::PolysymplecticReduction1eq}, but our previous approach illustrates how we obtained our counterexample. Indeed, in our present counterexample, the fact that $\widetilde{\pi}_p^2$ is not surjective implies that \eqref{Eq::PolysymplecticReduction1eq} does not hold. Recall that 
$$
    \ker \T_p\mathbf{J}^\Phi_2=\langle \partial_x,\partial_y,\partial_z\rangle\,,\qquad \forall p\in \mathbf{J}^{\Phi-1}(\bm\mu)\,,
$$
while 
$$
    \T_p(\mathbf{J}^{\Phi-1}(\bm\mu)) + \ker\omega^2_p + \T_p(G_{\mu^2}p) = \langle \partial_x,\partial_z\rangle+\{0\}+\langle \partial_x\rangle = \langle \partial_x,\partial_z\rangle \,,\qquad \forall p\in \mathbf{J}^{\Phi-1}(\bm\mu)\,.
$$
On the other hand, condition \eqref{Eq::PolysymplecticReduction2eq} is satisfied since
$$
    \T_p(G_{\bm\mu} p) = \langle \partial_x\rangle
$$
and
$$
    (\ker \omega_p^1 + \T_p(G_{\mu^1}p))\cap(\ker \omega_p^2 + \T_p(G_{\mu^2}p))\cap  \T_p(\mathbf{J}^{\Phi-1}(\bm\mu)) = \langle \partial_x\rangle\,,
$$
reads
$$
(\langle \partial_t,\partial_z\rangle +\langle \partial_x\rangle)\cap (\langle 0 \rangle +\langle \partial_x\rangle)\cap \langle \partial_x,\partial_z\rangle=\langle \partial_x\rangle\,.
$$
Since this example is constructed so as to obtain a one-dimensional reduced manifold, it is known that the reduction of the two-polysymplectic form is not a $k$-polysymplectic form.

The following examples illustrate some relations between the conditions \eqref{Eq::PolysymplecticReduction1eq}, \eqref{Eq::PolysymplecticReduction2eq}
and the existence of $k$-polysymplectic Marsden--Weinstein reductions.

\begin{example} This example shows that if condition \eqref{Eq::PolysymplecticReduction1eq} is satisfied, then condition \eqref{Eq::PolysymplecticReduction2eq} does not need to hold.  
Consider a two-polysymplectic manifold $(\mathbb{R}^6,\boldsymbol{\omega})$. Let  $\{x_1,x_2,x_3,x_4,x_5,x_6\}$ be global linear coordinates on  $\mathbb{R}^6$ and define$$
    \boldsymbol{\omega} = \omega^1\otimes e_1+\omega^2\otimes e_2 = (\d x_1\wedge \d x_2 + \d x_5\wedge \d x_6)\otimes e_1 + (\d x_3\wedge \d x_4 + \d x_5\wedge \d x_6)\otimes e_2\,.
$$
Then, $\ker\omega^1_p=\langle \partial_{3},\partial_4\rangle$, $\ker\omega^2_p=\langle \partial_1,\partial_2\rangle$, and $\ker\omega^1_p\cap\ker\omega^2_p=0$ for every $p\in\mathbb{R}^6$. This turns $\boldsymbol{ \omega}$ into a two-polysymplectic form. 

Let us provide now a Lie group action proving our initial claim. Given the Lie group action $\Phi:(\lambda;x_1,x_2,x_3,x_4,x_5,x_6)\in\mathbb{R}\times \mathbb{R}^6\mapsto (x_1+\lambda,x_2, x_3+\lambda,x_4,x_5,x_6)\in\mathbb{R}^6$, its Lie algebra of fundamental vector fields reads $\langle\partial_1+\partial_3\rangle$. The two-polysymplectic momentum map associated with $\Phi$ is given by
$$
    \mathbf{J}^\Phi:(x_1,x_2,x_3,x_4,x_5,x_6)\in\mathbb{R}^6\longmapsto (x_2,x_4)=\bm\mu\in (\mathbb{R}^*)^2,
$$
which is $\Ad^{*2}$-equivariant. Moreover, every $\bm\mu=(x_2,x_4)\in (\mathbb{R}^*)^2$ is a regular value of $\mathbf{J}^\Phi$. Therefore,  $\mathbf{J}^{\Phi-1}(\bm\mu)=\{(x_1,x_2,x_3,x_4,x_5,x_6)\in\mathbb{R}^6\,:\, x_1,x_3,x_5,x_6\in \mathbb{R}\}\simeq \mathbb{R}^4$ is a submanifold of $\mathbb{R}^6$ for every $\bm\mu\in(\mathbb{R}^*)^2$ and
$$
    \T_p(\mathbf{J}^{\Phi-1}(\bm\mu))=\langle \partial_1,\partial_3,\partial_5,\partial_6\rangle\,,\qquad \forall p\in \mathbf{J}^{\Phi-1}(\bm\mu)\,.
$$
Hence, $\ker \T_p\mathbf{J}^{\Phi}_1=\langle\partial_1,\partial_3,\partial_4,\partial_5,\partial_6\rangle$ while $\ker \T_p\mathbf{J}^{\Phi}_2=\langle\partial_1,\partial_2,\partial_3,\partial_5,\partial_6\rangle$. Condition \eqref{Eq::PolysymplecticReduction1eq} holds because both sides of the condition are equal to
$$
\begin{gathered}
\langle \partial_1,\partial_3,\partial_4,\partial_5,\partial_6\rangle=  \langle \partial_1,\partial_3,\partial_5,\partial_6\rangle+\langle \partial_3,\partial_4\rangle+\langle \partial_1+\partial_3\rangle\,,\\
\langle \partial_1,\partial_2,\partial_3,\partial_5,\partial_6\rangle=  \langle \partial_1,\partial_3,\partial_5,\partial_6\rangle+\langle \partial_1,\partial_2\rangle+\langle \partial_1+\partial_3\rangle\,,
\end{gathered}
$$
for $\mathbf{J}^\Phi_1,\mathbf{J}^\Phi_2$, respectively.
However, condition \eqref{Eq::PolysymplecticReduction2eq} is not satisfied, namely
\begin{multline*}
    \bigcap^{2}_{\alpha=1}\left(\ker\omega^\alpha_p + \T_p(G_{\mu^\alpha}p) \right)\cap \T_p\mathbf{J}^{\Phi-1}(\bm\mu) \\ = \left(\langle \partial_3,\partial_4\rangle +\langle \partial_1+\partial_3\rangle \right)\cap\left(\langle\partial_1,\partial_2\rangle+\langle\partial_1+\partial_3\rangle\right)
    \cap\langle \partial_1,\partial_3,\partial_5,\partial_6\rangle \\ =\langle\partial_1,\partial_3\rangle \neq \langle \partial_1+\partial_3\rangle = \T_p(G_{\bm\mu}p)\,,
\end{multline*}
for any $p\in \mathbf{J}^{\Phi-1}(\bm\mu)$. By \cite[Lemmas 3.15 and 3.16]{MRSV_15}, one has that $\widetilde{\pi}^1_p$ and 
$\widetilde{\pi}^2_p$ are surjective but $\ker\widetilde{\pi}^1_p\cap \ker\widetilde{\pi}^2_p\neq 0$. One can also verify this fact by computing $\widetilde{\pi}_p^\alpha$ for $\alpha=1,2$. Namely, this follows from
\[
\begin{gathered}
\widetilde{\pi}_p^1:\langle[\partial_1],[\partial_5],[\partial_6]\rangle\in \T_p(\mathbf{J}^{\Phi-1}(\bm\mu))/\T_p(G_{\bm\mu} p)\longmapsto \langle [\partial_5],[\partial_6]\rangle=(\ker \T_p\mathbf{J}^\Phi_1/\ker \omega_p^1)/\langle [\partial_1+\partial_3]\rangle\,.\\
\widetilde{\pi}_p^2:\langle[\partial_1],[\partial_5],[\partial_6]\rangle\in \T_p(\mathbf{J}^{\Phi-1}(\bm\mu))/\T_p(G_{\bm\mu} p)\longmapsto \langle [\partial_5],[\partial_6]\rangle=(\ker \T_p\mathbf{J}^\Phi_2/\ker \omega_p^2)/\langle [\partial_1+\partial_3]\rangle\,.
\end{gathered}
\]
for all $p\in\mathbf{J}^{\Phi-1}(\bm\mu)$. Note that $[\partial_1+\partial_3]=[\partial_1]$ in the first line, while $[\partial_1+\partial_3]=[\partial_3]$ in the second.
\demo
\end{example}

\begin{example} Let us prove that the $k$-polysymplectic Marsden--Weinstein reduction theorem in \cite{MRSV_15} gives sufficient, but not necessary conditions for the reduction to hold. In this respect, there are cases where the reduction is possible, condition \eqref{Eq::PolysymplecticReduction2eq} holds, while condition \eqref{Eq::PolysymplecticReduction1eq} does not.
To illustrate this, let us consider a two-polysymplectic manifold $(\mathbb{R}^7,\boldsymbol{\omega})$, where $\{x_1,\ldots,x_7\}$ are global linear coordinates and
\begin{align*}
    \boldsymbol{\omega} &= \omega^1\otimes e_1+ \omega^2\otimes e_2 \\
    &= \left(\d x_1\wedge \d x_2 + \d x_5\wedge \d x_7 + \d x_3\wedge \d x_6 \right)\otimes e_1 + \left(\d x_3\wedge \d x_4 + \d x_5\wedge \d x_6\right)\otimes e_2\,.
\end{align*}
This give rise to a two-polysymplectic structure on $\mathbb{R}^7$ since $\ker\omega^1=\langle \partial_4\rangle$, $\ker\omega^2=\langle\partial_1, \partial_2,\partial_7\rangle$ and $\ker\omega^1\cap\ker\omega^2=0$. Consider the Lie group action $\Phi:\mathbb{R}\times\mathbb{R}^7\rightarrow \mathbb{R}^7$ corresponding to translations along the $x_5$ coordinate. Then, its Lie algebra of fundamental vector fields is $\langle\partial_5\rangle$. A two-polysymplectic momentum map associated with $\Phi$ reads
$$
    \mathbf{J}^\Phi:(x_1,x_2,x_3,x_4,x_5,x_6,x_7)\in\mathbb{R}^7\longmapsto (x_7,x_6)=\bm\mu\in(\mathbb{R}^2)^*\,.
$$
Note that $\mathbf{J}^\Phi$ is $\Ad^{*2}$-equivariant and every $\bm\mu\in \mathbb{R}^{2*}$ is a regular value of $\mathbf{J}^\Phi$. Then, 
$$
    \mathbf{J}^{\Phi-1}(\bm\mu)=\{(x_1,x_2,x_3,x_4,x_5,x_6,x_7)\in\mathbb{R}^7\mid  x_1,x_2,x_3,x_4,x_5\in \mathbb{R}\}\simeq \mathbb{R}^5
$$
is a submanifold of $\mathbb{R}^7$ for every $\bm\mu=(x_7,x_6)\in\mathbb{R}^2$ and
$$
    \T_p(\mathbf{J}^{\Phi-1}(\bm\mu))=\langle \partial_1,\partial_2,\partial_3,\partial_4,\partial_5\rangle\,,\qquad \forall p\in \mathbf{J}^{\Phi-1}(\bm\mu)\,.
$$
Condition \eqref{Eq::PolysymplecticReduction2eq} is satisfied, while \eqref{Eq::PolysymplecticReduction1eq} for $\mathbf{J}^\Phi_1$ is not since
$$
    \widetilde{\pi}_p^1:\langle[\partial_1],[\partial_2],[\partial_3],[\partial_4]\rangle\in \T_p(\mathbf{J}^{\Phi-1}(\bm\mu))/\T_p(G_{\bm\mu} p)\mapsto \langle [\partial_1],[\partial_2],[\partial_3],[\partial_6]\rangle=(\ker \T_p\mathbf{J}^\Phi_1/\ker \omega_p^1)/\langle [\partial_5]\rangle\,.
$$
Therefore, $\widetilde{\pi}_p^1$ is not surjective. However, the reduced manifold $P_{\bm\mu} = \T_p(\mathbf{J}^{\Phi-1}(\bm\mu))/\T_p(G_{\bm\mu}p)\simeq \mathbb{R}^4$ inherits a two-polysymplectic form, namely
\[
\boldsymbol{\omega}_{\bm\mu} = \d x_1\wedge \d x_2\otimes e_1 + \d x_3\wedge \d x_4 \otimes e_2\,,\]
in the variables $x_1,x_2,x_3,x_4$ naturally defined in $P_{\bm\mu}$. 
In summary, both conditions \eqref{Eq::PolysymplecticReduction1eq} and \eqref{Eq::PolysymplecticReduction2eq} ensure a $k$-polysymplectic Marsden--Weinstein reduction. But they are not necessary, they are only sufficient.\demo
\end{example}

\subsection{On the \texorpdfstring{$k$}{}-polysymplectic manifold given by the product of \texorpdfstring{$k$}{} symplectic manifolds}
\label{Sec::SymplProduct}

Let us review a relevant example of $k$-polysymplectic manifold and apply Theorem \ref{Th::PolisymplecticReductionJ} to it (see \cite{MRSV_15} for details). This will illustrate how the $k$-polysymplectic reduction theorem works. Remarkably, many practical examples have a related $k$-polysymplectic manifold similar to the one in this section. Moreover, this structure will be used in one of the physical examples studied in Section \ref{Sec::examples}.

Let $P=P_1\times\dotsb\times P_k$ for some symplectic manifolds $(P_\alpha,\omega^\alpha)$ with $\alpha=1,\ldots,k$. If $\pr_\alpha: P\rightarrow P_\alpha$ is the canonical projection onto the $\alpha$-th component, $P_\alpha$, in $P$, then $(P,\bomega=\sum_{\alpha=1}^k\pr_\alpha^*\omega^\alpha\otimes e_\alpha)$ is a $k$-polysymplectic manifold. To simplify the notation, we will write $\pr_\alpha^\ast\omega^\alpha$ as $\omega^\alpha$. Moreover, assume that a Lie group action $\Phi^\alpha: G_\alpha\x P_\alpha\rightarrow P_\alpha$  admits a symplectic momentum map $\mathbf{J}^{\Phi^\alpha}: P_\alpha\rightarrow\mathfrak{g}^*_\alpha$ and each $\Phi^\alpha$ acts in a quotientable manner on the level sets given by weak regular values of ${\bf J}^{\Phi^\alpha}$ for each $\alpha=1,\ldots,k$.

Define the Lie group action 
\begin{equation}\label{eq:GrAck}
    \Phi:G\x P\ni (g_1,\ldots,g_k,x_1,\ldots,x_k)\longmapsto (\Phi^1_{g_1}(x_1),\ldots,\Phi^k_{g_k}(x_k))\in P\,.
\end{equation}
Then, $\mathfrak{g}=\mathfrak{g}_1\x\dotsb\x\mathfrak{g}_k$ is the Lie algebra of $G$ and we have the $k$-polysymplectic momentum map 
\[
\mathbf{J}:P\ni(x_1,\ldots,x_k)\longmapsto
(0,\ldots, {\bf J}^\alpha,\ldots,0)\otimes {e}_\alpha\in \mathfrak{g}^{*k}\,,
\]
where $\mathbf{J}^\alpha(x_1,\ldots,x_k)=\mathbf{J}^{\Phi^\alpha}(x_\alpha)$ for $\alpha=1,\ldots,k$ and $\mathfrak{g}^*=\mathfrak{g}_1^*\x\dotsb\x\mathfrak{g}_k^*$ is the dual space to $\mathfrak{g}$.
Suppose, that $\mu^\alpha\in\mathfrak{g}^*_\alpha$ is a weak regular value of $\mathbf{J}^{\Phi^\alpha}:P_\alpha\rightarrow\mathfrak{g}^*_\alpha$ for each $\alpha=1,\ldots,k$. Hence, $\bm{\mu}= (0,\ldots,\mu^\alpha,\ldots,0) \otimes e_\alpha\in(\mathfrak{g}^*)^k$ is a weak regular value of $\mathbf{J}$. Then, $\Phi$ acts in a quotientable on the level sets of ${\bf J}$.

Therefore, if $p=(x_1,\ldots,x_k)\in\mathbf{J}^{-1}(\bm\mu)$, it follows that
\begin{align}
    \ker\T_p \mathbf{J}^{\Phi^\alpha} &= \T_{x_1}P_1\oplus\dotsb\oplus \ker \T_{x_\alpha} \mathbf{J}^{\Phi^\alpha}\oplus\dotsb\oplus \T_{x_k}P_k\,,\\
    \T_p\left(\mathbf{J}^{-1}(\bm\mu)\right)&=\ker \T_{x_1}\mathbf{J}^{\Phi^1}\oplus\dotsb\oplus \ker \T_{x_k}\mathbf{J}^{\Phi^k}\,,\\
    \ker \omega_p^\alpha &= \T_{x_1}P_1\oplus\dotsb\oplus \T_{x_{\alpha-1}}P_{\alpha-1}\oplus \{0\}\oplus \T_{x_{\alpha+1}}P_{\alpha+1}\oplus\dotsb\oplus \T_{x_k}P_k\,,\\
    \T_p\left(G^{\Delta^\alpha}_{\mu^\alpha}p\right) &= \T_{x_1}\left(G_1 x_1\right)\oplus\dotsb\oplus \T_{x_\alpha}\left(G^{\Delta^\alpha}_{\alpha\mu^\alpha}x_\alpha\right)\oplus\dotsb\oplus \T_{x_k}\left(G_{k}x_k\right)\,,\\
    \T_p\left(G_{\bm\mu}^{\bm\Delta} p\right) &= \T_{x_1}\left(G^{\Delta^1}_{1\mu^1}x_1\right)\oplus\dotsb\oplus \T_{x_k}\left(G^{\Delta^k}_{k\mu^k}x_k\right)\,.
\end{align}
Then, it follows immediately that
$$
\ker\T_p\mathbf{J}^{\Phi^\alpha} = \T_p\left( \mathbf{J}^{-1}({\bf \mu})\right)+\ker\omega_p^\alpha+ \T_p\left(G^{\Delta^\alpha}_{\mu^\alpha}p\right) \,,\qquad \alpha=1,\ldots,k\,,
$$
and
$$
\T_p\left(G_{\bm\mu}^{\bm\Delta} p\right)=\bigcap^k_{\beta=1}\left(\ker\omega_p^\beta + \T_p\left(G^{\Delta^\beta}_{\mu^\beta}p\right)\right)\cap  \T_p\left(\mathbf{J}^{-1}(\bm\mu)\right)\,,
$$
for every weakly regular $\bm\mu\in (\mathfrak{g}^*)^k$ and $p\in {\bf J}^{-1}({\bm \mu})$. Recall that, by Theorem \ref{Th::PolisymplecticReductionJ}, these equations guarantee that the reduced space $\mathbf{J}^{-1}(\bm\mu)/G^{\bm\Delta}_{\bm\mu}$ can be endowed with a $k$-polysymplectic structure, while
\[
\mathbf{J}^{-1}(\bm\mu)/G^{\bm\Delta}_{\bm\mu}\simeq \mathbf{J}^{\Phi^1-1}(\mu^1)/G^{\Delta^1}_{1\mu^1}\x\dotsb\x \mathbf{J}^{\Phi^k-1}(\mu^k)/G^{\Delta^k}_{k\mu^k}\,.
\]
\demo

\section{The \texorpdfstring{$k$}{}-polysymplectic energy momentum-method}\label{Sec::energy-momentum}

\subsection{\texorpdfstring{$k$}{}-polysymplectic relative equilibrium points}

This section introduces the notion of $k$-polysymplectic relative equilibrium point relative to a $\bomega$-Hamiltonian vector field $X$. This notion is devised to analyse the relative stability of $\bomega$-Hamiltonian vector fields and extends the relative equilibrium point notion for symplectic manifolds to the $k$-symplectic realm (see \cite{AM_78} for details on the symplectic case). In brief, a relative equilibrium point for a dynamical system given by a vector field is a point whose evolution is given by a Lie group symmetry of the vector field. If the vector field is additionally Hamiltonian relative to some geometric structure, then it is usual to demand the Lie group symmetries to leave invariant the same geometric structure \cite{LMZ_23}. It is generally interesting to analyse the behaviour of solutions close to relative equilibrium points, i.e. if they get closer or move away from those points.

\begin{definition}
\label{Def::kEqPoint}
Let $(P,\bm\omega,\bh,\mathbf{J}^\Phi)$ be a $G$-invariant $\bomega$-Hamiltonian system. A point $z_e\in P$ is a {\it $k$-polysymplectic relative equilibrium point} of the $\bomega$-Hamiltonian vector field $X_\bh$ if there exists $\xi\in\mathfrak{g}$ so that
\begin{equation}
(X_{\bm h})(z_e)=(\xi_P)(z_e)\,.
\end{equation}
\end{definition}
The above definition retrieves, for $k=1$, the standard relative equilibrium notion for symplectic Hamiltonian systems \cite{AM_78}. Furthermore, Lemma \ref{Lemm::NonAdPerpPS} and the fact that $X_{\bh}$ is tangent to the level sets of ${\bf J}^{\Phi}$ show that $\xi\in\mathfrak{g}$ in Definition \ref{Def::kEqPoint} is, in fact, an element of $\mathfrak{g}^{\boldsymbol{\Delta}}_{\boldsymbol{\mu}_e}$, which is a Lie subalgebra of $\mathfrak{g}$, and ${\bm\mu_e}={\bf J}^\Phi(z_e)$.

Note that a $k$-polysymplectic relative equilibrium point $z_e\in P$ projects onto $\pi_{\boldsymbol{\mu}_e}(z_e)$, with $\boldsymbol{\mu}_e={\bf J}^\Phi(z_e)$, which becomes an equilibrium point of the vector field $X_{\boldsymbol{f}_{\boldsymbol{\mu}_e}}$, obtained by projection of $X_{\bh}$ onto the reduced space $\mathbf{J}^{\Phi-1}(\boldsymbol{\mu}_e)/G^{\boldsymbol{\Delta}}_{\boldsymbol{\mu}_e}$. This explains the term {\it relative} used in the \textit{relative equilibrium point} term.

The following theorem provides the characterisation of $k$-polysymplectic relative equilibrium points of an $\bomega$-Hamiltonian vector field $X_\bh$ by studying the critical points of a modified $\mathbb{R}^k$-valued function $\boldsymbol{h}_\xi$ on $P$. This is an application of the Lagrange multiplier theorem, where the role of the multiplier is played by $\xi\in\Lg$. 

\begin{theorem}
\label{Th::ksymEM}
    Let $(P,\boldsymbol{\omega},\boldsymbol{h},\mathbf{J}^\Phi)$ be a $G$-invariant $\bomega$-Hamiltonian system. Then, $z_e\in P$ is a $k$-polysymplectic relative equilibrium point of $X_{\boldsymbol{h}}$ if and only if there exists $\xi\in\mathfrak{g}$ such that $z_e$ is a critical point of the following $\mathbb{R}^k$-valued function
    \begin{equation}
    \label{Eq::EMfunction}
    \boldsymbol{h}_\xi:=\boldsymbol{h}-\langle \mathbf{J}^\Phi-\boldsymbol{\mu}_e,\xi \rangle \,,
    \end{equation}
    where $\boldsymbol{\mu}_e:=\mathbf{J}^\Phi(z_e)\in(\mathfrak{g}^*)^k$.
\end{theorem}
\begin{proof}
    Let $z_e$ be a $k$-polysymplectic relative equilibrium point of $X_{\boldsymbol{h}}$, i.e. $X_{\boldsymbol{h}}(z_e)=\xi_P(z_e)$ for some $\xi\in\Lg$. Then,
    \begin{equation}
    \label{Eq::kEMproof}
    \d\boldsymbol{h}_{\xi}(z_e)=\d(\boldsymbol{h}-\langle \mathbf{J}^\Phi,\xi\rangle )(z_e)=(\inn{X_{\bm h}-\xi_{P}}\bomega)(z_e) = 0\,.
    \end{equation}
    Hence, $z_e\in P$ is a critical point of the $\mathbb{R}^k$-valued function $\boldsymbol{h}_\xi$. 

    Conversely, assume that $z_e$ is a critical point of some $\boldsymbol{h}_\xi$ with $\xi\in\mathfrak{g}$. Then,
    $0=\d \boldsymbol{h}_\xi(z_e)=(\iota_{X_\bh-\xi_P}\boldsymbol{\omega})(z_e)=0$
    and $(X_\bh-\xi_P)(z_e)\in \ker\boldsymbol{\omega}_{z_e}$. Since $\ker\boldsymbol{\omega}=0$, one has that $X_\bh(z_e)=\xi_P(z_e)$. Hence, $z_e$ is a $k$-polysymplectic relative equilibrium point of $X_\bh$.
\end{proof}

\begin{example}{(The cotangent bundle of two-covelocities of $\R^2$)} Let $Q$ be an $n$-dimensional manifold  and let $\pi_Q: \cT Q\rightarrow Q$ be the cotangent bundle projection. Consider the Whitney sum $\bigoplus^k\cT Q=\cT Q\oplus_{Q}\overset{(k)}{\dotsb} \oplus_{Q}\cT Q$ of $k$ copies of $\cT Q$ and the projection $\pi^k_Q:\bigoplus^k \cT Q\rightarrow Q$. It is well-known that $\bigoplus^k\cT Q$ can be identified with the space of first-jets, $J^1(Q,\R^k)_0$, of maps $\gamma:Q\rightarrow \R^k$ with $\gamma(q)=0$, via  the diffeomorphism $J^1(Q,\R^k)_0\ni j^1_q\gamma\mapsto \left( \d\gamma^1(q),\ldots, \d\gamma^k(q)\right)\in\bigoplus^k\cT Q$, where 
 $\gamma^\alpha$ is the $\alpha$-th component of $\gamma$ \cite{LSV_15}. 
 Then, $\bigoplus^k \cT Q$ is called {\it the cotangent bundle of $k$-covelocities of $Q$}. Moreover, $J^1(Q,\R^k)_0$ is a $k$-polysymplectic manifold (see \cite{LSV_15} for details).

The following example will illustrate our $k$-polysymplectic energy-momentum method. 
Consider that $(\mathbb{R}^6,\boldsymbol{\omega})$ is a two-polysymplectic manifold with the two-polysymplectic form
\[
\bomega=\omega^1\otimes e_1+\omega^2\otimes e_2=\left(\d x_1\wedge \d x_3 + \d x_2\wedge \d x_4\right)\otimes e_1+\left(\d x_1\wedge \d x_5 + \d x_2\wedge \d x_6\right)\otimes e_2\,
\]
since
\[
\ker\omega^1=\left\langle \frac{\partial}{\partial x_5},\frac{\partial}{\partial x_6}\right\rangle \,,\qquad\ker\omega^2=\left\langle \frac{\partial}{\partial x_3},\frac{\partial}{\partial x_4}\right\rangle \,,
\]
and $\ker\omega^1\cap\ker\omega^2=0$.
Let us consider the Lie group action $\Phi:\R\times \R^6\rightarrow \R^6$ given by
\[
\Phi:\mathbb{R}\times\mathbb{R}^6\ni(\lambda;x_1,x_2,x_3,x_4,x_5,x_6)\longmapsto (x_1+\lambda,x_2+\lambda,x_3+\lambda, x_4+\lambda,x_5+\lambda,x_6+\lambda)\in\mathbb{R}^6\,.
\]
The fundamental vector fields associated with the Lie group action $\Phi$ are spanned by
\[
    \xi_P=\frac{\partial}{\partial x_1}+\frac{\partial}{\partial x_2}+\frac{\partial}{\partial x_3}+\frac{\partial}{\partial x_4}+\frac{\partial}{\partial x_5}+\frac{\partial}{\partial x_6}\,.
\]
Note that the Lie group action $\Phi$ is two-polysymplectic since it leaves $\bomega$ invariant, namely $\Lie_{\xi_P}\bomega = 0$. Then, $\Phi$ gives rise to a two-polysymplectic momentum map $\mathbf{J}^\Phi$ for $\bm \mu=(\mu^1,\mu^2)$ given by
\[
    \mathbf{J}^\Phi:\mathbb{R}^6\ni(x_1,x_2,x_3,x_4,x_5,x_6)\longmapsto (x_3+x_4-x_1-x_2,x_5+x_6-x_1-x_2)=\boldsymbol{\mu}\in (\mathbb{R}^*)^2\,.
\]
Therefore, the level set of the two-polysymplectic momentum map $\mathbf{J}^\Phi$ has the following form 
\begin{equation}
\label{Eq::ExMomentumMap1}
\mathbf{J}^{\Phi-1}(\boldsymbol{\mu})=\left\{ (x_1,x_2,x_3,x_4,x_5,x_6)\!\in\!\mathbb{R}^6\mid x_3\!+\!x_4\!-\!x_1\!-\!x_2=\mu^1,\,\,\,x_5\!+\!x_6\!-\!x_1\!-\!x_2\!=\mu^2\right\}.
\end{equation}
Note that every $\boldsymbol{\mu}\in(\mathbb{R}^*)^2$ is a regular value of a two-polysymplectic momentum map $\mathbf{J}^\Phi$ and  $\mathbf{J}^{\Phi-1}(\boldsymbol{\mu})\simeq \mathbb{R}^4$. Since $\Phi$ is defined on a connected one-dimensional Lie group $\mathbb{R}=G$, one has that $\mathbf{J}^\Phi$ is an $\Ad^{*2}$-equivariant two-polysymplectic momentum map. Then,
\begin{gather*}
\T_p(G_{\boldsymbol{\mu}}p) = \T_p(G_{\mu^1}p) = \T_p(G_{\mu^2}p)=\left\langle \frac{\partial}{\partial x_1}+\frac{\partial}{\partial x_2}+\frac{\partial}{\partial x_3}+\frac{\partial}{\partial x_4}+\frac{\partial}{\partial x_5}+\frac{\partial}{\partial x_6} \right\rangle \,,\\
\ker \T_p\mathbf{J}^{\Phi}_1=\left\langle \frac{\partial}{\partial x_2}-\frac{\partial}{\partial x_1},\frac{\partial}{\partial x_1}+\frac{\partial}{\partial x_3},\frac{\partial}{\partial x_1}+\frac{\partial}{\partial x_4},\frac{\partial}{\partial x_5},\frac{\partial}{\partial x_6} \right\rangle \,,\\
\ker \T_p\mathbf{J}^{\Phi}_{2}=\left\langle  \frac{\partial}{\partial x_2}-\frac{\partial}{\partial x_1},\frac{\partial}{\partial x_3},\frac{\partial}{\partial x_4},\frac{\partial}{\partial x_1}+\frac{\partial}{\partial x_5},\frac{\partial}{\partial x_1}+\frac{\partial}{\partial x_6}\right\rangle \,,\\
\T_p(\mathbf{J}^{\Phi-1}(\boldsymbol{\mu}))=\left\langle  \sum^6_{i=1}\frac{\partial}{\partial x_i},\frac{\partial}{\partial x_3}-\frac{\partial}{\partial x_4},\frac{\partial}{\partial x_2}+\frac{\partial}{\partial x_3}+\frac{\partial}{\partial x_5},\frac{\partial}{\partial x_1}-\frac{\partial}{\partial x_2}\right\rangle \,,
\end{gather*}
and one can verify that conditions \eqref{Eq::PolysymplecticReduction1eq} and \eqref{Eq::PolysymplecticReduction2eq} are fulfilled.

Recall that $\iota_{\bm\mu}:\mathbf{J}^{\Phi-1}(\boldsymbol{\mu})\hookrightarrow P$ is the natural immersion and $\pi_{\boldsymbol{\mu}}:\mathbf{J}^{\Phi-1}(\boldsymbol{\mu})\rightarrow \mathbf{J}^{\Phi-1}(\boldsymbol{\mu})/G_{\boldsymbol{\mu}}$ is the canonical projection. Then, remembering that the elements of the Lie group $\mathbb{R}$ act by translations on $\mathbb{R}^6$ via $\Phi$, Theorem \ref{Th::PolisymplecticReductionJ} yields that the reduced manifold $(\mathbf{J}^{\Phi-1}(\boldsymbol{\mu})/G_{\boldsymbol{\mu}}\simeq \mathbb{R}^3,\boldsymbol{\omega}_{\boldsymbol{\mu}})$ is a two-polysymplectic manifold with coordinates $(y_1,y_2,y_3)\in\mathbb{R}^3$, satisfying that
\begin{align*}
y_1&=x_1-x_2\,,& y_2&=x_3-x_1\,,& y_3&=x_5-x_1\,,\\ y_4&=x_1+x_2-x_3-x_4\,,&
y_5&=x_1+x_2-x_5-x_6\,,
& y_6&=x_1\,,
\end{align*}
with
\[
\boldsymbol{\omega}_{\boldsymbol{\mu}}=\omega^1_{\boldsymbol{\mu}}\otimes e_1 + \omega^2_{\boldsymbol{\mu}}\otimes e_2=\d y_1\wedge \d y_2\otimes e_1+\d y_1\wedge \d y_3 \otimes e_2\,.
\]
Next, let us consider an $\bm{\omega}$-Hamiltonian vector field, $X_{\bm{h}}$, on $P=\mathbb{R}^6$ whose $\bm\omega$-Hamiltonian function is $\mathbb{R}$-invariant. Then, $X_{\bm{h}}$ is tangent to each $\mathbf{J}^{\Phi-1}(\boldsymbol{\mu})$, and it will have the following form
\[
X_{\bh}=F_1\sum^6_{i=1}\frac{\partial}{\partial x_i}+F_2\left(\frac{\partial}{\partial x_3}-\frac{\partial}{\partial x_4}\right)+F_3\left(\frac{\partial}{\partial x_2}+\frac{\partial}{\partial x_3}+\frac{\partial}{\partial x_5}\right)+F_4\left(\frac{\partial}{\partial x_1}-\frac{\partial}{\partial x_2}\right),
\]
for certain uniquely defined $G$-invariant $F_1,\ldots,F_4\in \Cinfty(P)$. Then, a point $z_e\in P$ is a two-polysymplectic relative equilibrium point of $X_\bh$ if and only if $X_\bh(z_e)=\xi_P(z_e)$, which holds, if and only if, $F_1(z_e)=1$ and $F_2(z_e)=F_3(z_e)=F_4(z_e)=0$. However, let us verify that we obtain the same result using Theorem \ref{Th::ksymEM}. 

First, $\d h^1$ and $\d h^2$ read
\begin{align*}
\d h^1 &= \inn{X_{\bm h}}\omega^1=-\left(F_1+F_2+F_3\right)\d x_1-\left(F_1-F_2 \right)\d x_2+\left(F_1+F_4 \right)\d x_3+\left(F_1+F_3-F_4\right)\d x_4\,,\\
\d h^2 &= \inn{X_{\bm h}}\omega^2=-\left( F_1+F_3\right)\d x_1-F_1\d x_2+\left(F_1+F_4\right)\d x_5+\left(F_1+F_3-F_4 \right)\d x_6\,.
\end{align*}
Then, Theorem \ref{Th::ksymEM} yields that $z_e\in P$ is a two-polysymplectic relative equilibrium point of $X_\bh$ if and only if $\d h^1_\xi(z_e)=0$ and $\d h^2_{\xi}(z_e)=0$ for some $\xi\in \mathbb{R}$. Indeed, using \eqref{Eq::ExMomentumMap1}, one has
\begin{multline}
\label{Eq::ExEMf1}
    \d h_\xi^1 = \d h^1 - \d J^1_\xi=-\left(F_1+F_2+F_3-\xi\right)\d x_1-\left(F_1-F_2-\xi\right)\d x_2\\+\left(F_1+F_4-\xi\right)\d x_3+\left(F_1+F_3-F_4-\xi\right)\d x_4\,,
\end{multline}
\begin{multline}
\label{Eq::ExEMf2}
\d h_\xi^2=\d h^2-\d J^2_\xi=-\left(F_1+F_3-\xi\right)\d x_1-\left(F_1-\xi\right)\d x_2\\+\left( F_1+F_4-\xi\right)\d x_5+\left( F_1+F_3-F_4-\xi\right)\d x_6\,,
\end{multline}
for $\xi\in\R$. Since at $z_e$ both \eqref{Eq::ExEMf1} and \eqref{Eq::ExEMf2} must vanish, one gets that this happens if and only if $F_1(z_e)=\xi$ and $F_2(z_e)=F_3(z_e)=F_4(z_e)=0$. Therefore, $z_e\in P$ is a two-polysymplectic relative equilibrium point of $X_\bh$ under the above-mentioned conditions. 

Finally, let us verify that $\pi_{\boldsymbol{\mu}_e}(z_e)$ is a critical point of the $f^\alpha_{\boldsymbol{\mu}_e}\in \Cinfty(\mathbf{J}^{\Phi-1}(\boldsymbol{\mu}_e)/G_{\boldsymbol{\mu}_e})$. The reduced vector field $X_{\boldsymbol{f}_{\boldsymbol{\mu}_e}}$ has the form
\[
X_{\boldsymbol{f}_{\bm\mu_{e}}}=(2\widetilde{F}_4-\widetilde{F}_3)\frac{\partial}{\partial y_1}+(\widetilde{F}_2+\widetilde{F}_3-\widetilde{F}_4)\frac{\partial}{\partial y_2}-\widetilde{F}_4\frac{\partial}{\partial y_3}\,,
\]
where $F_i=\pi_{\boldsymbol{\mu}_e}^*\widetilde{F}_i$ for $i=2,3,4$. Note that the projection exists because $F_2, F_3, F_4$ are $G$-invariant. 
Then,
\begin{multline*}
    \d {f}^1_{\boldsymbol{\mu}_e}(\pi_{\boldsymbol{\mu}_e}(z_e)) = \left(\inn{X_{\boldsymbol{f}_{\boldsymbol{\mu}_e}}}\omega^1_{\boldsymbol{\mu}_e}\right)_{\pi_{\boldsymbol{\mu}_e}(z_e)}=\\=\left(2\widetilde{F}_4(\pi_{\boldsymbol{\mu}_e}(z_e))-\widetilde{F}_3(\pi_{\boldsymbol{\mu}_e}(z_e))\right)\d y_2+\left(\widetilde{F}_4(\pi_{\boldsymbol{\mu}_e}(z_e))-\widetilde{F}_2(\pi_{\boldsymbol{\mu}_e}(z_e))-\widetilde{F}_3(\pi_{\boldsymbol{\mu}_e}(z_e))\right)\d y_1=0\,,
\end{multline*}
and
\[
    \d {f}^2_{\boldsymbol{\mu}_e}(\pi_{\boldsymbol{\mu}_e}(z_e)) = \left(\inn{X_{\boldsymbol{f}_{\boldsymbol{\mu}_e}}}\omega^2_{\boldsymbol{\mu}_e}\right)_{\pi_{\boldsymbol{\mu}_e}(z_e)}=\left(2\widetilde{F}_4(\pi_{\boldsymbol{\mu}_e}(z_e))-\widetilde{F}_3(\pi_{\boldsymbol{\mu}_e}(z_e))\right)\d y_3+\widetilde{F}_4(\pi_{\boldsymbol{\mu}_e}(z_e))\d y_1=0\,.
\]
Indeed, $\pi_{\boldsymbol{\mu}_e}(z_e)$ is a critical point of $\boldsymbol{f}_{\bm\mu_e}$, hence $z_e\in P$ is a $k$-polysymplectic relative equilibrium point of $X_{\boldsymbol{f}_{\bm\mu_e}}$.
\end{example}

\subsection{Stability in the \texorpdfstring{$k$}{}-polysymplectic energy momentum-method}

Let us develop the stability analysis related to the $k$-polysymplectic energy-momentum method relative to a $k$-polysymplectic manifold $(P,\boldsymbol{\omega})$. Recall that Theorem \ref{Th::ksymEM} characterises $k$-polysymplectic relative equilibrium points as critical points of the $\R^k$-valued function \eqref{Eq::EMfunction}. However, when studying the stability of $k$-polysymplectic relative equilibrium points, due to the symmetry of our problems, 
we need to investigate how the second variation of $\boldsymbol{h}_{\xi}$ in the directions tangent to the isotropy group $G^{\boldsymbol{\Delta}}_{\boldsymbol{\mu}_e}$ affects the positive definiteness of $\boldsymbol{h}_\xi$. Note also that the results of this section are concerned with cases when a $k$-polysymplectic reduction is possible and \eqref{Eq::BlackerCondition} is satisfied.

 Let us define the second variation of $\boldsymbol{h}_{\xi}$ at a $k$-polysymplectic relative equilibrium point $z_e\in {\bf J}^{\Phi-1}({\bm \mu}_e)$ as the mapping
$(\delta^2 {\bm h}_\xi)_{z_e}: \T_{z_e}({\bf J}^{\Phi-1}(\bm \mu_e))\times \T_{z_e}({\bf J}^{\Phi-1}(\bm \mu_e))\rightarrow \mathbb{R}$, with $\bm\mu_e={\bf J}^{\Phi}(z_e)$, of the form
\begin{equation}
\label{Eq::SecVarf}
\left(\delta^2 \boldsymbol{h}_\xi\right)_{z_e}(v_1,v_2)=\sum^k_{\alpha=1}\inn{Y}\left(\d\left(\inn{X}\d h^\alpha_\xi\right)\right)_{z_e}\otimes e_\alpha\,,
\end{equation}
for some vector fields $X,Y$ on $P$ defined on a neighbourhood of $z_e\in P$ and such that $v_1=X_{z_e}$, $v_2=Y_{z_e}$. The following proposition shows that, since $z_e$ is a $k$-polysymplectic relative equilibrium point, the above definition does not depend on the value of the particular chosen vector fields $X, Y$ out of $z_e$ and $(\delta^2\boldsymbol{h}_\xi)_{z_e}$ is well-defined.

\begin{proposition}
Let $z_e\in P$ be a $k$-polysymplectic relative equilibrium point of $X_\bh$ on a $k$-polysymplectic manifold $(P,\boldsymbol{\omega})$. If $\{x_1,\ldots,x_{n}\}$ are coordinates on a neighbourhood of $z_e\in P$, then
\begin{equation}
\label{Eq::SecVarFormh}
(\delta^2h_{\xi}^\alpha)_{z_e}(w,v)=\sum^{n}_{i,j=1}\frac{\partial^2h^\alpha_{\xi}}{\partial x_i\partial x_j}(z_e)w_iv_j\,,\qquad \forall w,v\in \T_{z_e}({\bf J}^{\Phi-1}(\bm\mu_e))\,,\qquad \alpha=1,\ldots,k\,,
\end{equation}
where $w=\sum_{i=1}^{n}w_i\partial/\partial x_i$ and $v=\sum_{i=1}^{n}v_i\partial/\partial x_i$.
\end{proposition}
\begin{proof}
From \eqref{Eq::SecVarf} for $\alpha=1,\ldots,k$, we have
\begin{align*}
(\delta^2h^\alpha_{\xi})_{z_e}(w,v) &= \iota_Y(\d\iota_{X}\d  h^\alpha_{\xi})_{z_e}\\
&= \sum^{n}_{i,j=1}\frac{\partial^2h^\alpha_{\xi}}{\partial x_i\partial x_j}(z_e)w_iv_j+\sum^{n}_{i,j=1}\frac{\partial h^\alpha_{\xi}}{\partial x_i}(z_e)\frac{\partial X_i}{\partial x_j}(z_e)v_j\\
&= \sum^{n}_{i,j=1}\frac{\partial^2h^\alpha_{\xi}}{\partial x_i\partial x_j}(z_e)w_iv_j\,,
\end{align*}
where $X=\sum_{i=1}^{n}X_i\partial/\partial x_i$ with $X(z_e)=w$, and we have used that $z_e$ is a $k$-polysymplectic relative equilibrium point and, therefore, ${\bm h}_\xi$ has a critical point at $z_e$, namely $(Zh^\alpha_\xi)(z_e)=0$ for every vector field $Z$ on $P$ and $\alpha=1,\ldots,k$.
\end{proof}
Note that the maps $(\delta^2h_\xi^\alpha)_{z_e}$ are symmetric for $\alpha=1,\ldots,k$. Therefore, $(\delta^2\bm h_\xi)_{z_e}$ is symmetric. Let us study \eqref{Eq::SecVarf} in more detail.

\begin{proposition}
\label{Prop::GaugeDir}
     Let $(P,\bomega,\bh,\mathbf{J}^\Phi)$ be a $G$-invariant $\bomega$-Hamiltonian system and let $z_e\in P$ be a $k$-polysymplectic relative equilibrium point of $X_{\bh}$. Then,
    \[ (\delta^2\boldsymbol{h}_\xi)_{z_e}((\zeta_P)_{z_e},v_{z_e})=0\,,\qquad \forall \zeta\in \mathfrak{g}^{\boldsymbol{\Delta}}_{\boldsymbol{\mu}_e}\,,\qquad \forall v_{z_e}\in \T_{z_e}({\bf J}^{\Phi-1}({\bm \mu}_e))\,,
    \]
    with ${\bm \mu}_e={\bf J}^{\Phi}(z_e)$. Moreover,
     \begin{equation}\label{eq:SecDeg} (\delta^2{h}^\alpha_\xi)_{z_e}(Y_{z_e},\cdot)=0\,,\qquad \forall \,Y_{z_e}\in \ker \omega_{z_e}^\alpha\cap \T_{z_e}({\bf J}^{\Phi-1}(\bm\mu_e))\,,\qquad \alpha=1,\ldots,k\,.
    \end{equation}
\end{proposition}
\begin{proof}
First, since $\bh\in \Cinfty(P,\R^k)$ is $G$-invariant and $\mathbf{J}^\Phi$ is equivariant with respect to the $k$-polysymplectic affine Lie group action $\boldsymbol{\Delta}:G \times (\mathfrak{g}^*)^k\rightarrow (\mathfrak{g}^*)^k$, then for every $g\in G$ and $p\in P$, one has
\begin{multline*}
\boldsymbol{h}_\xi(\Phi_g(p)) = \boldsymbol{h}(\Phi_g(p))-\langle \mathbf{J}^\Phi(\Phi_g(p)),\xi\rangle +\langle {\bm \mu}_e,\xi\rangle \\
= \boldsymbol{h}(p)-\langle \boldsymbol{\Delta}_g\mathbf{J}^\Phi(p),\xi\rangle +\langle \bm\mu_e,\xi\rangle 
=\boldsymbol{h}(p)-\sum_{\alpha=1}^k\langle \mathbf{J}_\alpha^\Phi(p),\Delta_{g\alpha}^T\xi\rangle \otimes e_\alpha+\langle \bm\mu_e,\xi\rangle,
\end{multline*}
where $\boldsymbol{\Delta}^T_g:\mathfrak{g}^k\rightarrow \mathfrak{g}^k$ is the transpose of $\boldsymbol{\Delta}_g$ for $g\in G$ and $\Delta_{g1},\ldots,\Delta_{gk}$  are its components. Let us substitute $g=\exp(t\zeta)$, with $\zeta\in\mathfrak{g}$, and differentiate with respect to $t$. Then,
\begin{equation}
\label{Eq::1stvar}
\left(\inn{\zeta_P}\d\boldsymbol{h}_\xi\right)_{z_e}=-\sum_{\alpha=1}^k\left\langle \mathbf{J}_\alpha^\Phi(p),\restr{\frac{\d}{\d t}}{t=0}{\Delta}^T_{\exp{(t\zeta)\alpha}}\xi \right\rangle\otimes e_\alpha=-\sum_{\alpha=1}^k\left\langle \mathbf{J}_\alpha^\Phi(p),(\zeta^{{\Delta}_\alpha}_{\mathfrak{g}})_{\xi}\right\rangle \otimes e_\alpha\,,
\end{equation}
where $(\zeta^{{\Delta}_\alpha}_{\mathfrak{g}})_{\xi}$ is the fundamental vector field of ${\Delta}_\alpha^T:G\times\mathfrak{g}\rightarrow \mathfrak{g}$ at $\xi\in\mathfrak{g}$ for $\alpha=1,\ldots,k$. Taking the second variation of \eqref{Eq::1stvar} relative to $p\in P$, evaluating at $z_e\in P$, and contracting with $v_{z_e}$, one has
\[
\left(\delta^2 \boldsymbol{h}_\xi \right)_{z_e}((\zeta_P)_{z_e},v_{z_e})=-\sum_{\alpha=1}^k\left\langle \T_{z_e}\mathbf{J}_\alpha^\Phi\left(v_{z_e}\right), (\zeta^{{\Delta}_\alpha}_{\mathfrak{g}})_{\xi}\right\rangle\otimes e_\alpha .
\]
Therefore, the second variation $\left(\delta^2 \boldsymbol{h}_\xi \right)_{z_e}((\zeta_P)_{z_e},v_{z_e})$ vanishes since $v_{z_e}\in \T_{z_e}(\mathbf{J}^{\Phi-1}(\boldsymbol{\mu}_e))\subset \ker \T_{z_e}\mathbf{J}^{\Phi}_\alpha$.

Concerning \eqref{eq:SecDeg}, it is a consequence of \eqref{Eq::SecVarf} and the fact that, for every vector field $Y$ on ${\bf J}^{\Phi-1}(\bm\mu_e)$ taking values in $\ker\omega^\alpha\cap \T({\bf J}^{\Phi-1}(\bm\mu_e))$, it follows that
$$
\iota_Y\d h^\alpha=\omega^\alpha(X_{\bm h},Y)=0\,,\qquad \iota_Y\d \langle {\bf J}_\alpha^\Phi,\xi\rangle =\omega^\alpha(\xi_P,Y)=0\,,
$$
for $\alpha=1,\ldots,k$ and every $\xi\in \mathfrak{g}$ on ${\bf J}^{\Phi-1}(\bm\mu_e)$. 
\end{proof}

Proposition \ref{Prop::GaugeDir} and Proposition \ref{Lemm::NonAdPerpPS} state that $\left(\delta^2 \boldsymbol{h}_{\xi} \right)_{z_e}$ is degenerate in the directions tangent to $\T_{z_e}\left(G^{\boldsymbol{\Delta}}_{\boldsymbol{\mu}_e}z_e\right)$, while each $(\delta^2{ h}^\alpha_\xi)_{z_e}$ is degenerate in the directions of $\ker\omega^\alpha_{z_e}\cap \T_{z_e}({\bf J}^{\Phi-1}(\bm\mu_e))$. On the other hand, since $\ker (\delta^2{\bm h}_{\xi})_{z_e}$ contains $\ker \T_{z_e}\pi_{\bm \mu_{e}}$, one can define a bilinear two-form on $\T_{\pi_{\bm \mu_{e}}(z_e)}P_{\bm \mu_{e}}$, with $P_{\bm \mu_{e}}={\bf J}^{\Phi-1}({\bm \mu}_{e})/G^{\bm\Delta}_{\bm\mu_{e}}$, by reducing to that space  the bilinear two-form $(\delta^2{\bm h}_{\xi})_{z_e}$. By using an adapted coordinate system, one can prove that the reduction of $(\delta^2{\bm h}_{\xi})_{z_e}$ to $\T_{\pi_{\bm \mu_{e}}(z_e)}P_{\bm \mu_{e}}$ gives the behaviour of the Hessian of ${\bm f}_{\bm\mu_{e}}$ on $P_{\bm \mu_{e}}$. It is worth noting that the reduction ${\bm f}_{\mu_{e}}$ to $P_{\bm\mu}$ of ${\bm h}_\xi$ on ${\bf J}^{\Phi-1}(\bm\mu_e)$ does not depend on $\xi$, as the value of ${\bm h}_\xi$ on points of ${\bf J}^{\Phi-1}(\bm \mu_{e})$ does not really depend on $\xi$: it is only the restriction of $\bm h$ to ${\bf J}^{\Phi-1}(\bm \mu_{e})$. Note also that only directions transverse to the orbit of $G^{\boldsymbol{\Delta}}_{\boldsymbol{\mu}_e}$ are significant for determining, via the variation of ${\bm h}_\xi$, the stability character of ${\bm f}_{\bm\mu_{e}}$ at one of its equilibrium points.  

There are many manners to ensure the stability of a $k$-polysymplectic reduced Hamiltonian system. This suggests us to give the following definition of formal stability. For the case of a symplectic manifold, it retrieves the standard condition for the stability of a reduced symplectic problem \cite{MS_88}. 

\begin{definition}
\label{Prop::SecDerivEM}
Let $(P,\bomega,\bh,\mathbf{J}^\Phi)$ be a $G$-invariant $\bomega$-Hamiltonian system and let $z_e\in P$ be a $k$-polysymplectic relative equilibrium point of  $X_\bh$. Then, $z_e$ is called a {\it formally  stable $k$-polysymplectic relative equilibrium point} if, for a family of supplementary spaces $\mathcal{S}^\alpha$ such that $\mathcal{S}^\alpha\oplus( \T_{z_e}(G^{\boldsymbol{\Delta}}_{\boldsymbol{\mu}_e}z_e) + \ker\omega^\alpha_{z_e}\cap \T_{z_e}(\mathbf{J}^{\Phi-1}(\boldsymbol{\mu}_e)))=\T_{z_e}(\mathbf{J}^{\Phi-1}(\boldsymbol{\mu}_e))$ and $\mathcal{S}^1+\dotsb+\mathcal{S}^k+\T_{z_e}(G^{\bm\Delta}_{\bm\mu_e}z_e)=\T_{z_e}({\bf J}^{\Phi-1}(\bm\mu_e))$, one has
\begin{equation}\label{eq:FS}
    \left(\delta^2 { h}^\alpha_\xi\right)_{z_e}(v_{z_e},v_{z_e}) > 0\,,\qquad \forall v_{z_e}\in \mathcal{S}^\alpha\backslash\{0\}\,,\qquad \alpha=1,\ldots,k\,.
\end{equation}
\end{definition}

Note that, given a space of a family $W_1,\ldots,W_k$ of a vector space $E$ such that $\cap_{\alpha=1}^kW_\alpha=0$, one cannot infer that any supplementary spaces $V_\alpha\oplus W_\alpha=E$ will satisfy $V_1+\dotsb+V_k=E$. This is, essentially, why the condition $\mathcal{S}^1+\dotsb+\mathcal{S}^k+\T_{z_e}(G^\Delta_{\bm \mu _e})$ was added. Indeed, to ensure the stability on the reduced manifold, we will use the fact that the projection of $\mathcal{S}^1+\dotsb+\mathcal{S}^k$ to the tangent space to an equilibrium point in a reduced manifold spans the total tangent space at that point.

If a system satisfies our formal stability, then $\sum_{\alpha=1}^kf^\alpha_{\bm \mu_e}$ has a strict minimum at $\pi_{\bm\mu_e}(z_e)$ and the function is invariant relative to the evolution of the reduced $\bomega_{\bm\mu}$-Hamiltonian system. Hence, that system is stable. The converse is not true, as in the symplectic case \cite{AM_78}. We will not study all methods to prove stability in the reduced $k$-symplectic Hamiltonian system in this paper, and we will leave this for further work. 

The proof of the above-mentioned fact relies on using a coordinate system on ${\bf J}^{\Phi-1}(\bm \mu_e)$ adapted to its fibration onto $P_{\bm\mu_e}$ and the fact that the obtained results involve geometric objects that are independent of the coordinate system (see \cite{LZ_21,Zaw_21} for a symplectic analogue). In the adapted coordinate system, the Hessian of $\boldsymbol{f}_{\bm\mu_e}$ on the reduced space $P_{\bm\mu_e}$ at $\pi_{\bm\mu_e}(z_e)$ is retrieved by the Hessian of ${\bh}_\xi$  on directions of $\T_{z_e}({\bf J}^{\Phi-1}(\bm\mu_e))$ that are not tangent to $\ker \T_{z_e}\pi_{\bm \mu_e}$. The Hessian of the reduced function $\boldsymbol{f}_{\bm \mu_e}$ can be decomposed into $k$ components. The vector subspaces $\mathcal{S}^1,\ldots,\mathcal{S}^k$ project onto a series of spaces spanning $\T_{\pi_{\bm\mu_e}(z_e)}P_{\bm \mu_e}$. Condition \eqref{eq:FS} implies that 
\begin{align}
\frac{\partial^2f^\alpha_{\bm\mu_e}}{\partial z_i\partial z_j}(\pi_{\bm \mu_e}(z_e))v^iv^j &> 0\,,\qquad \forall v\in {\rm Im}\T_{\pi_{\bm\mu_e}(z_e)}\pi_{\bm\mu_e}(\mathcal{S}^\alpha)\backslash\{0\}\,,\\ \frac{\partial^2f^\alpha_{\bm\mu_e}}{\partial z_i\partial z_j}(\pi_{\bm \mu_e}(z_e))v^iv^j &\geq  0\,,\qquad \forall v\in \T_{\pi_{\bm\mu_e}(z_e)}P_{\bm\mu_e}\,,
\end{align}
for $\alpha=1,\ldots,k$.  Then, 
$$
\sum_{\alpha=1}^k\frac{\partial^2f^\alpha_{\bm\mu_e}}{\partial z_i\partial z_j}(\pi_{\bm \mu_e}(z_e))v^iv^j>0\,,\qquad \forall v\in \T_{\pi_{\bm\mu_e}(z_e)}P_{\bm\mu_e}\backslash\{0\}\,.
$$
Consequently, the second-order Taylor part of $\sum_{\alpha=1}^kf_{\bm\mu_e}^\alpha$ is definite-positive and we have a strict minimum. The components $f^\alpha_{\bm \mu_e}$ are constants of motion for $X_{\boldsymbol{f}_{\bm \mu_e}}$, and hence the motion of $X_{\boldsymbol{f}_{\bm \mu_e}}$, for an initial condition close enough to $\pi_{\bm\mu_e}(z_e)$ can be restricted to an open neighbourhood of $\pi_{\bm\mu_e}(z_e)$.

It is worth noting that we will also call {\it formally stable $k$-polysymplectic relative equilibrium points} points for which each \eqref{eq:FS} is negative-definite, as similar results can be obtained. In particular, their projections will be stable equilibrium points. It is simple to obtain many more stability criteria.

\section{Applications and examples}\label{Sec::examples}

This section illustrates how the theory and applications of the previous sections can be applied to relevant examples with physical and mathematical applications.

\subsection{Complex Schwarz equations}\label{Sec:ComScEq}
The first example illustrates how locally automorphic Lie systems \cite{GLMV_19} can be seen as a $\bomega$-Hamiltonian system relative to a $k$-polysymplectic structure. 

Consider the differential equation the $t$-dependent complex differential equation given by
\begin{equation}\label{Eq:ComplSchwa}
    \frac{\d z}{\d t} = v\,,\qquad \frac{\d v}{\d t} = a\,,\qquad\frac{\d a}{\d t} = \frac 32\frac{a^2}{v} + 2b(t)v\,,\qquad z,v,a\in \mathbb{C}\,,
\end{equation}
for a certain complex $t$-dependent function $b(t)$ defined on $\mathcal{O}=\{(z,v,a)
\in \T^2\mathbb{C}:v
\neq 0\}$, which can be considered as a system of real differential equations in a natural manner.

The system \eqref{Eq:ComplSchwa} can be understood as the complex analogue of the Lie system on $\mathcal{O}_\mathbb{R}=\{(z,v, a)
\in \T^2\mathbb{R}:v
\neq 0\}$ studied in \cite{LV_15}.  More specifically, \eqref{Eq:ComplSchwa} is a first-order representation for the third-order complex differential equation
$$
\frac{\d^3z}{\d t^3}\left(\frac{\d z}{\d t}\right)^{-1}-\frac{3}{2}\left(\frac{\d^2 z}{\d t^2}\right)^2\left(\frac{\d z}{\d t}\right)^{-2}=2b(t)\,.
$$
The left-hand side of the above expression retrieves, for $z\in \mathbb{R}$, exactly the real version of the {\it Schwarzian derivative} (also called {\it Schwarz equation}) of a function $z(t)$ of $t$, usually represented by  $\{z(t),t\}_{sc}$, which appears in many research problems. The ideas in our work and \eqref{Eq:ComplSchwa} can be used to potentially extend to the complex realm results for the real third-order Kummer--Schwarz equation and Schwarz derivatives obtained via Lie systems (see \cite{BC_20,LS_20} and references therein). It is worth noting that the Schwarz derivative plays a significant role in studying linearisation in time-dependent systems, projective systems, mathematical functions theory, and more (cf. \cite{GR_07, Hil_97, Leh_79}).  

In real coordinates
\[
v_1 = \Rp (z)\,,\quad v_2 = \Ip(z)\,,\quad v_3 = \Rp(v)\,,\quad v_4 = \Ip(v)\,,\quad v_5 = \Rp(a)\,,\quad v_6 = \Ip(a)\,,
\]
the system \eqref{Eq:ComplSchwa} is associated with the $t$-dependent vector field
$$
X=X_1+2b_R(t)X_2+2b_I(t)X_3\,,
$$
where $b_R(t) = \Rp(b(t))$, $b_I(t) = \Ip(b(t))$, and
\begin{gather*}
    X_1 = v_3\parder{}{v_1} + v_4\parder{}{v_2} + v_5\parder{}{v_3} + v_6\parder{}{v_4} + \frac{3}{2}\frac{2v_4v_5v_6 + (v_5^2 - v_6^2)v_3}{v_3^2 + v_4^2}\parder{}{v_5} + \frac{3}{2}\frac{2v_3v_5v_6 - v_4(v_5^2 - v_6^2)}{v_3^2 + v_4^2}\parder{}{v_6}\,,\\
    X_2 = v_3\parder{}{v_5} + v_4\parder{}{v_6}\,,\qquad X_3 = -v_4\parder{}{v_5} + v_3\parder{}{v_6}\,,\\
    X_4 = -v_3\parder{}{v_3} - v_4\parder{}{v_4} - 2v_5\parder{}{v_5} - 2v_6\parder{}{v_6}\,,\qquad X_5 = v_4\parder{}{v_3} - v_3\parder{}{v_4} + 2v_6\parder{}{v_5} - 2v_5\parder{}{v_6}\,,\\
    X_6 = -v_4\parder{}{v_1} + v_3\parder{}{v_2} - v_6\parder{}{v_3} + v_5\parder{}{v_4} - \frac{3}{2}\frac{2v_3v_5v_6 - v_4(v_5^2 - v_6^2)}{(v_3^2 + v_4^2)}\parder{}{v_5} + \frac{3}{2}\frac{2v_4v_5v_6 + v_3(v_5^2 - v_6^2)}{(v_3^2 + v_4^2)}\parder{}{v_6}\,.
\end{gather*}
These vector fields satisfy the following commutation relations
\begin{align}
    [X_1,X_2] &= X_4\,, & [X_1,X_3] &= X_5\,, & [X_1,X_4] &= X_1\,, & [X_1,X_5] &= X_6\,, & [X_1,X_6] &= 0\,,\\
    &&[X_2,X_3] &= 0\,, & [X_2,X_4] &= -X_2\,, & [X_2,X_5] &= -X_3\,, & [X_2,X_6] &= -X_5\,,\\
    &&&&[X_3,X_4] &= -X_3\,, & [X_3,X_5] &= X_2\,, & [X_3,X_6] &= X_4\,,\\
    &&&&&&[X_4,X_5] &= 0\,, & [X_4,X_6] &= -X_6\,,\\
    &&&&&&&&[X_5,X_6] &= X_1\,,
\end{align}
Then, $X_1,\ldots,X_k$ give rise to a Lie algebra of vector fields $V_{sc}$ that is isomorphic to $\mathbb{C}\otimes \mathfrak{sl}_2$ as a real vector space. Indeed, $\langle X_1,X_2,X_4\rangle \simeq \mathfrak{sl}(2,\R)\simeq\langle X_3,X_4,X_6\rangle $. Additionally, $\mathbb{C}\otimes \mathfrak{sl}_2$ decomposes as $\langle X_1,X_4,X_2\rangle\oplus \langle X_6,X_5,X_3\rangle$. Then $V_{sc}$ is graded as $V_{sc}=E_{-1}\oplus E_0\oplus E_1$, where $E_{-1}=\langle X_6,X_1\rangle$, $E_0=\langle X_4,X_5\rangle$, and $E_1=\langle X_3,X_2\rangle$, with $[E_i,E_j]=E_{i+j}$, where the sum is in the additive group $
\{-1,0,1\}$. A long calculation shows that $X_1\wedge\dotsb\wedge X_6\neq 0$ almost everywhere. The latter linear independence and the fact that $X_1,\ldots,X_6$ span a Lie algebra of vector fields spanning $\T\mathcal{O}$ explains  why it is said that \eqref{Eq:ComplSchwa} can be related to a locally automorphic Lie system (cf. \cite{GLMV_19}).

Meanwhile, the Lie algebra of Lie symmetries of the system \eqref{Eq:ComplSchwa} related to the  Lie algebra $V_{sc}$ reads 
\begin{gather}
2Y_1 =(v_1^2-v_2^2)\frac{\partial}{\partial v_1}+2v_1v_2\frac{\partial}{\partial v_2}+2(v_1v_3-v_2v_4)\frac{\partial}{\partial v_3} + 2(v_3v_2+v_1v_4)\frac{\partial}{\partial v_4}\\ \quad + 2(v_3^2+v_1v_5-v_4^2-v_2v_6)\frac{\partial}{\partial v_5}+2(v_5v_2+2v_3v_4+v_2v_6)\frac{\partial}{\partial v_6}\,,\\
Y_2=\frac{\partial}{\partial v_1}\,,\qquad
Y_3 = \frac{\partial}{\partial v_2}\,,\\
Y_4 = -v_1\frac{\partial}{\partial v_1}-v_2\frac{\partial}{\partial v_2}-v_3\frac{\partial}{\partial v_3}-v_4\frac{\partial}{\partial v_4}-v_5\frac{\partial}{\partial v_5}-v_6\frac{\partial}{\partial v_6}\,,\\
Y_5 = v_2\frac{\partial}{\partial v_1}-v_1\frac{\partial}{\partial v_2}+v_4\frac{\partial}{\partial v_3}-v_3\frac{\partial}{\partial v_4}+v_6\frac{\partial}{\partial v_5}-v_5\frac{\partial}{\partial v_6}\,.\\
2Y_6 = -2v_1v_2\frac{\partial}{\partial v_1}+(v_1^2-v_2^2)\frac{\partial}{\partial v_2}-2(v_2v_3+v_1v_4)\frac{\partial}{\partial v_3}+2(v_1v_3-v_2v_4)\frac{\partial}{\partial v_4} \\
\quad -2(2v_3v_4+v_2v_5+v_1v_6)\frac{\partial}{\partial v_5}+2(v_3^2-v_4^2+v_1v_5-v_2v_6)\frac{\partial}{\partial v_6}\,.
\end{gather}
In other words, $[X_i,Y_j]=0$ for every $i,j=1,\ldots,6$. The commutation relations for the vector fields $Y_1,\ldots,Y_6$ are
\begin{align}
    [Y_1,Y_2] &= Y_4\,, & [Y_1,Y_3] &= Y_5\,, & [Y_1,Y_4] &= Y_1\,, & [Y_1,Y_5] &= Y_6\,, & [Y_1,Y_6] &= 0\,,\\
    &&[Y_2,Y_3] &= 0\,, & [Y_2,Y_4] &= -Y_2\,, & [Y_2,Y_5] &= -Y_3\,, & [Y_2,Y_6] &= -Y_5\,,\\
    &&&&[Y_3,Y_4] &=-Y_3\,, & [Y_3,Y_5] &=Y_2\,, & [Y_3,Y_6] &= Y_4\,,\\
    &&&&&&[Y_4,Y_5] &= 0\,, & [Y_4,Y_6] &= -Y_6\,,\\
    &&&&&&&&[Y_5,Y_6] &= Y_1\,.
\end{align}
Note that $Y_1,\ldots,Y_6$ admit identical structure constants as $X_1,\ldots,X_6$. One can choose one-forms $\eta^1,\ldots,\eta^6$ to be the dual to $Y_1,\ldots,Y_6$. The existence of these dual forms is ensured by the condition $Y_1\wedge \dotsb \wedge Y_6\neq 0$ and the fact that $Y_1,\ldots,Y_6$ span $\T\mathcal{O}$. These dual one-forms remain invariant relative to the Lie derivatives with respect to the vector fields $X_1\ldots, X_6$, i.e. $\Lie_{X_i}\eta^j=0$ for $i,j=1,\ldots,6$.

Moreover, the differential forms $\d\eta^1,\ldots, \d\eta^6$, or their linear combinations, are closed differential forms that are invariant relative to the Lie derivatives along $X_1,\ldots, X_6$. These properties make them Hamiltonian vector fields relative to the presymplectic forms $\d\eta^1,\ldots, \d\eta^6$.

The appropriate linear combinations of these forms yield a set of presymplectic forms with the zero intersection of their kernels,  resulting in $X_1,\ldots, X_6$ being $\bomega$-Hamiltonian vector fields.

In particular, 
\begin{align}
    \d\eta^1 &= -\eta^5\wedge\eta^6 - \eta^1\wedge \eta^4\,, & \d\eta^2 &= -\eta^3\wedge \eta^5 - \eta^4\wedge\eta^2\,,\\
    \d\eta^3 &= -\eta^4\wedge\eta^3-\eta^5\wedge\eta^2\,, & \d\eta^4 &= -\eta^1\wedge\eta^2-\eta^3\wedge\eta^6\,,\\
    \d\eta^5 &= -\eta^1\wedge\eta^3-\eta^6\wedge\eta^2\,, & \d\eta^6 &= -\eta^1\wedge\eta^5-\eta^6\wedge\eta^4\,.
\end{align}
Every vector field in $\langle X_1,\ldots,X_6\rangle$ becomes an $\bomega$-Hamiltonian vector field relative to the two-polysymplectic form $\d\eta^1\otimes e_1+\d\eta^2\otimes e_2$. The same applies to $\d\eta^5\otimes e_1+\d\eta^6\otimes e_2$, and many other two-polysymplectic forms. This also extends to three-polysymplectic forms such as $\d \eta^1\otimes e_1+\d \eta^2\otimes e_2+\d \eta^3\otimes e_3$, provided that the kernels of their presymplectic components have zero intersection.

Let us focus on the three-polysymplectic form  defined by
\[
\bomega=\omega^1\otimes e_1+\omega^2\otimes e_2+\omega^3\otimes e_3=\d\eta^1\otimes e_1+\d\eta^2
\otimes e_2+\d\eta^4\otimes e_3\,.
\]

A two-polysymplectic Marsden--Weinstein reduction can be performed by taking, for instance, the $\bomega$-Hamiltonian vector field $X_1$ and the Lie symmetry $X_6$, which satisfies that $[X_1,X_6]=0$. Then, a two-polysymplectic momentum map $\mathbf{J}^\Phi:\mathcal{O}\rightarrow (\R^*)^3$ is given by
\[
\iota_{X_6}\d\mathbf{J}^\Phi = \iota_{X_6}\omega^1\otimes e_1+\iota_{X_6}\omega^2\otimes e_2+\iota_{X_6}\omega^3\otimes e_3=\d J_1\otimes e_1+\d J_2\otimes e_2+\d J_3\otimes e_3\,.
\]

It is a matter of a long calculation to prove that $\d J_1\wedge \d J_2\wedge \d J_3\neq 0$ based on the fact that $\partial (J_1,J_2,J_2)/\partial (v_1,v_2,v_3)\neq 0$ almost everywhere. Therefore, ${
\bf J}^{\Phi-1}(\boldsymbol{\mu})$ has dimension three. Moreover, due to $\iota_{X_6}\d\mathbf{J}^\Phi=0$, the reduced manifold ${
\bf J}^{\Phi-1}(\boldsymbol{\mu})/X_6$ is two-dimensional. 

Note that the vector field $X_1$  is tangent to the level set  $\mathbf{J}^{\Phi-1}(\boldsymbol{\mu})$ since
\[
\iota_{X_1}\iota_{X_6}\d\eta^\alpha=X_1J_\alpha=0\,,\qquad \alpha=1,2,3\,.
\]
Therefore, by Theorem \ref{Th::Xreduction} the vector field $X_1$ reduces onto the manifold $\mathbf{J}^{\Phi-1}(\boldsymbol{\mu})/X_6$.

Then, after some calculations, we obtain that condition \eqref{Eq::PolysymplecticReduction1eq} is fulfilled. To verify condition \eqref{Eq::PolysymplecticReduction2eq}, which has the form
\[
      \T_p(G_{\boldsymbol\mu}^{\boldsymbol\Delta} p) = \bigcap^k_{\alpha=1}\left(\ker\omega^\alpha_p+\T_p(G^{\Delta^\alpha}_{\mu^\alpha}p)\right)\cap \T_p(\mathbf{J}^{\Phi-1}(\boldsymbol\mu))\,,
\]
one can note that 
$$
\T_p(G_{\bm \mu} p)=\langle X_6\rangle\subset \T_p({\bf J}^{\Phi-1}(\bm{\mu}))\subset \T_p P\,.  
$$
Moreover, we have
\[
    \ker\omega^1=\langle Y_2,Y_3\rangle\,,\quad \ker \omega^2=\langle Y_1,Y_6\rangle \,,\quad \ker \omega^3=\langle Y_4, Y_5\rangle\,.
\]
In turn, this amounts to obtaining three determinants, each being non-zero, implying that no element of $\ker \omega^1, \ker \omega^2,\ker \omega^3$ belongs to $\T_p\mathbf{J}^{\Phi-1}(\bm{\mu})$. In particular, at a generic point,
$$
\det \begin{pmatrix}
    Y_2J_2&Y_2J_3\\Y_3J_2&Y_3J_3
\end{pmatrix}\neq 0\,,\qquad
\det \begin{pmatrix}
    Y_1J_1&Y_1J_3\\Y_6J_1&Y_6J_3
\end{pmatrix}\neq 0\,,\qquad
\det \begin{pmatrix}
    Y_4J_1&Y_4J_2\\Y_5J_1&Y_5J_2
\end{pmatrix}\neq 0\,.
$$
Finally, the condition \eqref{Eq::PolysymplecticReduction2eq} is satisfied, namely
\[
\big(\langle Y_2,Y_3\rangle +\langle X_6\rangle\big)\cap\big(\langle Y_1,Y_6\rangle +\langle X_6\rangle\big)\cap\big(\langle Y_4,Y_5\rangle +\langle X_6\rangle\big)\cap \T_p(\mathbf{J}^{\Phi-1}(\bm\mu))=\langle X_6\rangle\,.
\]
Hence, Theorem \ref{Th::PolisymplecticReductionJ} can be applied.

\subsection{The \texorpdfstring{$k$}{}-polysymplectic manifold given by the product of \texorpdfstring{$k$}{} symplectic manifolds}

This section presents an illustrative example of the \texorpdfstring{$k$}{}-polysymplectic Marsden--Weinstein reduction of a product of \texorpdfstring{$k$}{} symplectic manifolds (see Section \ref{Sec::SymplProduct}). This example shows different types of systems of differential equations that can be understood as Hamiltonian systems relative to a $k$-polysymplectic manifold and describes its reductions. In particular, the so-called diagonal prolongations of Lie--Hamilton systems, which appear also in the multidimensional generalisations of some integral systems, like in the case of the Winternitz--Smorodinsky oscillator on $\T^*\mathbb{R}$ (see \cite{LS_20}), can be considered as Hamiltonian systems relative to a $k$-polysymplectic manifold. One can also consider higher-dimensional Winternitz--Smorodinsky oscillators.

Let us provide some new details to the formalism in Section \ref{Sec::SymplProduct}. Define $P=P_1\times\dotsb\times P_k$ for some $k$ symplectic manifolds $(P_\alpha,\omega^\alpha)$, where $\alpha=1,\dotsc,k$. This gives rise to a $k$-polysymplectic manifold $(P,{\rm pr}_\alpha^*\omega^\alpha\otimes e_\alpha)$. Assume that each Lie group action $\Phi^\alpha: G_\alpha\x P_\alpha\rightarrow P_\alpha$  admits a symplectic momentum map $\mathbf{J}^{\Phi^\alpha}: P_\alpha\rightarrow\mathfrak{g}_\alpha^*$ for $\alpha=1,\ldots,k$. Define the Lie group action of $G=G_1\times\ldots\times G_k$ on $P$ as \eqref{eq:GrAck}.
 If one defines $\mathfrak{g}=\bigoplus_{\alpha =1}^k\mathfrak{g}_\alpha$, then there exists a $k$-polysymplectic momentum map 
\[
\mathbf{J}:P\ni(x_1,\ldots,x_k)\longmapsto
(0,\ldots, {\bf J}^\alpha,\ldots,0)\otimes {e}_\alpha = \begin{pmatrix}
\mathbf{J}^{1}&0 &\dotsb & 0\\ 
0& \mathbf{J}^{2}&\dotsb &0\\ 
\vdots&\vdots &\ddots&\vdots \\ 
0 &0 &\dotsb&\mathbf{J}^{k}\end{pmatrix}\in \mathfrak{g}^{*k}\,,
\]
where we assume $\mathbf{J}^\alpha(x_1,\ldots,x_k)=\mathbf{J}^{\Phi^\alpha}(x_\alpha)$ for $\alpha=1,\ldots,k$ and the matrix array is a practical representation of the image of ${\bf J}$.
Note that $\bm\mu=(0,\ldots,\mu^\alpha,\ldots,0)\otimes {e}_\alpha\in \mathfrak{g}^{*k}$ is a weak regular value of $\mathbf{J}$ if and only if each $\mu^\alpha\in\mathfrak{g}_\alpha^*$ is a weak regular point of its corresponding $\mathbf{J}^{\Phi^\alpha}$. Assume that some $G^\Delta_{\bm \mu}$ acts in a quotientable manner on the associated level ${\bf J}^{-1}(
\bm \mu)$. This happens if and only if every $G^{\Delta^\alpha}_{\mu^\alpha}$ acts on a quotientable manner on each ${\bf J}^{\Phi^\alpha-1}(\mu^\alpha)$ for $\alpha=1,\ldots,k$. %

We already showed that the conditions \eqref{Eq::PolysymplecticReduction1eq}  and \eqref{Eq::PolysymplecticReduction2eq} are satisfied. By Theorem \ref{Th::PolisymplecticReductionJ}, these equations guarantee that, on the reduced manifold $\mathbf{J}^{-1}(\bm\mu)/G^{\bm\Delta}_{\bm\mu}$, there exists a uniquely induced $k$-polysymplectic manifold, 
\[
 \left(\mathbf{J}^{-1}(\bm\mu)/G^{\bm\Delta}_{\bm\mu}\simeq \mathbf{J}^{\Phi^1-1}(\mu^1)/G^{\Delta^1}_{1\bm\mu}\x\dotsb\x \mathbf{J}^{\Phi^k-1}(\mu^k)/G^{\Delta^k}_{k\bm \mu}\ ,\ 
\bm{\omega}_{\bm\mu}=\sum_{\alpha=1}^k\omega^{\mu^\alpha}\otimes e_\alpha\right)
\]
for some reduced presymplectic forms $\omega_{\mu^1},\ldots,\omega_{\mu^k}$.

Next, let us consider a vector field $X$ on $P$ that is $\bomega$-Hamiltonian and $G$-invariant. By Theorem \ref{Th::Xreduction}, the vector field $X$ can be written in the following way
\[
X=\sum^k_{\alpha=1}X_\alpha\,,
\]
where each $X_\alpha$ can be considered as a vector field on $P_\alpha$ that is tangent to $\mathbf{J}^{\Phi^\alpha-1}(\mu^\alpha)$ for $\alpha=1,\ldots, k$. Recall that $\iota_{X_\alpha}\omega^\beta =\delta_\alpha^\beta \d h^\alpha$ for $\alpha,\beta=1,\ldots,k$. Moreover, this frequently happens in diagonal prolongations of Lie--Hamilton systems, where we have a vector field $X^{[m]}$ defined on a manifold of the form $N^m$ that can be considered as a copy of a Hamiltonian system on each $N$ relative to a symplectic manifold on that $N$ (cf. \cite{LS_20}). Then,
\[
    \d\boldsymbol{h}=\sum^k_{\alpha=1}\d h^\alpha\otimes e_\alpha=\sum^k_{\alpha=1}\iota_{X}\omega^\alpha\otimes e_\alpha\,.
\]
Next, $\bm h_\xi=\bm h-\langle \mathbf{J}-\bm{\mu}_e,\xi\rangle$ for $\xi\in \mathfrak{g}$, and Theorem \ref{Th::ksymEM} yields that $z_e=(z_{1e},\dots,z_{ke})\in P$ is a $k$-polysymplectic relative equilibrium point if and only if each $z_{\alpha e}$ is a symplectic relative equilibrium point of a Hamiltonian vector field $X_\alpha$ on the symplectic manifold $(P_\alpha,\omega^\alpha)$ relative to some $\xi_\alpha\in \mathfrak{g}_\alpha$. 
Then, a $k$-polysymplectic relative equilibrium point $z_e$ is formally stable if there exists a series of supplementary spaces $\mathcal{S}^\alpha$ to $\T_{z_e}(G^{\bm\Delta}_{\bm \mu_e}z_e)\oplus (\ker \omega^\alpha_{z_e}\cap \T_{z_e}{\bf J}^{-1}(\bm\mu_e))$ in $\T_{z_e}{\bf J}^{-1}(\bm\mu_e)$, with $\alpha=1,\ldots,k$, such that
\begin{equation}
    \label{Eq::ExProductSym}
\left(\delta^2  h^\alpha_\xi\right)_{z_e}(v_{z_e},v_{z_e})>0\,,\qquad \forall v_{z_e}\in \mathcal{S}^\alpha\setminus\{0\}\,,\qquad \alpha=1,\ldots,k\,
\end{equation}
and $\mathcal{S}^1+\ldots+\mathcal{S}^k+\T_{z_e}(G^{\bm\Delta}_{\bm \mu_e}z_e)=\T_{z_e}{\bf J}^{-1}(\bm\mu_e)$.

\subsubsection{Product of oscillators}

Let us detail a practical application of the formalism above. Consider the product of $k$ isotropic three-dimensional oscillators given by the equations
$$ \frac{\d ^2x^i_\alpha}{\d t^2} = -b_\alpha^2x^i_\alpha\,,\qquad \alpha = 1,\dotsc,k\,, \qquad i=1,2,3,$$
where the $b_\alpha>0$, with $\alpha=1,\ldots,k$, are a series of constants. The above system of second-order differential equations can be written as a first-order  system of differential equations
\begin{equation}\label{eq:product-oscillators}
    \begin{dcases}
        \frac{\d x^i_\alpha}{\d t} = p^i_\alpha\,,\\
        \frac{\d p^i_\alpha}{\d t} = -b_\alpha^2 x^i_\alpha\,,
    \end{dcases} \qquad \alpha = 1,\dotsc,k\,,\qquad i=1,2,3\,,
\end{equation}
on the product manifold $P = (\cT\R^3)^k$. The $\alpha$-th factor $\cT\R^3$ in $P$ is a symplectic manifold equipped with the symplectic form
\[
 \omega^\alpha = \sum_{i=1}^3\d x^i_\alpha\wedge\d p^i_\alpha\,,
\]
where we stress that there is no sum over the index $\alpha$.
Then, $P$ is a $k$-polysymplectic manifold when endowed with the $\R^k$-valued form $\bomega = \sum_{\alpha=1}^k\omega^\alpha\otimes e_\alpha$, where $\omega^1,\ldots,\omega^k$ are considered as pulled back to $P$ in the natural way. Moreover, \eqref{eq:product-oscillators} describes the integral curves of the vector field
$$ 
X_{\bm h} = \sum_{\alpha = 1}^k\sum_{i=1}^3 \left( p_\alpha^i\parder{}{x^i_\alpha} - b_\alpha^2 x_\alpha^i\parder{}{p_\alpha^i} \right)\,, 
$$
which is $\bomega$-Hamiltonian admitting an $\bomega$-Hamiltonian function
\begin{equation}\label{eq:kHamFun}
{\bm h}= \frac{1}{2}\sum_{\alpha=1}^k\left( p_\alpha^2 + b_\alpha^2x_\alpha^2 \right)\otimes e_\alpha\,,\qquad p_\alpha^2=\sum_{i=1}^3(p_\alpha^i)^2,\qquad x_\alpha^2=\sum_{i=1}^3(x_\alpha^i)^2.
\end{equation}
Let us consider a Lie group action $\Phi^\alpha: \SO(3)\times (\cT\R^3)_\alpha\rightarrow (\cT\R^3)_\alpha$, where each $\Phi^\alpha$ is the lift of the natural Lie group action $\Psi:\SO(3)\times \R^3\rightarrow \R^3$ induced by rotations on $\mathbb{R}^3$ to the $\alpha$-th copy of $\cT\R^3$ in $P$. Then, the resulting Lie group action $\Phi$ on $(\cT\R^3)^k$ given by \eqref{eq:GrAck} reads
\[
\Phi:\SO(3)^k\times (\cT\R^3)^k\longrightarrow (\cT\R^3)^k\,.
\]
The Lie algebra of fundamental vector fields of $\Phi$ is spanned by the basis of vector fields on $P$ of the form
\begin{gather}
    \xi^1_{\alpha P}= \left(x_\alpha^1\parder{}{x_\alpha^2} - x_\alpha^2\parder{}{x_\alpha^1} + p_\alpha^1\parder{}{p_\alpha^2} - p_\alpha^2\parder{}{p_\alpha^1}\right),\qquad \xi^2_{\alpha P} = \left(x_\alpha^2\parder{}{x_\alpha^3} - x_\alpha^3\parder{}{x_\alpha^2} + p_\alpha^2\parder{}{p_\alpha^3} - p_\alpha^3\parder{}{p_\alpha^2}\,\right),\\
    \xi^3_{\alpha P} = \left(x_\alpha^3\parder{}{x_\alpha^1} - x_\alpha^1\parder{}{x_\alpha^3} + p_\alpha^3\parder{}{p_\alpha^1} - p_\alpha^1\parder{}{p_\alpha^3}\right)\,,
\end{gather}
with $\alpha=1,\ldots,k$. These vector fields are Lie symmetries of $\bm\omega$ and $\bm h$. Moreover, a $k$-polysymplectic momentum map associated with $\Phi$ is given by
$\mathbf{J}:(\cT\mathbb{R}^3)^k\rightarrow [(\mathfrak{so}_3^k)^*]^k$ such that
\begin{equation*}
    \mathbf{J}(\boldsymbol{q}_1,\ldots,\boldsymbol{q}_k)=(0,0,0;\ldots;J_\alpha^1,J_\alpha^2,J_\alpha^3;\ldots;0,0,0)\otimes e_\alpha
\end{equation*}
where $\boldsymbol{q}_\alpha=(x^1_\alpha,x^2_\alpha,x^3_\alpha,p^1_\alpha,p^2_\alpha,p^3_\alpha)\in \T^*\R^3$ for $\alpha=1,\ldots,k$, while 
$$    (J^1_\alpha,J^2_\alpha,J^3_\alpha)=(x_\alpha^1p_\alpha^2 - x_\alpha^2p_\alpha^1, x_\alpha^2p_\alpha^3 - x_\alpha^3p_\alpha^2, x_\alpha^3p_\alpha^1 - x_\alpha^1p_\alpha^3)\,,
$$
and $\alpha=1,\ldots,k$. Note that the elements of $\mathfrak{so}_3^*$ are represented by the coordinates given in an appropriate basis. The function $x_\alpha^1p_\alpha^2 - x_\alpha^2p_\alpha^1$ is the angular momentum, $p_{\alpha\varphi}$, of the $\alpha$-th particle in the corresponding spherical coordinates $\{r_\alpha,\theta_\alpha,\varphi_\alpha\}$. Meanwhile, $L_\alpha^2 = (J_\alpha^1)^2+(J_\alpha^2)^2+(J_\alpha^3)^2$ is the square of the total angular momentum of the $\alpha$-th particle. Both quantities are conserved by the evolution of $X_{\bm h}$.

The momentum map $\mathbf{J}$ is $(\Ad^*)^k$-equivariant. Recall that $\mathbf{J}=(0,\ldots,\mathbf{J}^{\alpha},\ldots,0)\otimes e_\alpha$. Then, $\bm{\mu}=(0,0,0;\ldots;J_\alpha^1,J_\alpha^2,J_\alpha^3;\ldots;0,0,0)\otimes e_\alpha$ is weakly regular value of $\mathbf{J}$ if and only if each triple $\mu^\alpha=(J_\alpha^1,J_\alpha^2,J_\alpha^3)\in\mathfrak{so}_3^*$ is a weakly regular value of $\mathbf{J}^{\Phi^\alpha}$, where $\alpha=1,\ldots,k$. Let us fix some weakly regular $\bm\mu$. Then,
\[
    \T_{\boldsymbol{q}}(\mathbf{J}^{-1}(\bm\mu))=\bigoplus _{\alpha=1}^k\T_{{\boldsymbol{q}}_\alpha}(\mathbf{J}_\alpha^{\Phi^\alpha-1}(\mu^\alpha))\,,\qquad \forall {\bf q}=({\bf q}_1,\ldots,{\bf q}_k)\in P\,.
\]
Moreover, 
\[
\xi_{\alpha P}^iJ^j_\beta = -\delta_{\alpha\beta}\epsilon_{ijk}J^k_\beta\,,\qquad i,j=1,2,3,\qquad \alpha,\beta=1,\ldots,k.
\]
The isotropy subgroup of $\Phi$ at $\bm \mu$ is given by the Cartesian product of all the isotropy subgroups corresponding to each $\mu^\alpha$ relative to  $\Phi^\alpha$ and $\alpha=1,\ldots,k$. To obtain $G_{\mu^\alpha}$, one may verify when $\sum_{i=1}^3\lambda_i(\xi_{\alpha P}^i)$ belongs to $\T_{{\boldsymbol{q}_\alpha}}({\bf J}^{\Phi^\alpha-1}(\mu_\alpha))$, namely $\sum_{i=1}^3\lambda_i(\xi_P^i)_\alpha J^j_\alpha=0$ for $j=1,2,3$ (with no summation over $\alpha$), which occurs if and only if
$$
    \begin{pmatrix}
        0&-J_\alpha^3&J_\alpha^2\\
        J_\alpha^3&0&-J_\alpha^1\\
        -J_\alpha^2&J_\alpha^1&0
    \end{pmatrix}\begin{pmatrix}\lambda_1\\\lambda _2\\\lambda_3\end{pmatrix} = \begin{pmatrix}0\\0\\0\end{pmatrix}\,.
$$
The matrix of coefficients has rank two for $L^2_\alpha\neq 0$. Moreover,  $\mathbf{J}^{\Phi^\alpha}$ has a regular value at $\mu_\alpha$ when $L^2_\alpha\neq 0$ for every $\alpha=1,\ldots,k$. Let us restrict to that case. Each isotropy subgroup $G_{\mu^\alpha}$ has always dimension one. Hence, the reduced manifold ${\bf J}^{\Phi-1}({\bm \mu})/G_{\bm \mu}$ has dimension $6k-3k-k=2k$. The conditions for the $k$-polysymplectic Marsden--Weinstein reduction results, as already commented, from the ones for the symplectic reduction on each component, which are satisfied. Hence, the $k$-polysymplectic Marsden--Weinstein reduction exists.

Note that $\bm h=\frac 12\sum_{\alpha=1}^k(p_{\alpha r}^2+p_{\alpha \varphi}^2/(r_\alpha^2\sin^2 \theta_\alpha)+p_{\alpha \theta}^2/r_\alpha^2+b_\alpha^2r_\alpha ^2)\otimes e_\alpha$ and $\omega^\alpha=\d r_\alpha\wedge \d p_{\alpha r}+\d \theta_\alpha\wedge \d p_{\alpha\theta}+\d\varphi_\alpha\wedge \d p_{\alpha\varphi}$, for each $\alpha=1,\ldots,k$, in spherical coordinates for the $k$ component manifolds of $(\cT \mathbb{R}^3)^k$. Then, the differential equations for the integral curves of $X_{\bm h}$  read
\begin{equation}\label{eq:SysBezRed}
\begin{gathered}
\frac{\d p_{\alpha r}}{\d t}=\frac{p_{\alpha  \varphi}^2}{r_\alpha ^3\sin^2\theta_\alpha }+\frac{p_{\alpha \theta}^2}{r_\alpha ^3}-b^2_\alpha r_\alpha \,,\qquad \frac{\d p_{\alpha \varphi}}{\d t}=0\,,\qquad 
\frac{\d p_{\alpha \theta}}{\d t}=\frac{p_{\alpha \varphi}^2\cos\theta_\alpha }{r_\alpha^2\sin^3
\theta_\alpha}\,,\\ \frac{\d r_\alpha}{\d t} = p_{\alpha r}\,,\qquad   \frac{\d \theta_\alpha}{\d t}=\frac{p_{\alpha\theta}}{r_\alpha^2}\,,\qquad \frac{\d\varphi_{\alpha}}{\d t} = \frac{p_{\alpha\varphi }}{r_\alpha^2\sin^2\theta_\alpha}\,.
\end{gathered}
\end{equation}
$k$-Polysymplectic relative equilibrium points are given by those points for which the vector field $X_{\bm h}$ on $(\cT \mathbb{R}^3)^k$ corresponding to the dynamics is proportional to one of the fundamental vector fields of $\Phi$. In particular, let us take $z_e\in P$ such that $z_e=(r_{\alpha},\theta_\alpha\!\!=\!\!\frac{\pi}{2},\varphi_\alpha,p_{\alpha r}\!\!=\!\!0,p_{\alpha\theta}=0,p_{\alpha\varphi})$ and $L_\alpha=b_\alpha r^2_\alpha=p_{\alpha \varphi}$ for every $\alpha=1,\ldots,k$ on analysed points. Then, the $\bomega$-Hamiltonian vector field $X_{\bh}$ at such points is
\[
    X_{\bm h}=\sum_{\alpha=1}^k\frac{p_{\alpha\varphi}}{r_\alpha^2}\frac{\partial}{\partial \varphi_{\alpha}}.
\]
This implies that $z_e\in P$ is a $k$-polysymplectic relative equilibrium point of $X_{\bh}$. 
Let us demonstrate this by
applying Theorem \ref{Th::ksymEM}. This theorem ensures that $z_e$ is a $k$-polysymplectic relative equilibrium point of $X_{\bh}$ if and only if there exists $\xi\in \mathfrak{so}_3^k$ such that $\boldsymbol{h}_{\xi}=\bh-\langle \mathbf{J}^\Phi-\boldsymbol{\mu}_e,\xi\rangle $ has a critical point at $z_e$. Indeed, for 
\[
\xi=(p_{1\varphi}/r^2_1,0,0;\ldots;p_{k\varphi}/r^2_k,0,0)\in \mathfrak{so}_3^k,
\]
the $\R^k$-valued function
\[
\boldsymbol{h}_{\xi}=\boldsymbol{h}-\langle\mathbf{J}^\Phi-\boldsymbol{\mu}_e ,\xi\rangle=\sum_{\alpha=1}^k(h^\alpha-\langle (0,\ldots,{\bf J}^{\Phi_\alpha}-(L_\alpha,0,0),\ldots,0),\xi \rangle)\otimes e_\alpha\,,
\]
has a critical point at $z_e$. Therefore, $z_e$ is a $k$-polysymplectic relative equilibrium point.


By Theorem \ref{Th::PolisymplecticReductionJ}, the reduced manifolds is $(\cT \mathbb{R})^k$ with coordinates $\{r_\alpha,p_{\alpha r}\}$ for $\alpha=1,\ldots,k$. The reduced $k$-polysymplectic form on the reduced manifold is given by
\[
\boldsymbol{\omega}_{\bm \mu}=\sum^k_{\alpha=1}\d r_{\alpha}\wedge \d p_{\alpha r}\otimes e_\alpha,
\]
and the reduced $\bomega_{\bm\mu}$-Hamiltonian reads
$$
    {\bm f}_{\bm \mu} =\frac 12\sum_{\alpha=1}^k\left(p_{\alpha r}^2 + \frac{L_\alpha^2}{r^2_\alpha} + b_\alpha^2 r_\alpha^2 \right)\otimes e_\alpha\,.
$$
Furthermore, one has that
\[
\frac{\d p_{\alpha r}}{\d t}=-b_\alpha ^2r_\alpha +\frac{L_\alpha^2}{r^3_\alpha}\,,\qquad \frac{\d r_\alpha }{\d t}=p_{\alpha r}\,,\qquad \alpha=1,\dotsc,k\,.
\]
Thus, the equilibrium points of $X_{{\bm f}_{\bm\mu}}$ have $p_{\alpha r}=0$ and 
$$
    -b^2_\alpha r_\alpha+\frac{L_\alpha^2}{r^3_\alpha}=0\,
$$
for $\alpha=1,\ldots,k$. Note that this point is the projection of a $k$-polysymplectic relative equilibrium point $z_e\in P$.

The Hessian of the functions $f_{\bm \mu}^\alpha$ is positive-definite in a supplementary to the kernel of ${\bm \omega}_{\mu^\alpha}$ at the equilibrium point. Indeed, 
$$
\Hess( f_{\bm \mu}^\alpha) = \begin{pmatrix} 1&0 \\ 0&4b_\alpha^2 \end{pmatrix}\,.
$$
Moreover, the function $\sum_{\alpha=1}^kf^\alpha_{\bm \mu}$ has a positive-definite Hessian and the equilibrium point becomes a strict minimum. This means that the reduced $k$-polysymplectic relative equilibrium point is stable. In the original manifold, the orbits around $k$-polysymplectic relative equilibrium points remain in the anti-image in ${\bf J}^{\Phi-1}(\bm\mu)$ of an open neighbourhood of the projection of the $k$-polysymplectic relative equilibrium points.

\subsection{\texorpdfstring{$k$}{}-polysymplectic affine Lie systems}
Let us apply our techniques to a family of affine inhomogeneous systems of first-order differential equations. It is worth noting that all such systems are Lie systems \cite{CL_11}. We will hereafter call such differential equations {\it affine Lie systems}. Many such systems appear in control theory and other relevant disciplines \cite{CR_03}. In particular, we are here concerned with affine Lie systems admitting a Lie algebra of Hamiltonian vector fields relative to a $k$-polysymplectic form. We call them {\it $k$-polysymplectic affine Lie systems}. 

Although our techniques could be extended to other affine Lie systems, let us restrict ourselves to the particular case 

\begin{equation}\label{eq:AffineLiekpoly}
\frac{\d}{\d t}
\begin{pmatrix}
    x_1\\
    x_2\\
    x_3\\
    x_4\\
    x_5
\end{pmatrix}=
\begin{pmatrix}
    b_1(t)\\
    b_2(t)\\
    b_3(t)\\
    b_4(t)\\
    b_5(t)
\end{pmatrix}+b_6(t)
\begin{pmatrix}
    0&0&0&0&0\\
    0&0&0&0&0\\
    0&0&0&0&1\\
    0&0&0&0&0\\
    0&0&-1&0&0
\end{pmatrix}
\begin{pmatrix}
    x_1\\
   x_2\\
    x_3\\
    x_4\\
    x_5\\
\end{pmatrix}\,,
\end{equation}
where $b_1(t),\ldots,b_6(t)$ are arbitrary $t$-dependent functions.
The above  system is the system of differential equations describing the integral curves of the $t$-dependent vector field
$$
X=\sum_{\alpha=1}^6b_\alpha(t)X_\alpha\,,
$$
where
\begin{equation*}
    X_1=\frac{\partial}{\partial x_1}\,,\quad X_2=\frac{\partial}{\partial x_2}\,,\quad
    X_3=\frac{\partial }{\partial x_3}\,,\quad
    X_4=\frac{\partial}{\partial x_4}\,,\quad 
    X_5=\frac{\partial}{\partial x_5}\,,\quad
    X_6=x_5\frac{\partial}{\partial x_3}-x_3\frac{\partial }{\partial x_5}\,.
\end{equation*}
These vector fields span a six-dimensional Lie algebra of vector fields $V$, with the following non-vanishing commutation relations
\begin{equation*}
[X_3,X_6]=-X_5\,,\qquad
[X_5,X_6]=X_3\,.
\end{equation*}
Consider the case where $b_1(t),\ldots,b_6(t)$ are constants, denoted as $c_1,\ldots,c_6\in\R$, respectively.
Since the vector fields $X_1\wedge \dotsb\wedge X_6=0$, the methods presented in Section \ref{Sec:ComScEq} for describing $k$-polysymplectic forms compatible with Lie systems do not apply to \eqref{eq:AffineLiekpoly}. Nevertheless, there exists a two-polysymplectic form on $\R^5$ given by
\[
\bm\omega=(\d x_3\wedge \d x_5+\d x_4\wedge \d x_1)\otimes e_1+(\d x_3\wedge \d x_5 +\d x_4\wedge \d x_2)\otimes e_2
\]
turning all the vector fields of $V$ into $\bomega$-Hamiltonian vector fields. Indeed, $\bomega$-Hamiltonian functions for $X_1,\ldots,X_6$ have the form
\begin{align}
    \bh_1 &= -x_4\otimes e_1\,, & \bh_2 &= -x_4\otimes e_2\,, & \bh_3&=x_5\otimes e_1 + x_5\otimes e_2\,,\\
    \bh_4 &= x_1\otimes e_1+x_2\otimes e_2\,, & \bh_5 &= -x_3\otimes e_1-x_3\otimes e_2\,, & \bh_6&=\frac{1}{2}(x_3^2+x_5^2)\otimes e_1+\frac{1}{2}(x_3^2+x_5^2)\otimes e_2\,. 
\end{align}

The flow of the vector field $X_4$ gives rise to a two-polysymplectic Lie group action $\Phi:\R\times \R^5\rightarrow \R^5$. Moreover, $X_4$, which spans the space of fundamental vector fields of $\Phi$, is a Lie symmetry of system \eqref{eq:AffineLiekpoly}. Then, a two-polysymplectic momentum map associated with $\Phi$ reads
$$
{\bf J}^\Phi:(x_1,x_2,x_3,x_4,x_5)\in \mathbb{R}^5\mapsto (x_1,x_2)=\bm\mu\in \mathbb{R}^{*2}.
$$
Note that $\bm\mu\in \R^{*2}$ is a regular value of ${\bf J}^\Phi$, and ${\bf J}^\Phi$ is $\Ad^{*2}$-equivariant two-polysymplectic momentum map. It can be proved that the example satisfies the conditions \eqref{Eq::PolysymplecticReduction1eq} and \eqref{Eq::PolysymplecticReduction2eq}. Hence, Theorem \ref{Th::PolisymplecticReductionJ} can be applied. The vector field $X_4$ is tangent to $\mathbf{J}^{\Phi-1}(\bm \mu)$ and $\T_x(G_{\bm\mu}x)=\langle \frac{\partial}{\partial x_4}\rangle$ for $x\in \R^5$. Therefore, $P_{\bm\mu}=\mathbf{J}^{\Phi-1}(\bm \mu)/\mathbb{R}$ is a two-dimensional manifold and the variables $\{x_3,x_5\}$ can be considered in a natural manner as variables on $P_{\bm\mu}$. The reduced two-polysymplectic form reads
$$
\bm\omega_{\bm\mu}=
\omega^1_{\bm \mu}\otimes e_1+\omega^2_{\bm \mu}\otimes e_2=\d x_3\wedge \d x_5\otimes e_1+\d x_3\wedge \d x_5\otimes e_2\,.
$$
To apply Theorem \ref{Th::Xreduction}, the affine Lie system must be tangent to ${\bf J}^{\Phi-1}(\bm\mu)$, which can be ensured by assuming that its associated $\bm\omega$-Hamiltonian function has to be invariant relative to $X_4$. These conditions are satisfied by imposing $c_1=c_2=0$. The resulting vector field, $X_\bomega=c_3X_3+c_4X_4+c_5X_5+c_6X_6$, projects onto $P_{\bm\mu}= {\bf J}^{\Phi-1}(\bm\mu)/\mathbb{R}$ giving rise to an $\bm\omega_{\bm\mu}$-Hamiltonian vector field of the form 
$$
X_{{\bm \mu}}=c_6\left(x_5\frac{\partial}{\partial x_3}-x_3\frac{\partial}{\partial x_5}\right)+c_3\frac{\partial}{\partial x_3} + c_5\frac{\partial}{\partial x_5}\,.
$$
The $\bm\omega_{\bm \mu}$-Hamiltonian function of $X_{\bm\mu}$ reads
$$
\boldsymbol{f}_{\bm \mu}=\left(c_3x_5-c_5x_3+c_6\left(\frac{x_3^2}{2}+\frac{x_5^2}{2}\right)\right)\otimes e_1+\left(c_3x_5-c_5x_3+c_6\left(\frac{x_3^2}{2}+\frac{x_5^2}{2}\right)\right)\otimes e_2\,.
$$
Next, the methods introduced in Section \ref{Sec::energy-momentum} will be employed to find the two-polysymplectic relative equilibrium points of the $\bm\omega$-Hamiltonian vector field 
\[
Y=X_4+X_6
\]
and study their stability.
According to Theorem \ref{Th::ksymEM}, a two-polysymplectic relative equilibrium point $z_e\in P$ is a point for which there exists $\xi\in\mathfrak{g}\simeq\R$ such that $z_e$ is a critical point of each component of the $\R^2$-valued function 
\[
\bh_\xi=\left(x_1-\xi(x_1-\mu^1)+\frac{1}{2}(x_3^2+x_5^2)\right)\otimes e_1+\left(x_2-\xi(x_2-\mu^2)+\frac{1}{2}(x_3^2+x_5^2)\right)\otimes e_1\,.
\]
This happens for $\xi=1$ and $z_e=(x_1,x_2,x_3=0,x_4,x_5=0)\in\R^5$, where $x_1,x_2,x_4$ are arbitrary.

To examine the relative stability of $z_e\in\R^5$, note that
the supplementary spaces to $\T_{z_e}(G_{{\bm\mu}_e}z_e)+\ker\omega^1_{z_e}$ and $\T_{z_e}(G_{{\bm\mu}_e}z_e)+\ker\omega^2_{z_e}$ in $\T_{z_e}({\bf J}^{\Phi-1}(\bm\mu_e))$ have the form
$$
\mathcal{S}^1=\left\langle\frac{\partial}{\partial x_3},\frac{\partial}{\partial x_5} \right\rangle\,,\qquad \mathcal{S}^2=\left\langle\frac{\partial}{\partial x_3},\frac{\partial}{\partial x_5}\right\rangle\,,
$$
respectively. Then, $\mathcal{S}^1+\mathcal{S}^2+\T_{z_e}(G_{\bm\mu_e}z_e)=\T_{z_e}({\bf J}^{\Phi-1}(\bm\mu_e))$ and the Hessian of $(\delta^2h^1_\xi)_{z_e}$ at $z_e$ in the subspace $\mathcal{S}^1$ and the Hessian of $(\delta^2h^2_\xi)_{z_e}$ in the subspace $\mathcal{S}^2$ are definite-positive. 
Therefore, by our criterion, a two-polysymplectic relative equilibrium point $z_e\in\R^5$ is relatively stable, namely its projection to $P_{\bm\mu_e}$ is stable. More specifically, the reduced system has an $\bomega_{\bm\mu_e}$-Hamiltonian function whose components, $f^1_{\bm\mu_e},f^2_{\bm\mu_e}$ are such that their Hessians at equilibrium points $\pi_{\bm\mu_e}(z_e)=(x_3=0,x_5=0)$ are definite-positive in the directions of $\ker(\omega^1_{\bm \mu_e})_{z_e}$  and  $\ker(\omega^2_{\bm \mu_e})_{z_e}$, respectively. Indeed, the reduced $\bomega_{\bm\mu_e}$-Hamiltonian function reads
$$
\bm f_{\bm\mu_e}=\frac{1}2(x_3^2+x_5^2)\otimes e_1+\frac{1}2(x_3^2+x_5^2)\otimes e_2
$$
and the function
$$
f_{\bm\mu_e}^1+f_{\bm\mu_e}^2=x_3^2+x_5^2\,,
$$
is invariant under the dynamics of $Y_{\bm\mu_e}$ and has a strict minimum at $\pi_{\bm\mu_e}(z_e)=(x_3=0,x_5=0)$. Hence, the reduced two-polysymplectic Hamiltonian system is stable.

\subsection{Quantum quadratic Hamiltonian operators}

Next, let us analyse an example based upon the Wei--Norman equations for the automorphic Lie system related to quantum mechanical systems described by  quadratic Hamiltonian operators, which describe as particular cases quantum harmonic oscillators with/without dissipation \cite{CL_11,WN_63}. In this case, the system of differential equations under study is the one determining the integral curves of the time-dependent vector field
\begin{equation}\label{eq:CoeAutSys}
X=\sum_{\alpha=1}^6b_\alpha(t)X_\alpha^R,
\end{equation}
for certain $t$-dependent functions $b_1(t),\ldots,b_6(t)$ and the vector fields
\begin{equation}\label{Eq:Rel}
\begin{aligned}
&X^R_1=\frac{\partial}{\partial v_1}+v_5\frac{\partial }{\partial v_4}-\frac{1}{2}v_5^2\frac{\partial }{\partial v_6}\,,
&X^R_2&=v_1\frac{\partial}{\partial v_1}+\frac{\partial }{\partial v_2}+\frac{1}{2}v_4\frac{\partial }{\partial v_4}-\frac{1}{2}v_5\frac{\partial }{\partial v_5}\,,\\
&X^R_3=v_1^2\frac{\partial }{\partial v_1}+2v_1\frac{\partial }{\partial v_2}+e^{v_2}\frac{\partial }{\partial v_3}-v_4\frac{\partial }{\partial v_5}+\frac{1}{2}v_4^2\frac{\partial }{\partial v_6}\,,
&X^R_4&=\frac{\partial }{\partial v_4}\,,
\\
&X^R_5=\frac{\partial }{\partial v_5}-v_4\frac{\partial }{\partial v_6}\,,
&X^R_6&=\frac{\partial }{\partial v_6}\,.
\end{aligned}
\end{equation}
The commutation relations between the above vector fields read
{\small
\begin{equation}
\begin{aligned}
&[X^R_1,X^R_2]=X^R_1\,, &&&&&&&&
\\
&[X^R_1,X^R_3]=2\, X^R_2\,,
&[X^R_2,X^R_3]&=X^R_3\,,
&& && &&
\\
&[X^R_1,X^R_4]=0\,,
&[X^R_2,X^R_4]&=-\frac 12\, X^R_4\,,
&[X^R_3,X^R_4]&=X^R_5\,,
&& &&
\\
&[X^R_1,X^R_5]=-X^R_4\,,
&[X^R_2,X^R_5]&=\frac 12\, X^R_5\,,
&[X^R_3,X^R_5]&=0\,,
&[X^R_4,X^R_5]&=-X^R_6\,,
&&
\\
&[X^R_1,X^R_6]=0\,,
&[X^R_2,X^R_6]&=0\,,
&[X^R_3,X^R_6]&=0\,,
&[X^R_4,X^R_6]&=0\,,
&[X^R_5,X^R_6]&=0\,.
\end{aligned}
\end{equation}
}
It is known that the Lie algebra of Lie symmetries of $\langle X^R_1,\ldots,X_6^R\rangle$  is spanned by
\begin{equation}\label{Eq:Rel2}
\begin{gathered}
X^L_1=e^{v_2}\frac{\partial}{\partial v_1}+2v_3\frac{\partial }{\partial v_2}+v_3^2\frac{\partial }{\partial v_3}\,,\qquad
X^L_2=\frac{\partial}{\partial v_2}+v_3\frac{\partial }{\partial v_3}\,,\qquad
X^L_3=\frac{\partial }{\partial v_3}\,,\\
X^L_4=e^{-v_2/2}(e^{v_2}-v_1v_3)\frac{\partial }{\partial v_4}-e^{-v_2/2}v_3\frac{\partial}{\partial v_5}-e^{-v_2/2}(e^{v_2}-v_1v_3)v_5\frac{\partial }{\partial v_6}\,,\\X^L_5=v_1e^{-v_2/2}\frac{\partial }{\partial v_4}+e^{-v_2/2}\frac{\partial}{\partial v_5}-v_1v_5e^{-v_2/2}\frac{\partial}{\partial v_6}\,,\quad X^L_6=\frac{\partial }{\partial v_6}\,.
\end{gathered}
\end{equation}
In particular, let us focus on systems \eqref{eq:CoeAutSys} with constant coefficients and, in particular,  
$$
X^R_5 = \frac{\partial}{\partial v_5}-v_4\frac{\partial}{\partial v_6}\,.
$$
Then, a Lie symmetry of our system is given by
$$
Y = \frac{\partial}{\partial v_5}\,.
$$
A two-polysymplectic form on $\R^6$ can be defined in the following way
\begin{align*}
  \bm\omega &= \omega^1\otimes e_1+\omega^2\otimes e_2 \\
  &= \left(\d v_1\wedge \d v_3 + \d v_2\wedge \d v_4 + \d v_5\wedge \d v_1 + \d v_4\wedge\d v_6\right)\otimes e_1+\left( \d v_4\wedge\d v_6 - \d v_3\wedge\d v_5\right)\otimes e_2\,.
\end{align*}
Note that
$$
    \ker\omega^1 = \left\langle \parder{}{v_3} + \parder{}{v_5},\parder{}{v_2} + \parder{}{v_6}\right\rangle\,,\qquad 
    \ker\omega^2 = \left\langle \parder{}{v_1},\parder{}{v_2} \right\rangle,
$$
and $\ker\omega^1\cap\ker\omega^2=0$ and $(\R^6,\bm\omega)$ becomes a two-polysymplectic manifold. The vector field $Y$ is a Lie symmetry of the two-polysymplectic form, i.e. $\Lie_{Y}\bomega=0$. Then,
\[
\iota_{Y_3}\omega^1 = \d v_1\,,\qquad\iota_{Y_3}\omega^2 = \d v_3\,,
\]
and a two-polysymplectic momentum map $\mathbf{J}^\Phi$ associated with the Lie group action given by the flow of $Y$ reads
\[
\mathbf{J}^\Phi:\R^6\ni x \longmapsto (v_1,v_3)=\bm\mu\in (\mathfrak{g}^*)^2\simeq \R^{*2}\,.
\]
Note that every $\bm\mu=(\mu^1,\mu^2)\in \mathbb{R}^{*2}$ is a weakly regular value of $\mathbf{J}^\Phi$, which is ${\rm Ad}^{*2}$-equivariant. The isotropy group for every $\bm\mu \in \mathbb{R}^{*2}$ reads $G_{\bm\mu}=\mathbb{R}$. Hence, $\mathbf{J}^{\Phi-1}(\bm{\mu})$ is a submanifold, as well as $\mathbf{J}_1^{\Phi-1}(\mu^1)$ and $\mathbf{J}_2^{\Phi-1}(\mu^2)$. Since $Y_5$ is tangent to  $\mathbf{J}^{\Phi-1}(\bm{\mu})$, then $P_{\bm{ \mu}}=\mathbf{J}^{\Phi-1}(\bm{\mu})/G_{\bm \mu}$ can be locally coordinated by the functions $\{v_2,v_4,v_6\}$.  

The vector field $X_5^R$ is $\bm\omega$-Hamiltonian with
\begin{equation}\label{Eq:HamFunTot}
\iota_{X_5^R}\bm\omega=\iota_{X_5^R}\omega^1\otimes e_1+\iota_{X_5^R}\omega^2\otimes e_2 = \d\left(v_1 + \frac{v_4^2}{2}\right)\otimes e_1+ \d\left(v_3 + \frac{v_4^2}{2}\right)\otimes e_2=\d \bh^R_5\,.
\end{equation}
Then, the reduced two-forms read
$$
\omega^1_{\bm\mu} = \d v_2\wedge\d v_4 + \d v_4\wedge\d v_6\,,\qquad \omega^2_{\bm\mu} = \d v_4\wedge\d v_6\,.
$$
Furthermore, one has
$$ \ker\omega^1_{\bm\mu} = \left\langle \parder{}{v_2} + \parder{}{v_6} \right\rangle\,,\qquad \ker\omega^2_{\bm\mu} = \left\langle \parder{}{v_2} \right\rangle\,, $$
and $\omega_{\bm\mu}^1,\omega_{\bm \mu}^2$ define a two-polysymplectic form on $P_{\bm\mu}$. Moreover, the $\bomega$-Hamiltonian function of $X^R_5$ is invariant relative to $Y$. Then, Theorem \ref{Th::Xreduction} ensures that the projection of $X^R_5$ onto $P_{\bm{\mu}}$ exists and is given by
$$
X^R=-v_4\frac{\partial}{\partial v_6},
$$
which is the $\bm{\omega}_{\bm \mu}$-Hamiltonian vector field of the $\bm{\omega}_{\bm \mu}$-Hamiltonian function
$$
{\bm f}_{\bm\mu}=\left(\mu^1 + \frac{v_4^2}{2}\right)\otimes e_1+ \left(\mu^2 + \frac{v_4^2}{2}\right)\otimes e_2\,,
$$
which has a critical point at every point $(v_4=0,v_6)$, where $v_6$ is arbitrary. Such points are not stable equilibrium points. In particular, this $\bm{\omega}_{\bm\mu}$-Hamiltonian function does not satisfy that $f^1_{\bm\mu}+f^2_{\bm \mu}$ has a strict minimum at the equilibrium point: it has only a minimum. 
Note that the points in ${\bf J}^{\Phi-1}(\bm\mu)$ that project onto the above-mentioned equilibrium points are two-polysymplectic relative equilibrium points.
The analysis of \eqref{Eq:HamFunTot} with our two-polysymplectic energy-momentum methods at the mentioned two-polysymplectic relative equilibrium points suggests the same results.

\subsection{Equilibrium points and vector fields with polynomial coefficients}

Let us illustrate certain aspects of our $k$-polysymplectic energy-momentum method by studying vector fields with a polynomial behaviour. Moreover, our example will illustrate some features of weakly regular points of $k$-polysymplectic momentum maps and the character of their associated $k$-polysymplectic Marsden--Weinstein reductions. 

Consider coordinates $\{x_1,x_2,x_3,x_4,x_5,x_6,x_7,x_8\}$ on $\mathbb{R}^8$ and  the vector field $X$ on $\mathbb{R}^8$ given by
$$
    X = x_6^a\parder{}{x_2}+x_4^b\parder{}{x_3}-x_3^c\parder{}{x_4}+x_8^d\parder{}{x_7}-x_7^e\parder{}{x_8}\,,
$$
where $a,b,c,d,e\in \mathbb{N}$. Define  the two-polysymplectic for $\bomega$ 
 on $\mathbb{R}^8$ of the form
 $$
    \bm\omega = \omega^1\otimes e_1 + \omega^2\otimes e_2 = (\d x_3\wedge \d x_4 + \d x_1\wedge \d x_5)\otimes e_1 + (\d x_2\wedge \d x_6 + \d x_7\wedge \d x_8)\otimes e_2\,.
$$
Then, 
$$
\ker\omega_x^1 = \left\langle \parder{}{x_2},\parder{}{x_6},\parder{}{x_7},\parder{}{x_8}\right\rangle,\qquad  \ker\omega^2_x = \left\langle \parder{}{x_1},\parder{}{x_3},\parder{}{x_4},\parder{}{x_5}\right\rangle,\qquad \ker\omega^1_x\cap\ker\omega^2_x=0
$$ for any $x\in \mathbb{R}^8$, and thus $\bm\omega$ becomes a two-polysymplectic form. 

The vector field $X$ admits the Lie symmetries,
$Y_1 = \parder{}{x_2}\,,\  Y_2=\parder{}{x_1}\,,\  Y_3=\parder{}{x_5}$, which span a three-dimensional abelian Lie algebra of vector fields. These Lie symmetries are the infinitesimal generators of the translations along the $x_2$, $x_1$, and $x_5$ coordinates, and they also leave the two-polysymplectic form invariant, i.e. $\Lie_{Y_i}\omega^\alpha=0$ for $i=1,2,3$ and $\alpha=1,2$. They give rise to a Lie group action $\Phi:\mathbb{R}^3\times \mathbb{R}^8\rightarrow \mathbb{R}^8$ that leaves invariant $\bm{\omega}$. 

Since
\begin{align}
    \iota_{Y_1}\omega^1 &= 0\,, &\iota_{Y_2}\omega^1 &= \d x_5\,, &\iota_{Y_3}\omega^1 &= -\d x_1\,,\\
    \iota_{Y_1}\omega^2 &= \d x_6\,, &\iota_{Y_2}\omega^2 &= 0\,,  &\iota_{Y_3}\omega^2 &= 0\,,
\end{align}
a two-polysymplectic momentum map $\mathbf{J}^\Phi$ can be defined by setting
$$
    \bfJ^\Phi:\mathbb{R}^8\ni x\longmapsto \mathbf{J}^\Phi(x)=(0, x_5,-x_1;x_6,0,0)\in (\mathbb{R}^{3*})^2\simeq (\mathbb{R}^3)^2\,.
$$
Then, $\T_x\mathbf{J}^{\Phi-1}(\bm\mu) = \left\langle \parder{}{x_2},\parder{}{x_3},\parder{}{x_4},\parder{}{x_7},\parder{}{x_8}\right\rangle$ for each $x\in\mathbf{J}^{\Phi-1}(\bm\mu)$ and $\bm\mu\in (\mathbb{R}^3)^2$. The two-polysymplectic momentum map $\mathbf{J}^\Phi$ is $\Ad^{*2}$-equivariant and every $\bm\mu\in (\mathbb{R}^3)^2$ is a weakly regular value of $\mathbf{J}^\Phi$ since each $\mathbf{J}^{\Phi-1}(\bm\mu)$ is a five-dimensional submanifold of $\mathbb{R}^8$ and its tangent space at each point coincides with the kernel of ${\bf J}^\Phi$ at that point. Moreover, ${\bf J}^\Phi$ has no regular points.

Note that $Y_2$ and $Y_3$ do not take values at $x$ in  $\T_x(G_{\bm\mu}x)$ but $Y_1$ does. The assumptions of Theorem \ref{Th::PolisymplecticReductionJ} are satisfied, and the quotient space $\T_x\mathbf{J}^{\Phi-1}(\bm\mu)/\T_{x}\left(G_{\bm \mu}x\right)$ is a two-dimensional subspace, where
\[
    \T_x\mathbf{J}^{\Phi-1}(\bm\mu)/\T_{x}\left(G_{\bm \mu}x\right) = \left\langle\parder{}{x_3},\parder{}{x_4},\parder{}{x_7},\parder{}{x_8}\right\rangle\,,\qquad \forall x\in {\bf J}^{\Phi-1}(\bm{\mu})
\]
and
\[
    \bm\omega_{\bm\mu} = \omega^1_{\bm\mu}\otimes e_1+\omega^2_{\bm\mu}\otimes e_2 = (\d x_3\wedge \d x_4)\otimes e_1 + (\d x_7\wedge \d x_8)\otimes e_2\,.
\]
The vector field $X$ is $\bm\omega$-Hamiltonian relative to
\begin{align*}
\d \bm h &= \iota_X\bomega = \iota_X\omega^1\otimes e_1+\iota_X\omega^2\otimes e_2\\
&= \d\left(\frac{1}{1+b}x^{b+1}_4+\frac{1}{c+1}x_3^{c+1} \right)\otimes e_1+\d\left(\frac{1}{1+a}x^{a+1}_6+\frac{1}{d+1}x_8^{d+1}+\frac{1}{1+e}x_7^{e+1} \right)\otimes e_2\,.
\end{align*}
Moreover, $\bm h$ is invariant relative to the Lie symmetries $Y_1,Y_2,Y_3$. 
By Theorem \ref{Th::Xreduction}, the vector field $X$ projects onto the quotient manifold and its projection $X_{\bm{\mu}}$ is given by
\[
X_{\bm\mu} = x_4^b\parder{}{x_3}-x_3^c\parder{}{x_4}+x_8^d\parder{}{x_7}-x_7^e\parder{}{x_8}\,,
\]
which is an $\bm\omega_{\bm\mu}$-Hamiltonian vector field since
\[
    \d\bm f_{\bm\mu} = \iota_{X_{\bm\mu}}\bm\omega_{\bm\mu}=\d\left(\frac{1}{1+b}x^{b+1}_4+\frac{1}{c+1}x_3^{c+1} \right)\otimes e_1+\d\left(\frac{1}{d+1}x_8^{d+1}+\frac{1}{1+e}x_7^{e+1} \right)\otimes e_2\,.
\]
According to Theorem \ref{Th::ksymEM}, a point $z_e$ is a two-polysymplectic relative equilibrium point if it is a critical point of $\boldsymbol{h}_\xi$ for some $\xi=(\xi_1,\xi_2,\xi_3)\in \mathfrak{g}\simeq\R^3$. Then, one has that
\begin{align*}
    \d\bm h_\xi &= \d f^1_\xi\otimes e_1+\d f^2_\xi \otimes e_2\\
    &=\left(x_4^b \d x_4+x^c_3 \d x_3-\xi_2\d x_5+\xi_3 \d x_1\right)\otimes e_1+\left( (x^a_6-\xi_1)\d x_6 + x_8^d\d x_8 + x_7^e\d x_7 \right)\otimes e_2\,.
\end{align*}
It follows that $\xi_2=\xi_3=0$ and we have two-polysymplectic relative equilibrium points of $X$ of the form $z_e=(x_1,x_2,0,0,x_5,x_6,0,0)$ for $x^a_6=\xi_1$. In fact, $x_1,x_2,x_5,x_6$ are arbitrary. Indeed, $(X_{\bm\mu_e})_{[z_e]}=0$ for $\bm{\mu}_e={\bf J}^\Phi(z_e)$.

To analyse the stability of the above-mentioned two-polysymplectic relative equilibrium points, let us analyse the second derivatives of $\boldsymbol{h}_\xi$ at such points $z_e$. Then,
\begin{align*}
    (\delta^2 \bm h_{\xi})_{z_e} &= (\delta^2 h^1_{\xi})_{z_e}\otimes e_1+ (\delta^2 h^2_{\xi})_{z_e}\otimes e_2\\
    &=\left( cx_3^{c-1}\d x_3\otimes \d x_3 + bx_4^{b-1}\d x_4\otimes \d x_4\right)\otimes e_1+\left(ex_7^{e-1}\d x_7\otimes \d x_7 + d x_8^{d-1} \d x_8\otimes \d x_8\right)\otimes e_2\,.
\end{align*}
Taking into account that the supplementary spaces to $\T_{z_e}G_{\bm\mu_e}+ \ker \omega^1_{z_e}\cap \T_{z_e}({\bf J}^{\Phi-1}(\bm \mu_e))$ and $\T_{z_e}G_{\bm\mu_e}+ \ker \omega^2_{z_e}\cap \T_{z_e}({\bf J}^{\Phi-1}(\bm \mu_e))$ can be chosen $\mathcal{S}_{z_e}^1=\langle \parder{}{x_3},\parder{}{x_4}\rangle $ and $\mathcal{S}^2_{z_e}=\langle \parder{}{x_7},\parder{}{x_8}\rangle$, respectively, Definition \ref{Prop::SecDerivEM} and the posterior explanation give that the $z_e$ are stable two-polysymplectic relative equilibrium points if
\[
(\delta^2h^1_\xi)_{z_e}(v_{z_e},v_{z_e}) > 0\,,\qquad \forall v_{z_e}\in \mathcal{S}^1_{z_e}\setminus\{0\}\,,
\]
and
\[
(\delta^2h^2_\xi)_{z_e}(v_{z_e},v_{z_e}) > 0\,,\qquad \forall v_{z_e}\in \mathcal{S}^2_{z_e}\setminus\{0\}\,.
\]
These inequalities hold if and only if $b,c,d,e=1$. Hence, the $z_e$ are formally stable two-polysymplectic relative equilibrium points of $X$ for $b,c,d,e=1$.  Indeed, it is immediate that previous conditions of stability imply that in the reduced space close to the equilibrium point, the coordinates $x_3,x_4,x_7,x_8$ are bounded for every motion close enough to the equilibrium point, which ensures real stability.

\section{Conclusions and Outlook}
\label{Sec::Conclusions}

In this work, we have devised a new energy-momentum method for systems of ordinary differential equations given by Hamiltonian vector fields with respect to a $k$-polysymplectic form. This led to defining and characterising $k$-polysymplectic relative equilibrium points and introducing new techniques to study stability through $k$-polysymplectic geometry. In this respect, we have also reviewed several aspects and mistakes in previous Marsden--Weinstein reductions for $k$-polysymplectic forms and systems \cite{Bla_19,GM_23, MRS_04,MRSV_15}, which is a relevant part in energy-momentum methods. To illustrate our new energy-momentum method and theoretical results, we have studied several relevant examples in detail: complex Schwarz equations, the product of several symplectic manifolds, with a family of particles subjected to the effect of an isotropic potential for each of them, affine homogeneous differential equations (with potential applications to control theory and Lie systems theory), a quantum harmonic oscillator, integrable symplectic systems, and some dynamical systems with polynomial coefficients.

A non-autonomous analogue of the methods devised in this paper could be accomplished by using the Lyapunov theory depicted in \cite{LZ_21}. Note that the stability with respect to $k$-polysymplectic forms is a topic that requires further development. The criteria here used are enough for the family of examples to be studied, but a deeper study with an analysis of all possibilities is in order. Recall also that, even in the one-polysymplectic case, the criterion for the stability of the energy-momentum method, which we here recover as a particular case and is classical \cite{AM_78}, is not a necessary condition for the stability of the reduced system.

The study of complex Schwarz equations and the Schwarzian derivative could be more appropriately studied through a complex Lie system formalism. We aim to study this possibility in further work.

\addcontentsline{toc}{section}{Acknowledgements}
\section*{Acknowledgements}

J. de Lucas, X. Rivas and B.M. Zawora acknowledge partial financial support from the Nowe Idee 2B-POB II project PSP: 501-D111-20-2004310 funded by the ``Inicjatywa Doskonałości - Uczelnia Badawcza'' (IDUB) program. J. de Lucas acknowledges a CRM-Simons professorship funded by the Simons Foundation and the Centre de Recherches Math\'ematiques (CRM) of the Universit\'e Montr\'eal. X. Rivas acknowledges funding from J. de Lucas's CRM-Simons professorship to accomplish a stay at the CRM which allowed us to obtain several relevant findings of the present work. X. Rivas also acknowledges funding from the Spanish Ministry of Science and Innovation, grants  PID2021-125515NB-C21, and RED2022-134301-T of AEI, and Ministry of Research and Universities of the Catalan Government, project 2021 SGR 00603. B.M. Zawora acknowledges funding from the IDUB mikrogrant program to accomplish a research stay at the CRM. L. Colombo acknowledges financial support from Grant PID2022-137909NB-C21 funded by MCIN/AEI/ 10.13039/501100011033. B.M. Zawora and X. Rivas acknowledge partial funding from the Simons Foundation to take part in a conference in the Simons Center for Geometry and Physics of Stony Brook University, where the last parts of this paper were accomplished and several relevant results concerning the stability of reduced $k$-polysymplectic Hamiltonian systems and their applications were devised.



\bibliographystyle{abbrv}
\bibliography{references.bib}

\begin{thebibliography}{10}

\bibitem{AH_87}
H.~Abarbanel and D.~Holm.
\newblock {Nonlinear stability analysis of inviscid flows in three dimensions: incompressible fluids and barotropic fluids}.
\newblock {\em Phys. Fluids}, {\bf 30}(11):3369–3382, 1987.
\newblock \href{https://doi.org/10.1063/1.866469}{10.1063/1.866469}.

\bibitem{AM_78}
R.~Abraham and J.~E. Marsden.
\newblock {\em {Foundations of mechanics}}, volume 364 of {\em AMS Chelsea publishing}.
\newblock Benjamin/Cummings Pub. Co., New York, 2nd edition, 1978.
\newblock \href{https://doi.org/10.1090/chel/364}{10.1090/chel/364}.

\bibitem{Alb_89}
C.~Albert.
\newblock {Le théorème de réduction de Marsden--Weinstein en géométrie cosymplectique et de contact}.
\newblock {\em J. Geom. Phys.}, {\bf 6}(4):627–649, 1989.
\newblock \href{https://doi.org/10.1016/0393-0440(89)90029-6}{10.1016/0393-0440(89)90029-6}.

\bibitem{Awa_92}
A.~Awane.
\newblock {$k$-symplectic structures}.
\newblock {\em J. Math. Phys.}, {\bf 33}(12):4046, 1992.
\newblock \href{https://doi.org/10.1063/1.529855}{10.1063/1.529855}.

\bibitem{AG_00}
A.~Awane and M.~Goze.
\newblock {\em {Pfaffian systems, $k$-symplectic systems}}.
\newblock Springer, Dordrecht, 1st edition, 2000.
\newblock \href{https://doi.org/10.1007/978-94-015-9526-1}{10.1007/978-94-015-9526-1}.

\bibitem{BZ_14}
C.~Bai and H.~Zhang.
\newblock {A weak energy-momentum method for stochastic instability induced by dissipation and random excitations}.
\newblock {\em Probabilistic Eng. Mech.}, {\bf 37}:35--40, 2014.
\newblock \href{https://doi.org/10.1016/j.probengmech.2014.03.009}{10.1016/j.probengmech.2014.03.009}.

\bibitem{Bla_19}
C.~Blacker.
\newblock {Polysymplectic reduction and the moduli space of flat connections}.
\newblock {\em J. Phys. A: Math. Theor.}, {\bf 52}(33):335201, 2019.
\newblock \href{https://doi.org/10.1088/1751-8121/ab2eed}{10.1088/1751-8121/ab2eed}.

\bibitem{Bla_09}
A.~M. Blaga.
\newblock {Connections on $k$-symplectic manifolds}.
\newblock {\em Balk. J. Geom. Appl.}, {\bf 14}(2):28–33, 2009.
\newblock \url{https://www.emis.de/journals/BJGA/v14n2/B14-2-bl.pdf}.

\bibitem{BR_04}
G.~Blankenstein and T.~S. Ratiu.
\newblock {Singular reduction of implicit Hamiltonian systems}.
\newblock {\em Rep. Math. Phys.}, {\bf 53}(2):211--260, 2004.
\newblock \href{https://doi.org/10.1016/S0034-4877(04)90013-4}{10.1016/S0034-4877(04)90013-4}.

\bibitem{Blo_15}
A.~M. Bloch.
\newblock {\em {Nonholonomic mechanics and control}}, volume~24 of {\em Interdisciplinary Applied Mathematics}.
\newblock Springer, New York, 2003.
\newblock \href{https://doi.org/10.1007/978-1-4939-3017-3}{10.1007/978-1-4939-3017-3}.

\bibitem{BC_20}
Y.~D. Bozhkov and P.~R. da~Conceição.
\newblock {On the Generalizations of the Kummer--Schwarz Equation}.
\newblock {\em Nonlinear Anal.}, {\bf 192}:111691, 2020.
\newblock \href{https://doi.org/10.1016/j.na.2019.111691}{10.1016/j.na.2019.111691}.

\bibitem{BBLSV_15}
L.~Búa, I.~Bucataru, M.~de~León, M.~Salgado, and S.~Vilariño.
\newblock {Symmetries in Lagrangian field theory}.
\newblock {\em Rep. Math. Phys.}, {\bf 75}(3):333–357, 2015.
\newblock \href{https://doi.org/10.1016/S0034-4877(15)30010-0}{10.1016/S0034-4877(15)30010-0}.

\bibitem{CL_11}
J.~F. Cariñena and J.~de~Lucas.
\newblock {Lie systems: theory, generalisations, and applications}.
\newblock {\em Dissertationes Math.}, {\bf 479}:1–162, 2011.
\newblock \href{https://doi.org/10.4064/dm479-0-1}{10.4064/dm479-0-1}.

\bibitem{CGM_00}
J.~F. Cariñena, J.~Grabowski, and G.~Marmo.
\newblock {\em {Lie--Scheffers systems: a geometric approach}}.
\newblock Napoli Series on Physics and Astrophysics. Bibliopolis, Naples, 2000.

\bibitem{CGM_07}
J.~F. Cariñena, J.~Grabowski, and G.~Marmo.
\newblock Superposition rules, lie theorem, and partial differential equations.
\newblock {\em Rep. Math. Phys.}, {\bf 60}(2):237–258, 2007.
\newblock \href{https://doi.org/10.1016/S0034-4877(07)80137-6}{10.1016/S0034-4877(07)80137-6}.

\bibitem{CR_03}
J.~F. Cariñena and A.~Ramos.
\newblock {Applications of Lie systems in quantum mechanics and control theory}.
\newblock {\em Banach Center Publ.}, {\bf 59}:143–162, 2003.
\newblock \href{https://doi.org/10.4064/bc59-0-7}{10.4064/bc59-0-7}.

\bibitem{LMOS_97}
M.~de~León, E.~Merino, J.~A. Oubiña, and M.~Salgado.
\newblock {Stable almost cotangent structures}.
\newblock {\em Bolletino Unione Mat. Ital. B (7)}, {\bf 11}(3):509–529, 1997.

\bibitem{LMS_88}
M.~de~León, I.~Méndez, and M.~Salgado.
\newblock {$p$-almost tangent structures}.
\newblock {\em Rend. Circ. Mat. Palermo}, {\bf 37}(2):282–294, 1988.
\newblock \href{https://doi.org/10.1007/BF02844526}{10.1007/BF02844526}.

\bibitem{LMS_88a}
M.~de~León, I.~Méndez, and M.~Salgado.
\newblock {Regular $p$-almost cotangent structures}.
\newblock {\em {J. Korean Math. Soc.}}, {\bf 25}(2):273–287, 1988.
\newblock \url{https://jkms.kms.or.kr/journal/view.html?spage=273&volume=25&number=2}.

\bibitem{LSV_15}
M.~de~León, M.~Salgado, and S.~Vilariño.
\newblock {\em {Methods of Differential Geometry in Classical Field Theories}}.
\newblock World Scientific, 2015.
\newblock \href{https://doi.org/10.1142/9693}{10.1142/9693}.

\bibitem{LMZ_23}
J.~de~Lucas, A.~Maskalaniec, and B.~M. Zawora.
\newblock {A cosymplectic energy-momentum method with applications}.
\newblock \href{https://arxiv.org/abs/2302.05827}{2302.05827}, 2023.

\bibitem{LRVZ_23}
J.~de~Lucas, X.~Rivas, S.~Vilariño, and B.~M. Zawora.
\newblock {On $k$-polycosymplectic Marsden--Weinstein reductions}.
\newblock {\em J. Geom. Phys.}, {\bf 191}:104899, 2023.
\newblock \href{https://doi.org/10.1016/j.geomphys.2023.104899}{10.1016/j.geomphys.2023.104899}.

\bibitem{LS_20}
J.~de~Lucas and C.~Sardón.
\newblock {\em {A Guide to Lie Systems with Compatible Geometric Structures}}.
\newblock World Scientific Publishing Co. Pte. Ltd., Singapore, 2020.
\newblock \href{https://doi.org/10.1142/q0208}{10.1142/q0208}.

\bibitem{LV_15}
J.~de~Lucas and S.~Vilariño.
\newblock {$k$-symplectic Lie systems: theory and applications}.
\newblock {\em J. Differ. Equ.}, {\bf 258}(6):2221–2255, 2015.
\newblock \href{https://doi.org/10.1016/j.jde.2014.12.005}{10.1016/j.jde.2014.12.005}.

\bibitem{LZ_21}
J.~de~Lucas and B.~M. Zawora.
\newblock {A time-dependent energy-momentum method}.
\newblock {\em J. Geom. Phys.}, {\bf 170}:104364, 2021.
\newblock \href{https://doi.org/10.1016/j.geomphys.2021.104364}{10.1016/j.geomphys.2021.104364}.

\bibitem{EMR_96}
A.~Echeverría-Enríquez, M.~C. Muñoz-Lecanda, and N.~Román-Roy.
\newblock {Geometry of Lagrangian first-order classical field theories}.
\newblock {\em Fortschr. Phys.}, {\bf 44}(3):235–280, 1996.
\newblock \href{https://doi.org/10.1002/prop.2190440304}{10.1002/prop.2190440304}.

\bibitem{GM_23}
E.~{García-Toraño Andrés} and T.~Mestdag.
\newblock {Conditions for symmetry reduction of polysymplectic and polycosymplectic structures}.
\newblock {\em J. Phys. A: Math. Theor.}, {\bf 56}(33):335202, 2023.
\newblock \href{https://doi.org/10.1088/1751-8121/ace74c}{10.1088/1751-8121/ace74c}.

\bibitem{GMS_97}
G.~Giachetta, L.~Mangiarotti, and G.~A. Sardanashvily.
\newblock {\em {New Lagrangian and Hamiltonian Methods in Field Theory}}.
\newblock World Scientific, River Edge, 1997.
\newblock \href{https://doi.org/10.1142/2199}{10.1142/2199}.

\bibitem{GLMV_19}
X.~Gràcia, J.~de~Lucas, M.~C. Muñoz-Lecanda, and S.~Vilariño.
\newblock {Multisymplectic structures and invariant tensors for Lie systems}.
\newblock {\em J. Phys. A: Math. Theor.}, {\bf 52}(21):215201, 2019.
\newblock \href{https://doi.org/10.1088/1751-8121/ab15f2}{10.1088/1751-8121/ab15f2}.

\bibitem{GR_07}
L.~Guieu and C.~Roger.
\newblock {\em {L'algèbre et le Groupe de Virasoro. Aspects géométriques et algébriques, généralisations}}.
\newblock Les Publications CRM, Montréal, 2007.

\bibitem{Gun_87}
C.~Günther.
\newblock {The polysymplectic Hamiltonian formalism in field theory and calculus of variations I: The local case}.
\newblock {\em J. Differ. Geom.}, {\bf 25}(1):23–53, 1987.
\newblock \href{https://doi.org/10.4310/jdg/1214440723}{10.4310/jdg/1214440723}.

\bibitem{Hil_97}
E.~Hille.
\newblock {\em {Ordinary differential equations in the complex domain}}.
\newblock Dover Publications, Mineola, NY, 1997.

\bibitem{Kan_98}
I.~V. Kanatchikov.
\newblock {Canonical structure of classical field theory in the polymomentum phase space}.
\newblock {\em Rep. Math. Phys.}, {\bf 41}(1):49–90, 1998.
\newblock \href{https://doi.org/10.1016/S0034-4877(98)80182-1}{10.1016/S0034-4877(98)80182-1}.

\bibitem{Lee_09}
J.~M. Lee.
\newblock {\em {Manifolds and differential geometry}}, volume 107 of {\em Graduate Studies in Mathematics}.
\newblock American Mathematical Society, Providence, RI, 2009.

\bibitem{Lee_12}
J.~M. Lee.
\newblock {\em {Introduction to Smooth Manifolds}}, volume 218 of {\em Graduate Texts in Mathematics}.
\newblock Springer New York Heidelberg Dordrecht London, 2nd edition, 2012.
\newblock \href{http://doi.org/10.1007/978-1-4419-9982-5}{10.1007/978-1-4419-9982-5}.

\bibitem{Leh_79}
O.~Lehto.
\newblock {Remarks on Nehari's theorem about the Schwarzian derivative and schlicht functions}.
\newblock {\em J. Anal. Math.}, {\bf 36}:184–190, 1979.
\newblock \href{https://doi.org/10.1007/BF02798778}{10.1007/BF02798778}.

\bibitem{MRSV_15}
J.~C. Marrero, N.~Román-Roy, M.~Salgado, and S.~Vilariño.
\newblock {Reduction of polysymplectic manifolds}.
\newblock {\em J. Phys. A: Math. Theor.}, {\bf 48}(5):055206, 2015.
\newblock \href{http://doi.org/10.1088/1751-8113/48/5/055206}{10.1088/1751-8113/48/5/055206}.

\bibitem{MRSV_10}
J.~C. Marrero, N.~Román‐Roy, M.~Salgado, and S.~Vilariño.
\newblock {Symmetries, Noether's theorem and reduction in k‐cosymplectic field theories}.
\newblock {\em AIP Conference Proceedings}, {\bf 1260}(1):173--179, 2010.
\newblock \href{https://doi.org/10.1063/1.3479319}{10.1063/1.3479319}.

\bibitem{MPS_90}
J.~E. Marsden, T.~A. Posbergh, and J.~C. Simo.
\newblock {Stability of coupled rigid body and geometrically exact rods: block diagonalization and the energy–momentum method}.
\newblock {\em Phys. Rep.}, {\bf 193}:280–360, 1990.
\newblock \href{https://doi.org/10.1016/0370-1573(90)90125-L}{10.1016/0370-1573(90)90125-L}.

\bibitem{MR_99}
J.~E. Marsden and T.~Ratiu.
\newblock {\em {Introduction to Mechanics and Symmetry. A basic exposition of classical mechanical systems}}, volume~17 of {\em Texts in Applied Mathematics}.
\newblock Springer-Verlag, New York, 1999.
\newblock \href{https://doi.org/10.1007/978-0-387-21792-5}{10.1007/978-0-387-21792-5}.

\bibitem{MS_88}
J.~E. Marsden and J.~C. Simo.
\newblock {The energy momentum method}.
\newblock {\em Act. Acad. Sci. Tau.}, {\bf 1}(124):245–268, 1988.
\newblock \href{https://resolver.caltech.edu/CaltechAUTHORS:20101019-093814651}{https://resolver.caltech.edu/CaltechAUTHORS:20101019-093814651}.

\bibitem{MSLP_89}
J.~E. Marsden, J.~C. Simo, D.~Lewis, and T.~Posbergh.
\newblock {A block diagonalization theorem in the energy-momentum method}.
\newblock In {\em Dynamics and control of multibody systems}, volume~97 of {\em Contemp. Math.}, page 297–313. Amer. Math. Soc., Providence, RI, 1989.
\newblock \href{https://doi.org/10.1090/conm/097/1021043}{10.1090/conm/097/1021043}.

\bibitem{Ni15}
N.~Martinez.
\newblock Poly-symplectic groupoids and poly-{P}oisson structures.
\newblock {\em Lett. Math. Phys.}, 105(5):693--721, 2015.

\bibitem{MN_00}
M.~McLean and L.~K. Norris.
\newblock {Covariant field theory on frame bundles of fibered manifolds}.
\newblock {\em J. Math. Phys.}, {\bf 41}(10):6808–6823, 2000.
\newblock \href{https://doi.org/10.1063/1.1288797}{10.1063/1.1288797}.

\bibitem{Mer_97}
E.~E. Merino.
\newblock {\em {Geometría $k$-simpléctica y $k$-cosimpléctica. Aplicaciones a las teorías clásicas de campos}}.
\newblock PhD thesis, Universidade de Santiago de Compostela, 1997.

\bibitem{MRS_04}
F.~Munteanu, A.~M. Rey, and M.~Salgado.
\newblock {The Günther's formalism in classical field theory: momentum map and reduction}.
\newblock {\em J. Math. Phys.}, {\bf 45}(5):1730--1751, 2004.
\newblock \href{https://doi.org/10.1063/1.1688433}{10.1063/1.1688433}.

\bibitem{Nor_93}
L.~K. Norris.
\newblock {Generalized symplectic geometry on the frame bundle of a manifold}.
\newblock In {\em Proc. Symp. Pure Math.}, volume 54.2, page 435–465. Amer. Math. Soc., Providence RI, 1993.
\newblock \href{https://doi.org/10.1090/pspum/054.2}{10.1090/pspum/054.2}.

\bibitem{OPR_05}
J.~Ortega, V.~Planas-Bielsa, and T.~Ratiu.
\newblock {Asymptotic and Lyapunov stability of constrained and Poisson equilibria}.
\newblock {\em J. Differ. Equ.}, {\bf 214}(1):92–127, 2005.
\newblock \href{https://doi.org/10.1016/j.jde.2004.09.016}{10.1016/j.jde.2004.09.016}.

\bibitem{OR_04}
J.~P. Ortega and T.~S. Ratiu.
\newblock {\em {Momentum maps and Hamiltonian reduction}}, volume 222 of {\em Progress in Mathematics}.
\newblock Birkh\"auser Boston, Inc., Boston, 2004.
\newblock \href{https://doi.org/10.1007/978-1-4757-3811-7}{10.1007/978-1-4757-3811-7}.

\bibitem{RSV_07}
N.~Román-Roy, M.~Salgado, and S.~Vilariño.
\newblock {Symmetries and conservation laws in the Günther $k$-symplectic formalism of field theory}.
\newblock {\em Rev. Math. Phys.}, {\bf 19}(10):1117–1147, 2007.
\newblock \href{https://doi.org/10.1142/S0129055X07003188}{10.1142/S0129055X07003188}.

\bibitem{Sar_95}
G.~A. Sardanashvily.
\newblock {\em {Generalized Hamiltonian formalism for field theory. Constraint systems}}.
\newblock World Scientific Publishing Company, River Edge, NJ, 1995.
\newblock \href{https://doi.org/10.1142/2550}{10.1142/2550}.

\bibitem{SLM_91}
J.~Simo, D.~Lewis, and J.~Marsden.
\newblock {Stability of relative equilibria. Part I: The reduced energy-momentum method}.
\newblock {\em Arch. Rat. Mech. Anal.}, {\bf 115}:15–59, 1991.
\newblock \href{https://doi.org/10.1007/BF01881678}{10.1007/BF01881678}.

\bibitem{ST_92}
J.~Simo and N.~Tarnow.
\newblock {The discrete energy-momentum method. Conserving algorithms for nonlinear elastodynamics}.
\newblock {\em Z. Angew. Math. Phys.}, {\bf 43}:757–792, 1992.
\newblock \href{https://doi.org/10.1007/BF00913408}{10.1007/BF00913408}.

\bibitem{Vid_02}
M.~Vidyasagar.
\newblock {\em {Nonlinear systems analysis}}, volume~42 of {\em Classics in Applied Mathematics}.
\newblock Society for Industrial and Applied Mathematics (SIAM), Philadelphia, PA, 2002.
\newblock \href{https://doi.org/10.1137/1.9780898719185}{10.1137/1.9780898719185}.

\bibitem{WK_92}
L.~Wang and P.~Krishnaprasad.
\newblock {Gyroscopic control and stabilization}.
\newblock {\em J. Nonlinear Sci.}, {\bf 2}:367–415, 1992.
\newblock \href{https://doi.org/10.1007/BF01209527}{10.1007/BF01209527}.

\bibitem{WN_63}
J.~Wei and E.~Norman.
\newblock {Lie Algebraic Solution of Linear Differential Equations}.
\newblock {\em J. Math. Phys.}, {\bf 4}(4):575–581, 1963.
\newblock \href{https://doi.org/10.1063/1.1703993}{10.1063/1.1703993}.

\bibitem{Win_83}
P.~Winternitz.
\newblock {Lie groups and solutions of nonlinear differential equations}.
\newblock In K.~B. Wolf, editor, {\em Nonlinear Phenomena}, volume 189 of {\em Lecture Notes in Physics}, page 263–305. Springer, Berlin, Heidelberg, 1983.
\newblock \href{https://doi.org/10.1007/3-540-12730-5_12}{10.1007/3-540-12730-5\_12}.

\bibitem{Zaw_21}
B.~M. Zawora.
\newblock A time-dependent energy-momentum method.
\newblock Master's thesis, University of Warsaw, Faculty of Physics, 2021.

\bibitem{ZBM_98}
D.~V. Zenkov, A.~M. Bloch, and J.~E. Marsden.
\newblock {The energy-momentum method for the stability of non-holonomic systems}.
\newblock {\em Dyn. Stab. Syst.}, {\bf 13}(2):123–165, 1998.
\newblock \href{https://doi.org/10.1080/02681119808806257}{10.1080/02681119808806257}.

\end{thebibliography}

\end{document}